\documentclass[a4paper,11pt]{article}
\usepackage{jheppub} 
\usepackage{lineno}
\usepackage[dvipsnames]{xcolor} 
\usepackage{dirtytalk}
\usepackage{ytableau}
\usepackage[dvipsnames]{xcolor}
\usepackage{dsfont}
\usepackage{mathdots}
\usepackage{colortbl}
\usepackage{braket}
\usepackage{amsthm}
\usepackage{amsmath}
\usepackage{float}
\usepackage{bbm, dsfont}
\usepackage[normalem]{ulem}
\usepackage{slashed}
\usepackage{mathtools}
\usepackage{multirow}
\usepackage{cancel}
\usepackage{cleveref}
\usepackage{orcidlink}
\usepackage{enumerate}
\usepackage{fontawesome}
\usepackage{adjustbox}
\usepackage{tcolorbox}
\usepackage{relsize}
\usepackage{csquotes}

\usepackage{tikz}
    \usetikzlibrary{positioning}
    \usetikzlibrary{arrows}
    \usetikzlibrary{shapes}
    \usepackage{pgfplots}
    \usetikzlibrary{calc}
    \usetikzlibrary{decorations.markings}
     \usetikzlibrary{decorations.pathmorphing}
    \tikzset{snake it/.style={decorate, decoration=snake}}

    \usetikzlibrary{matrix}

    \pgfplotsset{compat=1.18}

\definecolor{ourblue}{RGB}{0, 57, 120} %
\colorlet{ourlightblue}{ourblue!50!white}
\newcommand{\bs}[1]{\boldsymbol{#1}}
\newcommand{\eps}{\varepsilon}
\newcommand{\rd}{\mathrm{d}}

\newcommand{\bx}{\bs{x}}

\newcommand{\Phit}{\tilde{\Phi}}

\newcommand{\bz}{\bs{z}}
\newcommand{\bA}{\bs{A}}

\newcommand{\bI}{\bs{I}}
\newcommand{\bJ}{\bs{J}}
\newcommand{\bP}{\bs{P}}
\newcommand{\bT}{\bs{T}}

\newcommand{\bO}{\bs{O}}

\newcommand{\bH}{\bs{H}}
\newcommand{\bC}{\bs{C}}
\newcommand{\bK}{\bs{K}}
\newcommand{\bX}{\bs{X}}
\newcommand{\bM}{\bs{M}}
\newcommand{\bN}{\bs{N}}
\newcommand{\bR}{\bs{R}}
\newcommand{\bU}{\bs{U}}

\newcommand{\Np}{\text{NT}^{{\Delta}+}}
\newcommand{\NT}{\text{NT}^{{\Delta}}}
\newcommand{\Nm}{\text{NT}^{{\Delta}-}}
\newcommand{\UT}{\mathrm{UT}^{{\Delta}}}
\newcommand{\OT}{\text{OT}^{{\Delta}}}
\newcommand{\SUT}{\text{ST}^{{\Delta}}}
\newcommand{\bDelta}{\bs{\Delta}}

\newcommand{\bOmega}{\bs{\Omega}}

\newcommand{\lb}{{\scriptscriptstyle (\overline{3F2})}}
\newcommand{\lbf}{{\scriptscriptstyle (4F3)}}
\newcommand{\lbt}{{\scriptscriptstyle (3F2)}}
\newcommand{\bUt}{\bs{U}_{\!\mathrm{t}}}
\newcommand{\bUss}{\bs{U}_{\!\mathrm{ss}}}
\newcommand{\bUeps}{\bs{U}_{\!\eps}}
\newcommand{\cFss}{\mathcal{F}_{\!\mathrm{ss}}}
\newcommand{\cFeps}{\mathcal{F}_{\!\eps}}
\def\beq{\begin{equation}}
\def\eeq{\end{equation}}
\def\bsp#1\esp{\begin{split}#1\end{split}}
\newcommand{\cF}{\mathcal{F}}
\newcommand{\Gpar}{G_{\mathrm{par}}}
\newcommand{\bD}{\bs{D}}
\newtheorem{lemma}{Lemma}
\newtheorem{observation}{Observation}
\newtheorem{proposition}{Proposition}

\newcommand{\lodd}{{\scriptscriptstyle(\mathrm{odd})}}
\newcommand{\leven}{{\scriptscriptstyle(\mathrm{even})}}

\title{\boldmath Draft Template}

\author{Claude Duhr${}^1$, Sara Maggio${}^1$, Franziska Porkert${}^1$, Cathrin Semper${}^1$, Yoann Sohnle${}^2$, Sven F. Stawinski${}^1$}
\affiliation{
\vskip 0.5 em
${}^1$Bethe Center for Theoretical Physics, Universit\"at Bonn, D-53115, Germany\\
${}^2${\it Department of Physics and Astronomy, Uppsala University, Box 516, 75120 Uppsala, Sweden}
\vskip 0.5 em}

\emailAdd{cduhr@uni-bonn.de, smaggio@uni-bonn.de, fporkert@uni-bonn.de, csemper@uni-bonn.de, yoann.sohnle@physics.uu.se,
sstawins@uni-bonn.de}

\title{Canonical differential equations and intersection matrices}

\abstract{
Differential equations are one of the main approaches to evaluate multi-loop Feynman integrals. 
The 
construction of a canonical or $\eps$-factorised basis for multi-loop integrals remains a key bottleneck for this approach, especially when the differential equation involves non dlog-forms. Recently, several methods have been proposed to find $\eps$-factorised differential equations. Many of them  introduce new functions that are themselves defined as iterated integrals. If and when these iterated integrals can be explicitly evaluated in terms of other classes of functions remains an open problem. In this paper we elaborate on the recent proposal that one can use the fact that the intersection matrix computed in a canonical basis can be used to derive polynomial relations between these iterated integrals. On the one hand, we discuss properties of the canonical intersection matrix, in particular methods to determine the intersection matrix in a canonical basis. On the other hand we show how one can reduce the non-linear constraints on the iterated integrals to linear ones. We illustrate these ideas on examples involving Calabi–Yau varieties and higher-genus Riemann surfaces. 
}

\begin{document}
\begin{flushright}
    BONN-TH/2025-30\\
    UUITP–27/25
    \end{flushright}

\maketitle
\flushbottom

\section{Introduction}
\label{sec:intro}

Feynman integrals are a cornerstone of almost all perturbative computations to make predictions for collider and gravitational wave experiments. Despite their importance, the explicit evaluation of multi-loop Feynman integrals is still a bottleneck. Over the last decades, various methods for their computation have been developed. The arguably most widely used approach is integration-by-parts (IBP) reduction~\cite{Tkachov:1981wb,Chetyrkin:1981qh} followed by the differential equations technique~\cite{Kotikov:1990kg,Kotikov:1991hm,Kotikov:1991pm,Gehrmann:1999as}, augmented by the idea of a canonical or $\eps$-factorised basis~\cite{Henn:2013pwa}, where the system of differential equations depends only linearly on the dimensional regularisation parameter $\eps$. This approach, which was inspired by the concepts of pure functions~\cite{Arkani-Hamed:2010pyv} and uniform transcendental weight~\cite{Kotikov:2010gf} in planar $\mathcal{N}=4$ Super Yang-Mills theory, was particularly successful in cases where the canonical differential equation matrix only involves dlog-forms. Many state-of-the-art results for multi-loop integrals have been obtained in this way.

Cases where it is not possible to obtain a canonical differential equation involving only dlog-forms are much less understood. These instances are tightly related to situations where the Feynman integrals are connected to non-trivial geometries, like families of elliptic curves, Riemann surfaces of higher genus, Calabi-Yau varieties or Fano varieties (for a characterisation of geometries that arise from certain classes of Feynman integrals, see, e.g., refs.~\cite{Bourjaily:2019hmc,Bourjaily:2018ycu,Bourjaily:2018yfy,Doran:2023yzu,Frellesvig:2023bbf,Brammer:2025rqo,Frellesvig:2024zph}). Various proposals have been made for how to extend the concepts of canonical and/or $\eps$-factorised bases to these situations~\cite{Primo:2017ipr,Adams:2018yfj,Broedel:2018rwm,Bogner:2019lfa,Frellesvig:2021hkr,Chen:2022lzr,Dlapa:2022wdu,Pogel:2022vat,Pogel:2022ken,Pogel:2022yat,Gorges:2023zgv,Chen:2025hzq,Duhr:2025lbz,Maggio:2025jel,e-collaboration:2025frv} (for an approach without $\eps$-factorisation, see ref.~\cite{Chaubey:2025adn}). While there is still no general consensus on what a good definition of a canonical basis beyond dlog-forms is, there is a growing body of evidence that bases that deem to be called `canonical' possess properties that go beyond the simple factorisation of the dimensional regulator $\eps$. Throughout this paper we focus on the definition of a canonical basis recently proposed in ref.~\cite{Duhr:2025lbz}, and which is itself based on the method for finding $\eps$-factorised bases of ref.~\cite{Gorges:2023zgv} (see also refs.~\cite{Pogel:2022vat,Pogel:2022ken,Pogel:2022yat,Maggio:2025jel,e-collaboration:2025frv}, which we expect to produce bases that are equivalent to those produced by the method of ref.~\cite{Gorges:2023zgv}). This method is inspired by the extension of the concepts of pure functions and uniform transcendental weight beyond dlog-forms  introduced in ref.~\cite{Broedel:2018qkq}. In a nutshell, the method of ref.~\cite{Gorges:2023zgv} proceeds as follows: after a suitable starting basis inspired by the geometry underlying the problem has been identified~\cite{Duhr:2025lbz}, a sequence of rotations is constructed that brings the original system of differential equations into an $\eps$-factorised form. In ref.~\cite{Duhr:2025lbz} is was oberserved that in all known examples the resulting system exhibits additional properties which go beyond mere $\eps$-factorisation. It was proposed that these additional properties, in conjunction with the factorisation of the dimensional regulator $\eps$, provide an appropriate working definition of canonical bases beyond the dlog scenario. 

Given that a coherent picture of a definition of a canonical basis beyond dlog-forms is starting to emerge, it is time to turn the problem around and to ask what are the additional properties and features one may expect the canonical basis integrals to satisfy. Since one desirable property seems to be the factorisation of the dimensional regulator $\eps$, it is natural to expect those properties to be tightly connected to  features introduced by dimensional regularisation~\cite{tHooft:1972tcz,Cicuta:1972jf,Bollini:1972ui}. A rigorous and axiomatic mathematical framework to define dimensionally regulated Feynman integrals was introduced a long time ago~\cite{Wilson:1972cf}, but it was realised only recently that the appropriate mathematical setting is twisted cohomology theory~\cite{Mastrolia:2018uzb} (cf., e.g., refs.~\cite{yoshida_hypergeometric_1997,aomoto_theory_2011,matsumoto_relative_2019-1} for a mathematical introduction). It is thus natural to expect that mathematical properties of canonical bases must be rooted in twisted cohomology, and so the study of the properties of canonical systems of differential equations should go hand in hand with the study of the relevant twisted cohomology groups. So far most of the applications of twisted cohomology to Feynman integrals have been in the context of finding and studying linear relations between integrals using the intersection pairing in cohomology~\cite{Frellesvig:2017aai,Frellesvig:2020qot,Frellesvig:2019kgj,Weinzierl:2020xyy,Weinzierl:2020nhw,Cacciatori:2021nli,Chestnov:2022xsy,Brunello:2023rpq,Frellesvig:2019uqt,Fontana:2023amt,Brunello:2024tqf,Crisanti:2024onv,Lu:2024dsb}. Fewer attempts have been made to apply results from twisted cohomology in the context of canonical bases~\cite{Caron-Huot:2021xqj,Caron-Huot:2021iev,Chen:2022lzr,Giroux:2022wav,Chen:2023kgw,Gasparotto:2023roh,Giroux:2024yxu,Chen:2024ovh,Duhr:2024xsy,e-collaboration:2025frv} (for other applications of twisted cohomology inspired by Feynman integrals, see, e.g., refs.~\cite{Mizera:2019vvs,Abreu:2019wzk,Britto:2021prf,Cacciatori:2022mbi,Duhr:2023bku,Bhardwaj:2023vvm,Duhr:2024rxe,Angius:2025drr,Pokraka:2025zlh}).

This paper is part of an effort to investigate properties of canonical bases through the lens of twisted cohomology theories. It builds on the result of ref.~\cite{Duhr:2024xsy}, where it was observed that, as a consequence of the twisted Riemann bilinear relations satisfied by Feynman integrals~\cite{Lee:2018jsw,Duhr:2024rxe,Cho_Matsumoto_1995} (see also refs.\cite{Broadhurst:2018tey,Bonisch:2021yfw}), the intersection matrix computed between elements of a canonical basis take a very simple form that is independent of the kinematics. 
In this paper we extend the study of the connection between canonical bases in the following direction.
The sequence of rotations to the canonical basis is constructed via solutions of first-order differential equations. As a consequence, the differential equation matrices obtained via the method of refs.~\cite{Gorges:2023zgv,Duhr:2025lbz} (but also those of refs.~\cite{Pogel:2022vat,Pogel:2022ken,Pogel:2022yat,Maggio:2025jel,e-collaboration:2025frv} expected to deliver equivalent bases) involve (iterated) integrals over algebraic functions and (derivatives of) periods of the underlying geometry. If and when these integrals -- dubbed \emph{$\eps$-functions} in ref.~\cite{Duhr:2025kkq} -- evaluate to simpler, known quantities, and what sets them apart from the iterated integrals that define the pure functions to which the canonical integrals evaluate, is currently an open question that lies at the heart of understanding the properties of canonical bases. In ref.~\cite{Duhr:2024uid} it was observed that (at least on the maximal cut) the constancy of the intersection matrix for canonical bases can be used to derive polynomial relations between $\eps$-functions (see also ref.~\cite{Pogel:2024sdi} for a related result). This allows one to express some of the integrals that define $\eps$-functions in terms of known functions, like algebraic functions and (derivatives of) the periods that 
define the underlying geometry.

The goal of this paper is to understand more systematically, properties of the intersection matrix computed in a canonical basis on the one hand, and how one can identify the $\eps$-functions that can be evaluated in terms of algebraic functions and (derivatives of) periods on the other. At the heart of our paper lies the decomposition of one of the rotation matrices on the maximal cut into (generalised) orthogonal and symmetric matrices. We then observe that the generalised symmetric part can always be expressed in terms of known functions, while the generalised orthogonal part parametrises the $\eps$-functions that we expect to define genuinely new functions. Moreover, the polynomial equations satisfied by the generalised symmetric part can always be reduced to linear constraints, which considerably simplifies solving these otherwise non-linear relations. The power of this approach was recently demonstrated in ref.~\cite{Duhr:2025kkq} for the computation of the canonical differential equations for the three-loop banana integral with four distinct non-zero masses (see also ref.~\cite{Pogel:2025bca} for an alternative approach), where more than half of the $\eps$-functions could be expressed in terms of rational functions and derivatives of the underlying K3 surface. 

This paper is organised as follows: In section~\ref{sec:conventions} we briefly review the method of ref.~\cite{Gorges:2023zgv} to obtain a canonical basis, as well as some basic results from twisted cohomology theory. In section~\ref{sec:can_int_matrix} we review what the intersection matrix becomes in a canonical basis and we describe how it  can be used to derive polynomial relations between $\eps$-functions on the maximal cut. We also present some novel results for the canonical intersection matrix. Section~\ref{sec:constraints} contains the main result of this paper, namely the decomposition into generalised orthogonal and symmetric parts, and how this decomposition can be used to efficiently obtain and solve the polynomial constraints for the $\eps$-functions on the maximal cut. We also present a novel way to obtain the canonical intersection matrix directly, without having to compute the intersection matrix in the original basis before rotation to the canonical form. In section~\ref{sec:examples} we illustrate our method on various examples, including several classes of hypergeometric functions and a four-loop banana integral with two distinct masses considered originally in ref.~\cite{Maggio:2025jel}. In section~\ref{sec:non-maximal_cuts} we illustrate on two examples how our ideas may be applied beyond maximal cuts. Finally, in section~\ref{sec:conclusions}, we draw our conclusions. We include two appendices and an ancillary file where we collect analytic formulas relevant to understanding the examples in the main text.

\section{Feynman integrals, differential equations and intersection matrices}
\label{sec:conventions}
In this section, we introduce the main objects of interest in this paper: Feynman integrals, their cuts and their differential equations. After defining Feynman integrals in subsection~\ref{subsecdef}, we review the differential equations they satisfy, focusing on an algorithm to obtain a so-called system of \emph{canonical} differential equations. 
Since our goal is to connect certain aspects of canonical differential equations to the theory of twisted (co-)homology groups we review this framework in subsection~\ref{intersectionmatrixintro}.

\subsection{Feynman Integrals and their cuts}
\label{subsecdef}

\paragraph{Basic definitions.} 
The main focus of our paper are Feynman integrals in dimensional regularisation~\cite{tHooft:1972tcz,Cicuta:1972jf,Bollini:1972ui} in $D=d-2\eps$, which have the form
\begin{align}
\label{eq:FI_def}
I_{\bs{\nu}}^D \left(\{p_i\cdot p_j\},\{m_i^2\} \right)= e^{L\gamma_E \varepsilon}\int \left(\prod_{j=1}^L \frac{\rd^D \ell_j}{i\pi^\frac{D}{2}}\right) \frac{1}{D_1^{\nu_1}\dots D_m^{\nu_m}}\, , 
\end{align}
where $d$ is a positive integer, $\bs{\nu} = (\nu_1,\ldots,\nu_m)$ is a vector of integers and $\gamma_E=-\Gamma'(1)$ is the Euler-Mascheroni constant. We denote by $m_i^2$ the squared mass of the $i^{\textrm{th}}$ propagator $D_i$ and the  momenta flowing through the $m$ propagators are linear combinations of the $L$ loop momenta $\ell_j$ and the  $E$ linearly independent external momenta $p_j$. 
A family of Feynman integrals is characterised by its set of propagators, and its members are labeled by the vector $\bs{\nu}$. The members of a Feynman integral family are not independent, but they are connected via  linear relations, the so-called \emph{integration-by-parts} (IBP) relations~\cite{Tkachov:1981wb,Chetyrkin:1981qh}. These are encoded by the identity
\begin{equation}\label{eq:IBP_tot_diff}
\int \rd^D \ell_i\,\frac{\partial}{\partial \ell_i^\mu}\left(\frac{v^{\mu}}{D_1^{\nu_1}\dots D_m^{\nu_m}} \right)=0\,,
\end{equation}
where $v^\mu$ is an internal or external momentum. Since differentiation in the integrand produces propagators with shifted exponents, this identity leads to linear relations among different members from the same family, with the coefficients being rational functions of the kinematic scales and the dimensional regulator $\eps$. The IBP relations can be solved to express all integrals in the family in terms of a set of basis integrals, commonly referred to as \emph{master integrals} in the literature. The number of master integrals is always finite~\cite{Smirnov:2010hn,Bitoun:2017nre}. Thus, in order to compute all members of a family, we just need to compute a set of master integrals.

Nowadays one of the most commonly used ways to compute Feynman integrals is the method of differential equations~\cite{Kotikov:1990kg,Kotikov:1991hm,Kotikov:1991pm,Gehrmann:1999as,Henn:2013pwa}.  If we collect the master integrals into a vector $\bs{I}$, then  this vector fulfills a differential equation of the form 
\begin{equation}\label{eq:DEQ_prototype}
\rd \bI(\bx,\eps) = \bOmega(\bx,\eps)\bI(\bx,\eps)\,,
\end{equation}
where $\rd = \sum_{i=1}^r\rd x_i\,\partial_{x_i}$ denotes the exterior derivative with respect to the external kinematic parameters $\bx= (x_1,\ldots,x_r)$.  The entries of the matrix $\bOmega(\bx,\eps)$ are rational one-forms in $\bx$ and rational functions in $\eps$, i.e., they can be expressed in the form 
\begin{equation}
\bOmega(\bx,\eps) = \sum_{i=1}^r\rd x_i\,\bOmega_i(\bx,\eps)\,,
\end{equation}
where the $\bOmega_i(\bx,\eps)$ are matrices of rational functions in $\bx$ and $\eps$. The problem of computing the vector of master integrals $\bI(\bx,\eps)$ is equivalent to finding the solutions to the first-order linear system in eq.~\eqref{eq:DEQ_prototype} (and imposing the appropriate boundary conditions). We will come back to the question of how to find the solutions to eq.~\eqref{eq:DEQ_prototype} below.

The vector space of Feynman integrals from a given family comes with a natural filtration given by the propagators, and differentiation is compatible with this filtration. It is then always possible to find a basis of master integrals in which the matrix $\bOmega(\bx,\eps)$ is block lower-triangular. We assume that we work with such a basis. 

\paragraph{Cuts of Feynman integrals.}
In addition to Feynman integrals, we also consider their \emph{cuts}. 
From a physics perspective, cuts of Feynman integrals are obtained by putting some of the propagators on their mass shell. Mathematically, this can be implemented by a residue prescription (cf., e.g., ref.~\cite{Britto:2024mna} for a recent review). A particularly convenient representation of a Feynman integral to realise this residue prescription is the so-called Baikov representation, where the integration variables $z_i$ are chosen to be combinations of irreducible scalar products, generally including the propagators. In particular, in the Baikov representation, a family of Feynman integrals takes the form
\begin{align}
\label{eq_Baikov}
I_{\bs{\nu}}^D \left(\{p_i\cdot p_j\},\{m_i^2\} \right)&= \frac{e^{L\gamma_E\varepsilon} \left[\det G(p_1,\ldots,p_E)\right]^\frac{-D+E+1}{2}}{\pi^{\frac{1}{2}(N-L)} \left[\det C\right] \prod_{j=1}^L \Gamma\left(\frac{D-E+1-j}{2}\right)}\,\hat{I}_{\bs{\nu}}^D \left(\{p_i\cdot p_j\},\{m_i^2\} \right)\,,
\end{align}
with
\begin{align}
\label{eq_Baikov2}
\hat{I}_{\bs{\nu}}^D \left(\{p_i\cdot p_j\},\{m_i^2\} \right)&=\int_{\mathcal{C}} \rd^{n} z \left[\mathcal{B}(\bs{z})\right]^{\frac{D-L-E-1}{2}} \prod_{s=1}^{n} z_s^{-\nu_s}\, . 
\end{align}
The prefactor includes the Gram determinant
\begin{align}
    \det G(q_1,\dots, q_k)=\det \left(-q_i\cdot q_j\right)\,,
\end{align}
and the integrand contains the \textit{Baikov polynomial}
\begin{align}
    \mathcal{B}(\bs{z}) =\det G(\ell_1,\ldots,\ell_L,p_1,\ldots,p_E) \,,
\end{align}
with $\bs{z}=(z_1,\ldots,z_{n})$, and the number $n$ of integration variables is 
\begin{align}
    n= \frac{1}{2}L (L+1)+EL\, .  
\end{align}
We may choose the first $m$ of these variables to be the propagators themselves. If $m<n$, we introduce additional propagators and set their exponents $\nu_i$ to zero. Equivalently, we can also introduce additional irreducible scalar products (ISPs) directly.  

For multi-loop integrals, it is often convenient to use the so-called \textit{loop-by-loop} approach \cite{Frellesvig:2017aai,Frellesvig:2024ymq} to compute the Baikov representation. This means that  we introduce a Baikov parametrisation for each loop separately, and this often results in integrals of lower dimension and an integrand that is a product of several Baikov polynomials of the form
\begin{equation}\label{Baikov_LL}
 \int_{\mathcal{C}'} \rd^{n'}\! z\,
 \mathcal{B}_1(\bs{z})^{\mu_1}\dots \mathcal{B}_K(\bs{z})^{\mu_K}\prod_{s=1}^{n'} z_s^{-\nu_s}\,,
\end{equation}
where $n'\le n$. The Baikov polynomials $\mathcal{B}_i$, the exponents $\mu_i$, the contour $\mathcal{C}'$ and the proportionality factor are determined by the order in which the loops are integrated out. 
We note that each exponent $\mu_i$ is linear in the dimensional regulator $\eps$. More specifically, it takes the form,
\beq\label{eq:mu_def}
\mu_i = \frac{m_i}{2}\pm\eps\,,\qquad m_i\in\mathbb{Z}\,.
\eeq

If we work with the Baikov representation, the residue prescription to define cuts of Feynman integrals is particularly simple: We just take residues around zero for all variables that correspond to the propagators that we want to cut, cf.,~e.g.,~refs.~\cite{Bosma:2017ens,Frellesvig:2017aai}. 
The residue prescription preserves the linear relations and the differential equations that the integrals fulfill~\cite{Anastasiou:2002yz,Anastasiou:2003yy}. Thus, the differential equations for the full integral family contain the differential equations for the cuts, and conversely we can directly deduce information about the full differential equation from the one for the cuts. In this paper we will be particularly concerned with the so-called \textit{maximal cuts}, which correspond to integrals where all propagators are put on their mass shell, meaning that the integral is evaluated on a contour that encircles all propagator poles. The diagonal blocks of the matrix $\bs{\Omega}$ in eq.~\eqref{eq:DEQ_prototype} describe the differential equations for the maximal cuts~\cite{Primo:2016ebd,Frellesvig:2017aai,Bosma:2017ens}. 

\paragraph{$\eps$-factorised and canonical differential equations.} Let us now discuss strategies for solving the system of differential equations in eq.~\eqref{eq:DEQ_prototype}. We note that within the framework of dimensional regularisation, we are usually not interested in the exact dependence of $\bI(\bx,\eps)$ on the dimensional regulator $\eps$, but only in the coefficients of the Laurent expansion around $\eps=0$.

We start by noting that there is some arbitrariness in how we choose our basis of master integrals $\bI(\bx,\eps)$, and we could have chosen another basis $\bJ(\bx,\eps)$, related to $\bI(\bx,\eps)$ by
\begin{align}
\label{conventiontrafo}
 \bs{J}(\bs{x},\eps)    = \bs{U}(\bs{x},\eps)\bs{I}(\bs{x},\eps)\, ,
\end{align}
where $\bs{U}(\bs{x},\eps)$ is some matrix of full rank.\footnote{We assume that also the new basis respects the natural filtration on the vector space of master integrals given by the propagators.} A judicious choice of basis may have an impact on our ability to solve the system of differential equations for the master integrals. 
A particularly convenient choice is a so-called \emph{$\eps$-factorised} basis, where the differential equation takes the form 
\begin{equation}\label{eq:DEQ_can_basis}
\rd\bJ(\bx,\eps) = \eps\bA(\bx)\bJ(\bx,\eps)\,, 
\end{equation}
where $\bs{A}(\bs{x})$ only depends on $\bx$, but not on $\eps$. The differential equation matrices of these bases are related by 
\begin{equation}\label{eq:gauge_to_canonical}
\eps\bA(\bx) =\Big(\bU(\bx,\eps)\bOmega(\bx,\eps)+\rd\bU(\bx,\eps)\Big) \bU(\bx,\eps)^{-1}\,.
\end{equation}
Such {$\eps$-factorised} bases were first introduced in ref.~\cite{Henn:2013pwa} in the context of Feynman integrals that evaluate to multiple polylogarithms~\cite{Goncharov:1998kja,Remiddi:1999ew,Gehrmann:1999as}, where the differential equation matrix $\bA(\bx)$ is a matrix of $\rd\!\log$-forms.\footnote{The converse, however, is not true, and the fact that $\bA(\bx)$ can be expressed in terms of dlog-forms does not imply that the solution can be expressed in terms of multiple polylogarithms~\cite{Duhr:2020gdd}.} Various proposals have been made for how to obtain $\eps$-factorised differential equations for Feynman integrals that evaluate to other classes of special functions~\cite{Primo:2017ipr,Adams:2018yfj,Broedel:2018rwm,Bogner:2019lfa,Frellesvig:2021hkr,Chen:2022lzr,Dlapa:2022wdu,Pogel:2022vat,Pogel:2022ken,Pogel:2022yat,Gorges:2023zgv,Chen:2025hzq,Duhr:2025lbz,Maggio:2025jel,e-collaboration:2025frv} (for an approach that advocates a form of the differential equations without $\eps$-factorisation, see ref.~\cite{Chaubey:2025adn}). 

A main advantage of working with an $\eps$-factorised basis is that it is easy to obtain the  coefficients of the Laurent expansion in $\eps$. 
The solution of a differential equation in $\eps$-form as in eq.~\eqref{eq:DEQ_can_basis} can be expressed as
\begin{equation}
\bJ(\bx,\eps) = \mathbb{P}_{\gamma}(\bx,\eps)\bJ_0(\eps)\,,
\end{equation}
where $\bJ_0(\eps)$ denotes the value of $\bJ(\bx,\eps)$ at some point $\bx=\bx_0$ and $\gamma$ is a path from $\bx_0$ to a generic point $\bx$. We denote by $\mathbb{P}_{\gamma}(\bs{x},\eps)$ the path-ordered exponential 
\begin{equation}\label{eq:Pexp_def}
\mathbb{P}_{\gamma}(\bx,\eps) = \mathbb{P}\exp\left[\eps\int_{\gamma}\bA(\bx)\right]\,,
\end{equation}
which can be expanded in $\eps$ up to any given order, and the coefficients of this expansion involve iterated integrals~\cite{ChenSymbol},
\begin{equation}\label{eq:Pexp_def_exp}
\mathbb{P}_{\gamma}(\bx,\eps) = \mathds{1} + \sum_{k=1}^\infty \eps^k\sum_{1\le i_1,\ldots,i_k\le p}\bA_{i_1}\cdots\bA_{i_k}  I_{\gamma}(\omega_{i_1},\ldots,\omega_{i_k}) \,,
\end{equation}
with
\beq\label{eq:A_basis_omega}
\bA(\bx) = \sum_{i=1}^p \bA_i \omega_i\,,
\eeq
and where we defined the iterated integral
\begin{equation}\bsp\label{eq:iterated_int_def}
 I_{\gamma}(\omega_{i_1},\ldots,\omega_{i_k}) &\,= \int_{\gamma} \omega_{i_1}\cdots \omega_{i_k}\\
 &\,= \int_{0\le \xi_k\le \cdots \le \xi_1\le1}\rd \xi_1 \,f_{i_1}(\xi_1)\,\rd \xi_2 \,f_{i_2}(\xi_2)\cdots \rd \xi_k\,f_{i_k}(\xi_k)\,,
\esp\end{equation}
with $\gamma^*\omega_{i_r}=\rd\xi_r\,f_{i_r}(\xi_r)$ the pullback of $\omega_i$ to the path $\gamma$. 
Generally, the entries of $\bs{A}(\bs{x})$ that enter the $\eps$-factorised differential equation, and thus also of the matrix $\bU(\bx,\eps)$ in eq.~\eqref{conventiontrafo}, are not rational in $\bs{x}$. In the literature there are examples where the matrix $\bs{A}(\bx)$ involves, for example, modular forms~\cite{Adams:2018yfj,Broedel:2018rwm,Adams:2018bsn,Dlapa:2022wdu,Muller:2022gec,Pogel:2022ken,Klemm:2024wtd}, the coefficients of the Kronecker-Eisenstein series~\cite{Bogner:2019lfa,Muller:2022gec}, (integrals of) Siegel modular forms~\cite{Duhr:2024uid}, or periods of Calabi-Yau varieties and integrals thereof~\cite{Pogel:2022ken,Pogel:2022vat,Pogel:2022yat,Gorges:2023zgv,Klemm:2024wtd,Frellesvig:2024rea,Driesse:2024feo,Duhr:2025lbz,Maggio:2025jel,Duhr:2025kkq,Pogel:2025bca}. 

While the methods of refs.~\cite{Primo:2017ipr,Adams:2018yfj,Broedel:2018rwm,Bogner:2019lfa,Frellesvig:2021hkr,Chen:2022lzr,Dlapa:2022wdu,Pogel:2022vat,Pogel:2022ken,Pogel:2022yat,Gorges:2023zgv,Chen:2025hzq,Duhr:2025lbz,Maggio:2025jel,e-collaboration:2025frv} all lead to $\eps$-factorised differential equations, the resulting bases are not fully equivalent, and the solutions may typically involve different classes of iterated integrals (see, e.g.,~ref.~\cite{Frellesvig:2023iwr} for a comparison of different bases). Here we use the basis defined by the method of ref.~\cite{Gorges:2023zgv}. Before we describe the algorithm underlying this approach, let us discuss some of its features. The method of ref.~\cite{Gorges:2023zgv} was introduced as a way to achieve $\eps$-factorisation while being inspired by the extension to the elliptic case~\cite{Broedel:2018qkq} of the notions of uniform transcendental weight~\cite{Kotikov:2010gf} and pure functions~\cite{Arkani-Hamed:2010pyv}, known from the case of Feynman integrals that evaluate to multiple polylogarithms.
In ref.~\cite{Duhr:2025lbz} it was empirically observed that the systems of differential equations obtained using the method of ref.~\cite{Gorges:2023zgv} have additional features besides mere $\eps$-factorisation. In particular, it was observed that the differential one-forms that define the matrix $\bA(\bx)$ define linearly independent cohomology classes\footnote{$\eps$-factorised systems that satisfy this particular property were dubbed in $C$-form in ref.~\cite{Duhr:2024xsy}.}  and locally they only have simple poles. Note that these additional properties are always satisfied if the matrix $\bA(\bx)$ is a matrix of dlog-forms. These two properties imply additional features for the classes of iterated integrals that appear in the $\eps$-expansion. Indeed, the first property implies that these iterated integrals are linearly independent~\cite{deneufchatel:hal-00558773,Duhr:2024xsy}, while a consequence of the second property is that locally the iterated integrals diverge at most like a power of a logarithm. We refer to systems that satisfy these properties in addition to $\eps$-factorisation as \emph{canonical}.

Let us conclude this subsection by briefly reviewing the algorithm introduced in refs.~\cite{Gorges:2023zgv,Duhr:2025lbz}. We expect that the basis constructed by applying this algorithm is equivalent to the one obtained from the methods described in refs.~\cite{Pogel:2022vat,Pogel:2022ken,Pogel:2022yat,Maggio:2025jel,e-collaboration:2025frv} (see refs.~\cite{Duhr:2025lbz,Maggio:2025jel} for a detailed discussion). 
An important point of the method of ref.~\cite{Gorges:2023zgv} is the choice of a good starting basis $\bs{I}(\bx,\eps)$. Often it is convenient to choose a so-called \textit{derivative basis}, i.e., a basis where the master integrands are derivatives of each other, see also refs.~\cite{Gorges:2023zgv, Maggio:2025jel,Duhr:2024uid,Duhr:2025lbz}. More generally, it was proposed in ref.~\cite{Duhr:2025lbz} that a good starting basis is aligned with the mixed-Hodge structure~\cite{PMIHES_1971__40__5_0,PMIHES_1974__44__5_0} of the geometry attached to the maximal cuts of the integrals. Since for the purposes of this paper the starting basis is not central, we just assume that we have identified a good starting basis $\bs{I}(\bx,\eps)$.
The goal is then to construct a rotation $\bs{U}(\bs{x},\eps)$, such that the basis $\bs{J}(\bs{x},\eps)$ defined through eq.~\eqref{conventiontrafo} is in $\eps$-factorised form. We stress again that the algorithm of ref.~\cite{Gorges:2023zgv} only aims at achieving $\eps$-factorisation, but that in all known cases the resulting bases are also canonical according to the definition given above. We expect that this is true in general. 

The algorithm of ref.~\cite{Gorges:2023zgv} constructs the rotation matrix $\bU(\bx,\eps)$ as a sequence of rotations,
\begin{align}\label{Usplitting}
\bU(\bs{x},\eps)=\bUt(\bx,\eps)\bUeps(\eps)\bUss(\bx)\,.
\end{align}
The first rotation $\bUss(\bx)$ is a block-diagonal matrix, where each block is the inverse of the semi-simple part of the period matrix~\cite{Broedel:2018qkq,Broedel:2019kmn,Ducker:2025wfl,Maggio:2025jel}, related to the different geometries obtained from the maximal cuts.
This rotation introduces rational functions of the periods (and their derivatives) of the corresponding geometries into the basis. 
The matrix $\bUeps(\eps)$ is diagonal and rescales the overall power of $\eps$. It is constructed such that after this rotation the highest power of $\eps$ appearing in the differential equation is $\eps^1$, and the non-positive (and zeroth) powers of $\eps$, are confined to a lower-triangular matrix. The final rotation $\bUt(\bx,\eps)$ is the central object of interest in this paper. Unlike the entries of $\bUss(\bx)$ and $\bUeps(\eps)$, which are fixed a priori without detailed knowledge of the specific differential equations, the entries of $\bUt(\bx,\eps)$
are obtained by making an ansatz and requiring that the transformed differential equation matrix in eq.~\eqref{eq:gauge_to_canonical} is $\eps$-factorised.
Due to the appearance of the differential $\rd \bs{U}$ in eq. \eqref{eq:gauge_to_canonical}, this leads to a set of differential equations that fix these entries.
If we expand $\bUt(\bx,\eps)$ in $\eps$, we can write
\begin{equation}\label{eq:expU}
    \bUt(\bx,\eps)=\bUt^{(0)}(\bx)+\frac{1}{\eps}\bUt^{(-1)}(\bx)+...+\frac{1}{\eps^{n-1}}\bUt^{(1-n)}(\bx)\,,
\end{equation}
where $n$ is the dimension of the underlying geometry. 

At this point we make a crucial observation. We have already mentioned that, even though the matrix $\bOmega(\bx,\eps)$ in eq.~\eqref{eq:DEQ_prototype} is a matrix of rational one-forms, this may not be the case for the matrix $\bA(\bx)$ describing the canonical differential equation. We can now see why this is the case. The rotation matrix $\bUss(\bx)$ introduces the (derivatives of the) periods of some algebraic varieties into the basis, and thus into the differential equations. Periods of algebraic varieties are typically not expressible in terms of rational (or even algebraic) functions, but they define transcendental functions. Said differently, if $\cF$ is the field of algebraic functions in $\bx$ and $\cFss$ is the field obtained by adding to $\cF$ the periods (and their derivatives) that appear in $\bUss(\bx)$, then we see that the entries of the matrices $\bOmega_i(\bx,\eps)$ lie in $\cF(\eps)$,\footnote{Here $\cF(\eps)$ denotes the field of rational functions in $\eps$ with coefficients in $\cF$.} while the differential one-forms that enter the differential equation take the form $\sum_{j=1}^rf_{ij}(\bx)\,\rd x_j$, with $f_{ij}\in\cFss$. 
As the rotation matrices $\bUt^{(i)}(\bx)$ are defined via differential equations with coefficients in $\cFss$, the solutions to these differential equations can be cast in the form of (iterated) integrals over kernels from $\cFss$. The functions defined in this way were called \emph{$\eps$-functions} in ref.~\cite{Duhr:2025kkq}. A priori, $\eps$-functions will define new classes of transcendental functions, that cannot necessarily be expressed in terms of algebraic functions involving periods and their derivatives. We denote by $\cFeps$ the field of functions obtained by adding to $\cFss$ the $\eps$-functions and their derivatives, i.e., $\cFeps$ is the field of rational functions in $\eps$-functions and their derivatives with coefficients in $\cFss$. Hence we have an inclusion of function fields,
\beq
\cF \subseteq \cFss\subseteq \cFeps\,.
\eeq

It is natural to ask if the new functions that are added at every step are genuinely new, or if they can be expressed in terms of functions that were already known. For the first inclusion, this boils down to the question if the corresponding periods of algebraic varieties are expressible in terms of rational or algebraic functions. While this is an open question in general, one typically expects that periods of families of non-trivial algebraic varieties, like families of Riemann surfaces of genus $g\ge 1$ or families of Calabi-Yau varieties, define transcendental functions, cf.,~e.g.,~refs.~\cite{Schneider1937,Wuestholz1989,10.1007/BFb0099460,waldschmidt:hal-00411301}. For the second inclusion, i.e., the question which $\eps$-functions can be expressed in terms of algebraic functions involving periods, and their derivatives, the situation is much less explored. One of the main goals of this paper is to take first steps in this direction, and to present a way to identify $\eps$-functions that are expressible in terms of elements from $\cFss$. In order to achieve this, we need to rely on results from twisted cohomology theory, which is the appropriate mathematical framework to study Feynman integrals in dimensional regularisation~\cite{Mizera:2017rqa,Mastrolia:2018uzb}. We therefore present a brief review of twisted cohomology groups in the next subsection.

\subsection{Feynman integrals and intersection matrices} 
\label{intersectionmatrixintro}

In general, Feynman integrals in dimensional regularisation are multi-valued, and they can naturally be interpreted as periods of twisted (co-)homology groups. Loosely speaking, twisted cohomology theory studies integrals of the form
\beq
\int_{\mathcal{C}}\Phi\varphi\,,
\eeq
where $\Phi$ is a multivalued function, $\varphi$ is a rational one-form, and $\mathcal{C}$ is a \emph{twisted cycle}, i.e., a cycle together with information on the branch on which $\Phi$ is to be evaluated. We give a very short review of the relevant concepts and their relation to Feynman integrals here.  We keep this review very brief, we refer to the standard literature for extensive discussions and reviews of the theory of twisted (co-)homology groups and their use in physics \cite{Mastrolia:2018uzb,yoshida_hypergeometric_1997,aomoto_theory_2011,Mizera:2017rqa,Mizera:2019gea,matsumoto_relative_2019-1,Frellesvig:2017aai,Frellesvig:2020qot,Frellesvig:2019kgj,Weinzierl:2020xyy,Weinzierl:2020nhw,Cacciatori:2021nli,Chestnov:2022xsy,Brunello:2023rpq,Frellesvig:2019uqt,Fontana:2023amt,Brunello:2024tqf,Crisanti:2024onv,Lu:2024dsb,
Caron-Huot:2021xqj,Caron-Huot:2021iev,Chen:2022lzr,Giroux:2022wav,Chen:2023kgw,Gasparotto:2023roh,Giroux:2024yxu,Chen:2024ovh,Duhr:2024xsy,e-collaboration:2025frv}.

In order to define the twisted (co)-homology groups related to a family of multi-valued integrals, we split the integrand into a multi-valued part $\Phi$, called the \textit{twist}, and a single-valued part $\varphi$. Here, we generally consider Feynman integrals in the Baikov representation computed in the loop-by-loop approach -- see eq. (\ref{Baikov_LL}) -- and define the twist to be
\begin{align}
\label{baikovtwist}
    \Phi= \mathcal{B}_1(\bs{z})^{\mu_1}\dots \mathcal{B}_K(\bs{z})^{\mu_K}\, . 
\end{align}
Due to the dimensional regulator $\eps$ appearing in the $\mu_i$, this factor will always be multi-valued. 
We want to interpret the single-valued part $\varphi$ as defining an element of a so-called twisted cohomology group. For Feynman integrals this single-valued factor typically takes the form
\begin{align}
 \varphi=   \rd^{n} z \,f(\bz)\,\prod_{s=1}^{n'} z_s^{-\nu_s}\,,
\end{align}
where $f(\bz)$ is a rational function with poles at most at the zeroes of the twist.
If we consider cuts, the integrand arises from this factor after a certain number of residues have been taken on the full integrand $\Phi\varphi$. In particular, for a Feynman integral, this factor will have poles. Since by definition the maximal cut is obtained by taking all residues at the propagators, $\varphi$ does not have poles at $z_s=0$ for maximal cuts. For the twisted cohomology groups to consider, this is an important distinction: If the factor $\varphi$ does not have poles at loci where the twist is non-vanishing, we can simply consider it as an element of a twisted cohomology group $H^n(X,\nabla_{\!\Phi})$, where  $X=\mathbb{C}^n-\Sigma$, and $\Sigma$ is a union of hypersurfaces defined by the vanishing of the twist and $\nabla_{\!\Phi}=\rd + \rd\! \log \Phi \wedge$. This group contains equivalence classes of forms that are closed with respect to $\nabla_{\!\Phi}$ modulo forms that are exact. If $\varphi$ has poles outside of $\Sigma$ -- which generally happens if we consider non-maximal cuts -- we need to work with \emph{relative} twisted cohomology groups, see  refs.~\cite{Caron-Huot:2021iev,Caron-Huot:2021xqj,matsumoto_relative_2019-1},  or introduce a regulator, i.e., a factor in the twist whose exponent is taken to zero at the end of the calculation. An extended discussion of these frameworks can be found in refs.~\cite{Brunello:2023rpq,Duhr:2024rxe}. 

A twisted cohomology group is a vector space and we can construct a basis for it. Similarly, one can define the dual twisted cohomology group $H_{\rm{dR}}^n(\check{X},\check{\nabla}_{\!\Phi})$ by the connection $\check{\nabla}_{\!\Phi}= \rd -\rd\! \log \Phi\wedge $ on a space $\check{X}$ which we do not specify much further. The group $H_{\rm{dR}}^n(\check{X},\check{\nabla}_{\!\Phi})$ consists of equivalence classes of \emph{dual} differential forms. We will not need the dual differential forms directly, and it suffices to say that if we do not consider poles outside $X$, then $X=\check{X}$ and the bases of both groups can be taken in a very similar form (see below). We will refer to this situation as \emph{self-duality}, and it will play an important role in this paper. Otherwise we need to work with relative twisted cohomology groups, and the elements of the dual basis take a  form that is different from the forms that define $H_{\rm{dR}}^n({X},{\nabla}_{\!\Phi})$. Explicit constructions in particular for Feynman integrals can be found in refs.~\cite{Caron-Huot:2021iev,Caron-Huot:2021xqj}. 

The duality between the twisted cohomology group and its dual is expressed by the intersection pairing $\langle . |.\rangle $. The \emph{cohomology intersection matrix} is the matrix of pairings between basis elements $\varphi_i$ and $\check{\varphi}_j$ of the cohomology group and its dual: 
\begin{align}
    {C_{ij}} =  \frac{1}{(2\pi i)^n} \langle \varphi_i |\check{\varphi}_j\rangle = \frac{1}{(2\pi i)^n}\int_X \varphi_i\wedge \check\varphi_{j}\, .
\end{align}
This object is central to all considerations in this paper. Its computation is discussed in detail in the literature, see refs.~\cite{yoshida_hypergeometric_1997,Mizera:2017rqa,Mizera:2019gea, aomoto_theory_2011}, and refs.~\cite{Caron-Huot:2021iev,Caron-Huot:2021xqj} for relative twisted cohomology groups. Note, that the cohomology intersection matrix is only defined after we specify bases $\{\varphi_i\}$ and $\{\check{\varphi}_i\}$ for the twisted cohomology group and its dual. If we perform a change of basis by rotating them with $\bs{U}$ and $\bs{\check{U}}$ respectively, the cohomology intersection matrix changes as 
\begin{align}
\label{trafointemat}
    \bs{C}_{\text{new}}= \bs{U}\,  \bs{C}\, \bs{\check{U}}^T\, . 
\end{align}

Similarly to the twisted cohomology groups, one can also define the twisted homology group $H_n(X, \check{\mathcal{L}}_\Phi)$ and its dual $H_n(\check{X},\mathcal{L}_\Phi)$, where $\check{\mathcal{L}}_\Phi$ is the local system that defines the multivalued twist $\Phi$. These groups contain equivalence classes of cycles. There exists an intersection pairing $[.\, |\, .]$ between the cycles and their duals, which counts the local intersections taking the local branch choice of the twist into account. We refer to the matrix whose entries are the intersections between fixed bases $\{\gamma_i\}$ and $\{\check{\gamma}_i\}$ of the homology group and its dual as the \emph{homology intersection matrix}, and we denote it by $\bs{H}$. 

Finally, we may also pair (dual) cycles and (dual) differential forms. This defines \emph{twisted periods}, obtained by pairing cycles from the (dual) homology group with forms from the (dual) cohomology group via integration. Assuming that we have fixed bases for all (co)homology groups and their duals, the \emph{period matrix} is obtained by applying this pairing to all basis elements:
\begin{align}
    {P}_{ij} = \langle \varphi_i |\gamma_j] = \int_{\gamma_j}\Phi\varphi_i\, . 
\end{align}
Similarly, one defines the \emph{dual period matrix},
\begin{align}
    {\check{P}}_{ij} = [\check{\gamma}_j|\check{\varphi}_i \rangle = \int_{\check{\gamma}_j}\Phi^{-1}\check{\varphi}_i\, . 
\end{align}

In the non-relative case, we can choose the bases and dual bases for both the homology and cohomology groups to be regularised versions of each other, i.e., $\check{\varphi}_i=[\varphi_i]_c$ and $\gamma_i = [\check{\gamma}_i]_c$. The details of the regularisation procedure implied by the notation $[\cdot ]_c$ are not important for the present discussion and can be found in refs.~\cite{aomoto_theory_2011, yoshida_hypergeometric_1997}. For that choice and the twist  as in eq. (\ref{baikovtwist}), we find: 
\begin{align} \label{eq:checkP_to_minus}
\check{P}_{ij} 
=P_{ij}|_{\mu_k\rightarrow -\mu_k}\, . 
\end{align} 
In particular, in ref.~\cite{Duhr:2024xsy} it was shown that, as a consequence of eq.~\eqref{eq:mu_def}, for maximal cuts we can always choose the dual basis such that $\bs{\check{P}}=\bs{P}|_{\eps\rightarrow -\eps}$. This may be interpreted as a concrete manifestation of the self-duality of maximal cuts.  

In general, the period matrix and its dual fulfill differential equations of the form 
\begin{align}
\label{deqperiod}
    \rd \bs{P} = \bs{\Omega} \bs{P}\textrm{~~~and~~~}\rd \bs{\check{P}} = \bs{\check{\Omega}} \bs{\check{P}}\, ,
\end{align}
where $\rd$ denotes the exterior derivative with respect to the external parameters, which in physical examples generally contain the kinematic parameters. The cohomology intersection matrix satisfies the differential equation
\begin{align}\label{DEQC}
\rd \bC=\bs{\Omega} \bC+\bC\bs{\check{\Omega}}^T\,.
\end{align} 
The period and intersection matrices are related by the twisted Riemann bilinear relations~\cite{Cho_Matsumoto_1995}
\beq
\label{eq:TRBRs}
\bH = \frac{1}{(2\pi i)^n}\,\bP^T\big(\bC^{-1}\big)^T\bs{\check{P}}\,.
\eeq


\section{The canonical intersection matrix}
\label{sec:can_int_matrix}
The goal of our paper is to identify relations between the $\eps$-functions that are introduced by the rotation $\bUt(\bx,\eps)$, or equivalently, to describe a (possibly minimal) set of generators for the function field $\cFeps$. As already mentioned, we do this by combining insights from canonical differential equations with tools from twisted cohomology, in particular our knowledge of the intersection matrix for a canonical basis. We start by reviewing some results of ref.~\cite{Duhr:2024xsy}, and we then extend them by some new results that will be useful to understand the relationship between canonical bases and twisted cohomology.

\subsection{The intersection matrix in a canonical basis}\label{sec:canon_int}
Assume that we have determined bases of master integrals and their duals, and denote the corresponding cohomology intersection matrix by $\bC(\bx,\eps)$. Assume furthermore that we have determined (e.g., by using the algorithms of refs.~\cite{Pogel:2022vat,Pogel:2022ken,Pogel:2022yat,Gorges:2023zgv,Duhr:2025lbz,Maggio:2025jel}) rotations $\bU(\bx,\eps)$ and $\bs{\check{U}}(\bx,\eps)$ to canonical bases for both the integrals and their duals. We denote the canonical differential equation matrices for the integrals and their duals by $\eps\bA(\bx)$ and $\eps\bs{\check{A}}(\bx)$, respectively. From eq.~\eqref{trafointemat} we know that after rotating to a canonical basis, the intersection matrix becomes
\beq\label{eq:Cc_def}
\bC_c(\bx,\eps) = \bU(\bx,\eps)\bC(\bx,\eps)\bs{\check{U}}(\bx,\eps)^T\,. 
\eeq
In ref.~\cite{Duhr:2024xsy} some of the authors have shown that the intersection matrix $\bC_c(\bx,\eps)$ for the canonical bases has a very special form. In particular, $\bC_c(\bx,\eps)$ is constant in $\bx$ and takes the form 
\beq\label{eq:Cc_to_Delta}
\bC_c(\bx,\eps) = \bDelta\,\bM(\eps)\,,
\eeq
where $\bDelta$ is a constant matrix (the entries will typically be rational numbers) and $\bM(\eps)$ is a matrix of rational functions in $\eps$ that commutes with $\bs{\check{A}}(\bx)^T$, 
\beq\label{eq:commutation}
\big[\bM(\eps), \bs{\check{A}}(\bx)^T\big]=0\,.
\eeq
Due to this commutation relation the precise form of $\bM(\eps)$ will be irrelevant, and it drops out of all equations, and all the information is contained in the constant matrix $\bDelta$. We refer to $\bDelta$ as the \emph{canonical intersection matrix}. Note that, strictly speaking, $\bDelta$ is not uniquely defined, as we can always redefine it by a constant matrix (in $\eps$) that commutes with $\bs{\check{A}}(\bx)^T$,
\beq
\bC_c(\bx,\eps) = \bDelta'\,\bM'(\eps), \textrm{~~~with~~~} \bDelta'=\bDelta\bs{E}^{-1},\quad \bM'(\eps)=\bs{E}\bM(\eps)\,.
\eeq
In practice, this does not cause any confusion, as we may fix $\bs{E}$ by requiring that, e.g., $\bM(\eps) = \eps^m\mathds{1}+\mathcal{O}(\eps^{m+1})$ for some non zero integer $m$. 

The canonical intersection matrix $\bDelta$ depends on the choices of canonical bases for the integrals and their duals. If we change the canonical bases by a rotation by some constant matrices $\bs{M}$ and $\bs{\check{M}}$, then the canonical intersection matrix changes to $\bs{M}\bDelta\bs{\check{M}}^T$. We may use this freedom to identify a specific canonical basis where the canonical intersection matrix takes a simpler form. We will return to this point in section~\ref{sec:consequences}.

The canonical intersection matrix relates the differential equation matrices for the canonical basis and its dual. More specifically, using eq.~\eqref{eq:commutation} and the fact that $\bC_c(\bx,\eps)$ is constant in $\bx$, we see that eq.~\eqref{DEQC} implies~\cite{Duhr:2024xsy},
\beq\label{eq:A_dual_Delta}
\bA(\bx) = -\varphi_{\Delta}\big(\bs{\check{A}}(\bx)\big) := -\bDelta\bs{\check{A}}(\bx)^T\bDelta\!^{-1} \,.
\eeq
This relation has an immediate consequence: In order to define the canonical intersection matrix, we need to independently rotate the basis of master integrals and their duals into a canonical basis using the rotation matrices $\bU(\bx,\eps)$ and $\bs{\check{U}}(\bx,\eps)$. From each of the two rotations, we obtain function fields $\cFss$, $\cFeps$ and $\check{\cF}_{\!\textrm{ss}}$, $\check{\cF}_{\!\eps}$, respectively. Equation~\eqref{eq:A_dual_Delta} shows that these are mutually the same,
\beq
\cF \subseteq  \check{\cF}_{\!\textrm{ss}} = \cFss \subseteq \check{\cF}_{\!\eps} = \cFeps\,.
\eeq
In particular, Feynman integrals and their duals share the same $\eps$-functions, and the dual canonical differential equation matrix $\bs{\check{A}}(\bx)$ can be expanded into the same basis of differential forms as $\bA(\bx)$,
\beq
\bs{\check{A}}(\bx) = \sum_{i=1}^p\bs{\check{A}}_i\,\omega_i\,,\qquad \bA_i=-\varphi_{\Delta}\big(\bs{\check{A}}_i\big)\,,
\eeq
where the differential one-forms $\omega_i$ were defined in eq.~\eqref{eq:A_basis_omega}. 

In the remainder of this paper we will argue that we can make even stronger statements when working with a self-dual scenario. If we assume that we are working with a dual basis such that $\bs{\check{P}} = \bP_{|\eps\to-\eps}$, then we must have
$\bs{\check{A}}(\bx) = -\bA(\bx)$, so that eq.~\eqref{eq:A_dual_Delta} reduces to~\cite{Duhr:2024xsy}
\beq\label{eq:A_phiA}
\bA(\bx) = \varphi_{\Delta}\big(\bA(\bx)\big) = \bDelta\bA(\bx)^T\bDelta\!^{-1}\,,\textrm{~~~if self-duality holds}\,,
\eeq
and $\bDelta$ is either symmetric or antisymmetric.
Moreover, if $\bA(\bx)$ is irreducible (i.e., we cannot find a basis such that $\bA(\bx)$ is block-triangular), Schur's lemma implies that the matrix $\bM(\eps)$ is proportional to the identity,
\beq\label{eqfeps}
\bM(\eps) = f(\eps)\,\mathds{1}\,,
\eeq
where $f$ is a rational function of $\eps$. In particular, this means that in this case $\bDelta$ is uniquely defined up to multiplication by a non-zero scalar. 

\subsection{Obtaining the canonical intersection matrix}
\label{subsec:obtaining_Delta}
In this section we discuss ways to determine the canonical intersection matrix $\bDelta$. We assume that we have determined a canonical basis for the master integrals and their duals, and we follow the notation and conventions of the previous sections. 

One way to obtain the canonical intersection matrix $\bDelta$ is to rotate the intersection matrix $\bC(\bx,\eps)$ from the original bases to the canonical ones using eq.~\eqref{eq:Cc_def}. This approach, however, requires one to compute the original intersection matrix $\bC(\bx,\eps)$ independently. While various methods exist to perform this task~\cite{Frellesvig:2019kgj,Frellesvig:2019uqt,Frellesvig:2020qot,Chestnov:2022xsy,Fontana:2023amt,Brunello:2023rpq,Brunello:2024tqf}, in practice this may be prohibitively complicated. However, as we now explain, it is often possible to obtain $\bDelta$ without explicit knowledge of the intersection matrix $\bC(\bx,\eps)$ computed in the original bases. The starting point is the realisation that the canonical intersection matrix is essentially unique, in the following sense:
\begin{lemma}
Let $\bD$ be a constant matrix such that eq.~\eqref{eq:A_dual_Delta} holds with $\bDelta$ replaced by $\bD$. Then there is a constant matrix $\bs{E}$ such that
\beq
\bD = \bDelta\,\bs{E}^{-1} \textrm{~~~and~~~} \big[\bs{E},\bs{\check{A}}(\bx)^T\big]=0\,.
\eeq
\end{lemma}
\begin{proof}
By assumption, both $\bDelta$ and $\bD$ satisfy eq.~\eqref{eq:A_dual_Delta}, and so
\beq
\bDelta\bs{\check{A}}(\bx)^T\bDelta\!^{-1} = \bD\bs{\check{A}}(\bx)^T\bD^{-1}\,,
\eeq
or equivalently
\beq
\bs{E}\bs{\check{A}}(\bx)^T=\bs{\check{A}}(\bx)^T\bs{E}\,,
\eeq
with $\bs{E} := \bDelta\!^{-1}\bD$. The claim then immediately follows.
\end{proof}
This lemma implies that it is sufficient to find any constant solution to eq.~\eqref{eq:A_dual_Delta} to obtain the canonical intersection matrix, up to a matrix that commutes with $\bs{\check{A}}(\bx)^T$. 

Let us now discuss a strategy to obtain the canonical intersection matrix $\bDelta$ without the need to know the original intersection matrix explicitly. 
Assume that we have determined a rotation to a canonical form for both the integrals and their duals. In particular, we assume that we have representations of all $\eps$-functions, i.e., of all entries of the matrix $\bUt^{(0)}$. For example, we may have determined series or integral representations for them, and we assume that we know how to evaluate them in some region. Note that we do not require that we know how to evaluate the $\eps$-functions everywhere in parameter space, but it is enough that we have numerical control in some small region.
We start by casting eq.~\eqref{eq:A_dual_Delta} in the form
\begin{equation}\label{num}
    \bs{A}(\bx)\bs{\Delta}+\bs{\Delta}\bs{\check{A}}(\bx)^T=0\,.
\end{equation}
We interpret this equation as a linear system for the entries of the canonical intersection matrix. Since $\bs{\Delta}$ is a constant, it is enough to evaluate the entries of $\bs{A}(\bx)$ and $\bs{\check{A}}(\bx)$ for some choices of $\bx$, and we can solve for $\bDelta$. We obtain in this way a numerical solution for $\bDelta$. Since the entries of $\bDelta$ are typically simple rational numbers, we can often determine an exact representation for $\bDelta$ from this numerical result. In section~\ref{sec:consequences} we will present an alternative  approach to determining $\bDelta$ valid for self-dual scenarios.

\subsection{The canonical homology intersection matrix}
So far we have only focused on the canonical \emph{cohomology} intersection matrix. We conclude this section by briefly commenting on the corresponding canonical \emph{homology} intersection matrix. We can compute the homology intersection matrix via the twisted Riemann bilinear relations in eq.~\eqref{eq:TRBRs}.

Our goal will be to compute the homology intersection matrix $\bH$ if, for the quantities on the right-hand side, we work in a canonical basis.
In order to fully define the homology intersection matrix, we need to specify the bases of the homology group and its dual. We fix a basis of the homology group such that the period matrix (in the canonical basis) is simply given by the path-ordered exponential in eq.~\eqref{eq:Pexp_def}. Similarly, we fix a basis of dual cycles such that the dual period matrix is
\beq
\check{\mathbb{P}}_{\gamma}(\bx,\eps) = \mathbb{P}\exp\left[\eps\int_{\gamma}\bs{\check{A}}(\bx)\right]\,.
\eeq
The original period matrix $\bP$ and its dual $\bs{\check{P}}$ can be expressed in terms of the path-ordered exponentials as,
\beq\bsp
\bP(\bx,\eps) &\,= \bU(\bx,\eps)^{-1}\mathbb{P}_{\gamma}(\bx,\eps)\bP_0(\eps)\,,\\
\bs{\check{P}}(\bx,\eps) &\,= \bs{\check{U}}(\bx,\eps)^{-1}\check{\mathbb{P}}_{\gamma}(\bx,\eps)\bs{\check{P}}_0(\eps)\,.
\esp\eeq
Here $\bU(\bx,\eps)$ and $\bs{\check{U}}(\bx,\eps)$ are the rotation matrices that encode the transformation from the original basis of the twisted cohomology group to the canonical one, cf.~eqs.~\eqref{conventiontrafo} and~\eqref{Usplitting}. The matrices $\bP_0$ and $\bs{\check{P}}_0$ encode the initial conditions of the system of differential equations, given by the value of the period matrices at the initial point of the path $\gamma$. We may also think of these matrices as the rotations from the original choice of basis of the twisted homology group to a canonical basis of cycles where the period matrix is given by the path-ordered exponential. Consequently, the canonical homology intersection matrix is given by
\beq\bH_c(\eps) = \big(\bP_0(\eps)^{-1}\big){}^T\bH(\eps)\left(\bs{\check{P}}_0(\eps)\right)^{-1}\,.
\eeq
Our goal is to determine the explicit form of $\bH_c$.

We start by proving the following identity, which should be thought of as the analogue of the twisted Riemann bilinear relations~\eqref{eq:TRBRs}, but for a canonical choice of basis for the cohomology and homology groups such that the period matrix is given by the path-ordered exponential. We have
\beq\label{eq:canonical_TRBR}
\big(\bDelta\!^{-1}\big){}^T = {{\mathbb{P}}}_{\gamma}(\bx,\eps)^T\big(\bDelta\!^{-1}\big){}^T\check{\mathbb{P}}_{\gamma}(\bx,\eps)\,.
\eeq
Let us present the proof of eq.~\eqref{eq:canonical_TRBR}. Using eq.~\eqref{eq:A_dual_Delta}, we have
\beq\bsp\label{eq:can_TRBR_proof}
\check{\mathbb{P}}&_{\gamma}(\bx,\eps) = \mathbb{P}\exp\left[\eps\int_{\gamma}\bs{\check{A}}(\bx)\right]\\
&= \sum_{k=0}^{\infty}\sum_{1\le i_1,\ldots,i_k\le p}\eps^k\bs{\check{A}}_{i_1}\cdots\bs{\check{A}}_{i_k}\int_{\gamma}\omega_{i_1}\cdots\omega_{i_k}\\
&= \sum_{k=0}^{\infty}\sum_{1\le i_1,\ldots,i_k\le p}\eps^k\bDelta\!^T\bs{{A}}_{i_1}^T\cdots\bs{{A}}_{i_k}^T\big(\bDelta\!^{-1}\big){}^T\,(-1)^k\int_{\gamma}\omega_{i_1}\cdots\omega_{i_k}\\
&= \sum_{k=0}^{\infty}\sum_{1\le i_1,\ldots,i_k\le p}\eps^k\bDelta\!^T\big(\bs{{A}}_{i_k}\cdots\bs{{A}}_{i_1}\big)^T\big(\bDelta\!^{-1}\big){}^T\,(-1)^k\int_{\gamma}\omega_{i_1}\cdots\omega_{i_k}\\
&= \sum_{k=0}^{\infty}\sum_{1\le i_1,\ldots,i_k\le p}\eps^k\bDelta\!^T\big(\bs{{A}}_{i_k}\cdots\bs{{A}}_{i_1}\big)^T\big(\bDelta\!^{-1}\big){}^T\int_{\gamma^{-1}}\omega_{i_k}\cdots\omega_{i_1}\\
&=\bDelta\!^T{\mathbb{P}}_{\gamma^{-1}}(\bx,\eps)^T\big(\bDelta\!^{-1}\big){}^T\,,
\esp\eeq
where in the fourth equality we used the fact that for iterated integrals reversing the orientation of the path is equivalent to reversing the order of the differential forms, up to a sign~\cite{ChenSymbol},
\beq
\int_{\gamma^{-1}}\omega_{i_k}\cdots\omega_{i_1} = (-1)^k\int_{\gamma}\omega_{i_1}\cdots\omega_{i_k}\,.
\eeq
Inserting eq.~\eqref{eq:can_TRBR_proof} into the right-hand of eq.~\eqref{eq:canonical_TRBR}, and using the fact that
\beq
{\mathbb{P}}_{\gamma^{-1}}(\bx,\eps){\mathbb{P}}_{\gamma}(\bx,\eps) = \mathds{1}\,,
\eeq
we see that eq.~\eqref{eq:canonical_TRBR} follows.

We can now use eq.~\eqref{eq:canonical_TRBR} to obtain the explicit expression for the canonical homology intersection matrix $\bH_c$. Using eqs.~\eqref{eq:Cc_def} and~\eqref{eq:Cc_to_Delta}, we find,
\beq\bsp
\bH_c(\eps) &\,= \frac{1}{(2\pi i)^n}\,\mathbb{P}_{\gamma}(\bx,\eps)^T\big(\bU(\bx,\eps)^{-1}\big){}^T\big(\bC(\bx,\eps)^{-1}\big){}^T\bs{\check{U}}(\bx,\eps)^{-1}\check{\mathbb{P}}_{\gamma}(\bx,\eps)\\
&\,= \frac{1}{(2\pi i)^n}\,\mathbb{P}_{\gamma}(\bx,\eps)^T\big(\bC_c(\bx,\eps)^{-1}\big){}^T\check{\mathbb{P}}_{\gamma}(\bx,\eps)\\
&\,= \frac{1}{(2\pi i)^n}\,\mathbb{P}_{\gamma}(\bx,\eps)^T\big(\bDelta\!^{-1}\big){}^T\big(\bM(\eps)^{-1}\big){}^T\check{\mathbb{P}}_{\gamma}(\bx,\eps)\,.
\esp\eeq
From eq.~\eqref{eq:commutation}, we see that $\big(\bM(\eps)^{-1}\big){}^T$ and $\check{\mathbb{P}}_{\gamma}(\bx,\eps)$ commute, and using eq.~\eqref{eq:canonical_TRBR}, we get
\beq\bsp
\bH_c(\eps) 
&\,= \frac{1}{(2\pi i)^n}\big(\bDelta\!^{-1}\big){}^T\big(\bM(\eps)^{-1}\big){}^T = \frac{1}{(2\pi i)^n}\big(\bC_c(\eps)^{-1}\big){}^T\,.
\esp\eeq


\section{Constraining \texorpdfstring{$\eps$}{eps}-functions}
\label{sec:constraints}

In this section we present the main result of our work, namely we address the question of how to find polynomial relations among $\eps$-functions with coefficients in $\cF_\mathrm{ss}$. Throughout this section we follow the notations from the previous section, and we assume that we have determined a rotation $\bU(\bx,\eps)$ to a canonical basis and a corresponding canonical intersection matrix $\bDelta$. The entries of the rotation matrix $\bUt(\bx,\eps)$ are then the $\eps$-functions, obtained by solving a system of first-order differential equations. We also assume that we work in a self-dual scenario (e.g., we work on the maximal cut). Our main result can then be summarised as a set of polynomial relations that relate the $\eps$-functions. Note that all $\eps$-functions that appear on the maximal cut also contribute to the full Feynman integrals, so that our results also apply beyond the maximal cut. We will come back to non-maximal cuts in section~\ref{sec:non-maximal_cuts}.

\subsection{Constraining \texorpdfstring{$\eps$}{eps}-functions on the maximal cut}
We consider a self-dual system with $N$ master integrals, e.g., the system of differential equations satisfied by $N$ independent maximal cuts in a given sector. Assume that we have performed the rotations by the matrices $\bUss$ and $\bUeps$, and denote the differential equation matrix after these two rotations by $\tilde{\bs{\Omega}}$. In general the system will not be $\eps$-factorised at this point, but, as already mentioned, the highest power in $\eps$ appearing in the Laurent expansion of $\tilde{\bs{\Omega}}$ is 1. Hence, we can write
\begin{align}\label{deqbeforeUt}
\tilde{\bs{\Omega}}&=\big[\bU_\eps\bU_{\mathrm{\mathrm{ss}}}\bs{\Omega} +\rd(\bU_\eps\bU_{\mathrm{\mathrm{ss}}})\big]\bU_{\mathrm{\mathrm{ss}}}^{-1}\bU_\eps^{-1}=\sum_{k=k_\text{min}}^0\tilde{\bs{\Omega}}_k\eps^k+\eps \tilde{\bs{\Omega}}_1\,.
\end{align}
It is easy to see that we have
\beq\label{eq:A_to_Ut0}
\bA = \bUt^{(0)}\tilde{\bs{\Omega}}_1\bUt^{(0)-1}\,.
\eeq
The entries of $\bA$ are one-forms defined from the function field $\cFeps$, while those of $\tilde{\bs{\Omega}}_1$ are defined from the function field $\cF_\mathrm{ss}$. The non-trivial entries of $\bUt^{(0)}$ can be interpreted as a set of \emph{generators} for $\cFeps$, by which we mean that every element of $\cFeps$ can be written as a rational function in the elements of $\bUt^{(0)}$ and their derivatives with coefficients in $\cF_\mathrm{ss}$. 
We stress that this set of generators is by no means expected to be minimal, and there may be polynomial relations with coefficients in $\cF_\mathrm{ss}$ between these generators. In ref.~\cite{Duhr:2024uid} it was pointed out that eq.~\eqref{eq:A_dual_Delta}, which follows from self-duality, can be used to obtain polynomial relations between the generators (see also ref.~\cite{Pogel:2024sdi} for a related idea). Solving these relations allows one to reduce the number of generators needed to describe $\cFeps$. However, the constraints may be highly non-linear, and therefore it may be hard to obtain explicit solutions to them. One of the goals of this section is to identify another set of generators such that the constraints that follow from eq.~\eqref{eq:A_dual_Delta} can be reduced to linear relations.

We start by discussing the structure of the matrix $\bUt$ in more detail. After having performed the rotation with the matrix $\bUss$, we expect that the matrix $\bUt$ is unipotent.\footnote{We recall that a matrix $\bU$ is called \emph{unipotent} if $\bU-\mathds{1}$ is nilpotent, i.e., there is $k\ge 2$ such that $(\bU-\mathds{1})^k=0$.} This means that we can find a basis in which $\bUt$ is lower triangular with 1's on the diagonal. We denote by UT the  group of unipotent lower triangular $N\times N$ matrices. 
Empirically, we find that $\bUt$ has more structure. In applications, where the matrix $\bUt$ is determined by the method of refs.~\cite{Gorges:2023zgv,Duhr:2025lbz}, $\bUt$ contains additional zeroes below the diagonal. These zero entries are related to a good choice of initial basis, which, according to ref.~\cite{Duhr:2025lbz}, takes into account the mixed Hodge structure of the geometry attached to the maximal cut. We can phrase this more mathematically by saying that the method of refs.~\cite{Gorges:2023zgv,Duhr:2025lbz} determines a filtration $F_{\textrm{t},N}\subseteq F_{\textrm{t},N-1}\subseteq\ldots\subseteq F_{\textrm{t},0}$ on each sector, and this filtration captures the way the master integrals from the sector are allowed to mix under the rotation $\bUt$. The matrix $\bUt$ must therefore lie in the group $\Gpar$ of unipotent lower triangular matrices that preserve this filtration.\footnote{This is the so-called unipotent radical of the parabolic subgroup of UT attached to the filtration.}

Remarkably, we find that there is a connection between the group $\Gpar$ determined by the filtration and the canonical intersection matrix $\bDelta$.
To state this precisely, we consider the map $\varphi_{\Delta}$ defined in eq.~\eqref{eq:A_dual_Delta}. It is easy to check that, if $\bDelta$ is symmetric or antisymmetric (which is always satisfied in a  self-dual scenario), we have
\beq
\varphi_{\Delta}^2 = \textrm{id} \textrm{~~~and~~~} \varphi_{\Delta}(\bM_1\bM_2) = \varphi_{\Delta}(\bM_2)\varphi_{\Delta}(\bM_1)\,.
\eeq
We refer to the operation $\varphi_{\Delta}$ as the \emph{$\Delta$-transposition}.\footnote{The name refers to the fact that for $\bDelta=\mathds{1}$ we recover the standard transposition of matrices.}
The $\Delta$-transpose of a matrix from UT will typically no longer be an element from this group.
For us only the subgroup $\Gpar\subseteq \mathrm{UT}$ is relevant.
Based on the examples that we have studied, we find:
\begin{observation}\label{obs:Ut_in_UT}
The group $\Gpar$ is closed under $\Delta$-transposition, i.e., $\varphi_{\Delta}(\bU)\in\Gpar$, for all $\bU\in\Gpar$.
\end{observation}
In the following we assume that this property holds.
We stress that, while this property holds for all examples we have studied and we expect it to hold in general, we do not have a formal proof. However, this property is easy to check on a case by case basis.
Note that we can phrase Observation~\ref{obs:Ut_in_UT} in the following equivalent way. Let $\UT$ denote the largest subgroup of UT that is closed under $\Delta$-transposition,
\beq
\UT = \big\{\bM\in\mathrm{UT}: \varphi_{\Delta}(\bM)\in \mathrm{UT}\big\}\,.
\eeq
Then Observation~\ref{obs:Ut_in_UT} is equivalent to the statement that $\Gpar$ is a subgroup of $\UT$.

We now explain how we can combine Observation~\ref{obs:Ut_in_UT} with eq.~\eqref{eq:A_dual_Delta} to obtain a set of generators for $\cFeps$ that satisfy relations that can be reduced to solving linear constraints. The main mathematical tool to identify this set of generators is the following result:
\begin{proposition}\label{prop:main_1}
Let $\bU\in\UT$. Then there are unique matrices $\bO,\bR\in\UT$ such that
\beq\label{eq:decomp}
\bU=\bO\bR\,,
\eeq
with
\beq\label{eq:sym_orth}
\varphi_{\Delta}(\bO) = \bO^{-1} \textrm{~~~and~~~}\varphi_{\Delta}(\bR) = \bR\,.
\eeq
\end{proposition}
The proof of Proposition~\ref{prop:main_1} will be presented in section~\ref{sec:proof}. 
Since Observation~\ref{obs:Ut_in_UT} implies that $\Gpar$ is a subgroup of $\UT$, the decomposition from Proposition~\ref{prop:main_1} also applies to $\Gpar$. We refer to matrices that satisfy the properties in eq.~\eqref{eq:sym_orth} as \emph{$\Delta$-orthogonal} and \emph{$\Delta$-symmetric}, respectively. It is easy to see that {$\Delta$-orthogonal} matrices form a subgroup of $\UT$, and this subgroup can be described as those matrices from $\UT$ that preserve the scalar product with Gram matrix $\bDelta\!^{-1}$ defined by the (inverse of the) canonical intersection matrix. We define,
\begin{align}
\OT&:=\UT\cap \textrm{O}\big(\Delta\!^{-1}\big) = \left\{\bM\in \UT:\varphi_{{\Delta}}(\bM)=\bM ^{-1}\right\}\,,
\end{align}
where $\textrm{O}\big(\Delta\!^{-1}\big)$ is the orthogonal group of $\bDelta\!^{-1}$,
\beq\label{eq:Orthgroup}
\textrm{O}\big(\Delta\!^{-1}\big) := \left\{\bM: \bM^T\bDelta\!^{-1}\bM=\bDelta\!^{-1}\right\}\,.
\eeq
We also define the set of $\Delta$-symmetric matrices,
\begin{align}
\SUT&:=\left\{\bM\in \UT:\varphi_{{\Delta}}(\bM)=\bM \right\}\,.
\end{align}
Note that, unlike $\OT$, $\SUT$ is not a subgroup of $\UT$. Finally, we note that the splitting in eq.~\eqref{eq:decomp} depends on the form of the canonical intersection matrix $\bDelta$, which itself depends on the precise choice of canonical basis. The group $\Gpar$, however, is independent of the precise choice of canonical basis,  and so for each choice of basis and $\bDelta$, we obtain a different splitting. We will comment on the way the splitting depends on the basis choice in section~\ref{sec:consequences}.

 We now discuss how we can apply the decomposition from Proposition~\ref{prop:main_1} to constrain $\eps$-functions. Before we do that, let us explain how we can compute the decomposition in eq.~\eqref{eq:decomp} in practice. Using eqs.~\eqref{eq:decomp} and~\eqref{eq:sym_orth} we immediately see that
\beq
\bR^2 = \varphi_\Delta(\bU)\bU\,.
\eeq
Thus, $\bR$ is a matrix square root of $\varphi_\Delta(\bU)\bU$, and it can be a complicated task to find all square roots. However, we know that both $\bR$ and $\varphi_\Delta(\bU)\bU$ are unipotent lower-triangular matrices, and every unipotent lower-triangular matrix has a unique square root that is itself unipotent and lower-triangular. We denote this square root in the following by $\sqrt{\varphi_\Delta(\bU)\bU}$. It can be computed in an algorithmic way by solving linear equations inductively in the number of rows and the number of lower diagonals. Hence, we see that we can easily compute $\bR$ as
\beq
\bR = \sqrt{\varphi_\Delta(\bU)\bU}\,.
\eeq
The $\Delta$-orthogonal matrix $\bO$ can then be obtained as
\beq
\bO = \bU\Big[\sqrt{\varphi_\Delta(\bU)\bU}\Big]^{-1}\,.
\eeq

Let us discuss the implications of Proposition~\ref{prop:main_1} for constraints on the $\eps$-functions. First, if Observation~\ref{obs:Ut_in_UT} holds, we may apply Proposition~\ref{prop:main_1} to $\bUt^{(0)}\in \Gpar$ to write
\beq\label{eq:Ut0_decomp}
\bUt^{(0)} = \bO\bR\,,
\eeq
where $\bO$ and $\bR$ are $\Delta$-orthogonal and $\Delta$-symmetric respectively.
We now claim that the entries of the $\Delta$-symmetric part $\bR$ actually lie in $\cF_\mathrm{ss}$. To see this, we insert the decomposition in eq.~\eqref{eq:Ut0_decomp} into eqs.~\eqref{eq:A_phiA} and~\eqref{eq:A_to_Ut0} to obtain
\begin{equation}
\bs{\Delta}\big(\bs{O}^{-1}\big)^T\big(\bR^{-1}\big)^T\Big(\tilde{\bs{\Omega}}_1^T\bs{R}^T\bs{O}^{T}\Big)\bs{\Delta}^{-1}
=\bs{O}\bs{R}\tilde{\bs{\Omega}}_1\bs{R}^{-1}\bs{O}^{-1}\notag\,.
\end{equation}
Using eq.~\eqref{eq:sym_orth}, we find
\begin{align}\label{constrantR}
\bs{R}^2 \tilde{\bs{\Omega}}_1=\varphi_{{\Delta}}(\tilde{\bs{\Omega}}_1)\bs{R}^2\,.
\end{align}
Let us interpret this equation. We see that $\bs{O}$ has dropped out, and we know that the entries of $\tilde{\bs{\Omega}}_1$ are one-forms defined from elements in $\cF_\mathrm{ss}$. We may thus interpret eq.~\eqref{constrantR} as a linear system for the entries of $\bR^2$, which in turn are polynomials in $\eps$-functions.
In particular, the matrix $\bR$ is fully fixed\footnote{We assume that the system in eq.~\eqref{constrantR} is of full rank, which is true in all examples we have studied. Note that, unlike what it may seem, the linear system for the entries of $\bR^2$ is not homogeneous, because $\bR^2$ is unipotent and lower-triangular.} in terms of the entries of $\tilde{\bs{\Omega}}_1$, i.e., (derivatives of) periods and algebraic functions. Once $\bR^2$ has been determined by solving the linear system in eq.~\eqref{constrantR}, we can easily compute its unipotent lower-triangular square root. Since this step also only involves solving linear constraints, we see that, using the decomposition from Proposition~\ref{prop:main_1}, we can solve the constraints from eq.~\eqref{eq:A_phiA} by solving only linear constraints.

\subsection{The proof of Proposition~\ref{prop:main_1}}
\label{sec:proof}

In this subsection we present the proof of Proposition~\ref{prop:main_1}. The content of this section does not interfere with the remainder of the paper, and the reader may skip it if he or she is not interested in the details of the proof.

We start by noting that the group UT of all unipotent lower-triangular $N\times N$ matrices is a Lie group. Its Lie algebra is the Lie algebra NT of all nilpotent, strictly lower-triangular $N\times N$ matrices. In other words, the elements of NT have the form
\beq
\left(\begin{array}{ccc}
    0 &\cdots  & 0 \\ 
    *    & \ddots\, \, &\vdots \, \, \\ 
*&*&0
    \end{array}\right)\,.
    \eeq
Similar to the case of UT, NT is in general not closed under $\Delta$-transposition. We define $\NT$ as the largest subspace of NT that is closed,
\beq
\NT :=\left\{\bX\in \textrm{NT}:\varphi_{\Delta}(\bX)\in\textrm{NT}\right\}\,.
\eeq
It is easy to see that $\NT$ is a Lie subalgebra of NT. It is in fact the Lie algebra of the Lie group $\UT$. The exponential map
\beq
\exp : \NT \to \UT; \quad \bX \mapsto \exp(\bX)\,.
\eeq
is surjective (because both the series expansions for the exponential and logarithm maps terminate for nilpotent matrices).

The $\Delta$-transposition acts linearly on $\NT$. Since $\varphi_{\Delta}$ is involutive, its eigenvalues are $\pm1$, and we can decompose $\NT$ into even and odd eigenspaces,
\begin{align}\label{LieDecomp}
\NT=\Np\oplus \Nm\,,
\end{align}
with
\begin{align}
\text{NT}^{{\Delta}\pm}:=\left\{\bX\in \NT\,|\,\varphi_{\Delta}(\bX)=\pm \bX\right\}\,.
\end{align}
Note that $\Nm$ is a Lie subalgebra of $\NT$, but $\Np$ is not closed under the Lie bracket.
From the injectivity of the exponential map it follows that the $\Delta$-orthogonal and $\Delta$-symmetric matrices can be written as the exponential of Lie algebra elements in $\Nm$  or $\Np$, respectively,
\begin{align} \label{expmap}
\OT&=\exp(\Nm)\,\quad \text{ and }\quad \SUT=\exp(\Np)\,.
\end{align}
Proposition~\ref{prop:main_1} states that every element from $\UT$ can be decomposed in a unique way into a product of matrices from $\OT$ and $\SUT$. In essence, at the level of Lie algebras, this corresponds to the decomposition into the even and odd eigenspaces in eq.~\eqref{LieDecomp}. However, due to the non-commuting nature of the matrices, we cannot immediately lift this statement to the level of the Lie groups via the exponential map, but we need some further prerequisites.
\begin{lemma}
\begin{enumerate}
    \item 
\textit{For every $\bM\in \mathrm{UT}^{\Delta}$  there exists an $\bs{X}\in \mathrm{NT}^{\Delta+}$ such that }
\begin{align}\label{req1}
\varphi_{\Delta}(\bM)\bM=\exp(\bs{X})\,.
\end{align}
\item \textit{The intersection of the sets of $\Delta$-symmetric and $\Delta$-orthogonal matrices is the identity}
\begin{align}\label{req2}
\mathrm{ST}^{\Delta}\cap\mathrm{OT}^{\Delta}=\{\mathds{1}\}\,.
\end{align}
\item \textit{If $\bM \in \mathrm{ST}^{\Delta}$ and $\bs{O} \in \mathrm{OT}^{\Delta}$ then}
\begin{align}\label{req3}
\bs{O}\bM\in \mathrm{ST}^{\Delta} \textrm{~~~iff.~~~} \bs{O}=\mathds{1}\,.
\end{align}
\end{enumerate}
\end{lemma}
\begin{proof} Equation~\eqref{req1} follows from the observation that $\varphi_{\Delta}(\bM)\bM\in \SUT$. From eq.~\eqref{expmap} it follows that there is $\bX\in \Np$ such that $\exp(\bX)=\varphi_{\Delta}(\bM)\bM$.

To show eq.~\eqref{req2},  we observe that if $\bs{M}\in \SUT\cap\OT$, it follows from eq.~\eqref{expmap} that 
$\bs{M}=\exp(\bs{X})$
with $\bs{X}\in\Np\cap\Nm=\{\bs{0}\}$,
such that $\bs{M}=\mathds{1}$.

To prove eq.~\eqref{req3}, we notice that, if $\bs{O}\bM\in \SUT$, then
\begin{align}
\bs{O}\bM=\varphi_{\Delta}(\bs{O}\bM)=\varphi_{\Delta}(\bM)\varphi_{\Delta}(\bs{O})=\bM \bs{O}^{-1}\,,
\end{align}
such that $(\bs{O}\bM)^2=\bM^2$. Since $\bM$ and $\bO$ are unipotent, and since the unipotent square root is unique, we conclude that $\bs{O}\bM=\bM$ and $\bs{O}=\mathds{1}.$
\end{proof}

The pivotal step in proving Proposition~\ref{prop:main_1} is the following statement.
\begin{proposition}\label{prop:main_2} {The map $f$ defined by}
\begin{align}\label{bijection}
f:\mathrm{OT}^{\Delta} \times \mathrm{NT}^{\Delta_+}& \rightarrow \mathrm{UT}^{\Delta}, \quad  (\bs{O},\bs{N})\mapsto\bs{O}\exp(\bs{N}) \,,
\end{align}
{is a bijection.} 
\end{proposition}

\begin{proof} We start by showing the surjectivity of $f$. Take an element $\bM\in \UT$. According to eq.~\eqref{req1}, there exists a corresponding $\bX\in \Np$ with $\varphi_{\Delta}(\bM)\bM=\exp(\bX)$. We define $\bs{R}\in \SUT$ with $\bs{R}=\exp(\bX/2)$ and a matrix $\bs{\bs{O}}\in \UT$ as $\bs{O}=\bM \bs{R}^{-1}$.
We have
\begin{align}
\varphi_{\Delta}(\bs{O})\bs{O}&=\varphi_{\Delta}(\bs{R}^{-1})\varphi_{\Delta}(\bM)\bM \bs{R}^{-1}=\exp(-\bX/2)\exp(\bX)\exp(-\bX/2)=\mathds{1}\,.
\end{align}
It follows that $\bO$ is $\Delta$-orthogonal,
i.e., $\bs{O}\in \OT$. We conclude that for every element $\bM \in \UT$ there exists a pair $(\bs{O},\bs{N})\in\OT \times \Np$ with $\bO\in\OT$ and $\bs{R}=\exp(\bs{N})\in\SUT$  such that $\bM=\bs{O}\bR$.

Let us now prove injectivity of $f$.
Consider two pairs $(\bs{O}_1,\bs{N}_1)$ and $(\bs{O}_2,\bs{N}_2)$ and assume that $f(\bs{O}_1,\bs{N}_1)=f(\bs{O}_2,\bs{N}_2)$. It follows that
\begin{align}
 \bs{O}_1^{-1}\bs{O}_2 \exp(\bs{N}_2)&=\exp(\bs{N}_1)\in\SUT\,.\notag
\end{align}
From $\bs{O}_1^{-1}\bs{O}_2\in\OT$ and eq.~\eqref{req3} we conclude $\bs{O}_1^{-1}\bs{O}_2 =\mathds{1}$ such that $\bs{O}_1=\bs{O}_2$ and $\exp(\bs{N}_1)=\exp(\bs{N}_2)$, i.e., $\bs{N}_1=\bs{N}_2$. Accordingly, $f$ is a bijection.
\end{proof}

Proposition~\ref{prop:main_1} is now a simple corollary of Proposition~\ref{prop:main_2}. More precisely, the decomposition from Proposition~\ref{prop:main_1} is equivalent to the surjectivity of the map $f$. Indeed, let $\bU\in\UT$. Then by the surjectivity of $f$, there is $(\bs{O},\bs{N})\in\OT \times \Np$ such that 
\beq
\bU = f(\bs{O},\bs{N}) = \bO\bR\,,
\eeq
with $\bR = \exp(\bN)$. The injectivity implies the uniqueness of the decomposition.

\subsection{Some consequences}
\label{sec:consequences}
We present in this subsection some consequences of the decomposition in Proposition~\ref{prop:main_1}.

\paragraph{Coset decompositions for $\UT$.}
 Proposition~\ref{prop:main_2} can also be interpreted as stating that there is a bijection between the right-cosets of $\OT$ in $\UT$ and the set $\Np$ of $\Delta$-symmetric nilpotent matrices. Since $\Np$ maps to $\SUT$ under the exponential map (cf.~eq.~\eqref{expmap}), we get the following explicit description of the right-coset space:
\beq
\OT\Big\backslash\UT \simeq \SUT\,.
\eeq
We have already observed in the previous subsection that $\SUT$ is not a subgroup of $\UT$, and consequently $\OT$ is not a normal subgroup of $\UT$. Hence, the left- and right-cosets are not identical. It is, however, also possible to describe the left-cosets explicitly. Indeed, we have
\beq
\bU = \bO\bR = \bO\bR\bO^{-1}\bO\,,
\eeq
and
\beq
\varphi_{\Delta}\big(\bO\bR\bO^{-1}\big) = \varphi_{\Delta}\big(\bO^{-1}\big)\varphi_{\Delta}(\bR)\varphi_{\Delta}(\bO) = \bO\bR\bO^{-1}\,.
\eeq
Hence, there is $\bO\in\OT$ and $\bR'=\bO\bR\bO^{-1}\in\SUT$ such that $\bU=\bR'\bO$. Using the same arguments as in the proof of Proposition~\ref{prop:main_1}, we can see that this decomposition
is unique, and we find a description of the left-cosets in terms of $\Delta$-symmetric matrices,
\beq
\UT\Big/\OT \simeq \SUT\,.
\eeq

\paragraph{An upper bound on the number of generators of $\cFeps$.}

Via Proposition~\ref{prop:main_1}, rather than choosing the non-trivial entries of $\bUt^{(0)}$ as the generators of $\cFeps$, we may choose a set of generators obtained by parametrising the $\Delta$-orthogonal and $\Delta$-symmetric matrices $\bO$ and $\bR$. From the previous argument, we know that the entries of $\bR$ are actually in $\cF_\mathrm{ss}$, and so they can be removed from our set of generators. Since $\bO$ drops out from eq.~\eqref{constrantR}, we cannot constrain its entries to lie in $\cF_\mathrm{ss}$.
Hence, the functions obtained by parametrising the $\Delta$-orthogonal matrix $\bO$ (and its derivatives) are sufficient to generate the whole function field $\cFeps$.

 It is natural to ask if there may be other constraints that allow one to express the $\eps$-functions in $\bO$ in terms of elements from $\cF_\mathrm{ss}$. While we do not have a definite answer to this question, the examples that we have studied suggest that the entries $\bO$ define genuine $\eps$-functions that cannot be expressed in term of functions from $\cF_\mathrm{ss}$. Either way, we see that the minimal number of generators of $\cFeps$ is bounded by the number of degrees of freedom in the matrix $\bO$. Equivalently, we may summarise this by the statement:
\begin{proposition}
The dimension of the Lie group $\mathrm{OT}^{\Delta}$ is an upper bound on the minimal number of generators of $\cFeps$.
\end{proposition}

\paragraph{Implications for the rotation to a canonical basis.}

Let us discuss an interpretation of the decomposition from Proposition~\ref{prop:main_1} in the context of performing the rotation to a canonical basis via the algorithm of ref.~\cite{Gorges:2023zgv}. We use the notations from section~\ref{sec:conventions}. In the following, it is convenient to rewrite the last rotation  as a sequence of rotations,
\beq
\bUt=\bUt^{(0)}\bU_{\!\mathrm{m}}\,.
\eeq
If we transform the system with the matrix $\bU_{\!\mathrm{X}} := \bU_{\!\mathrm{m}}\bUeps\bUss$, the differential equation matrix becomes,
\begin{align}\label{AltdeqbeforeUt}
\widehat{\bOmega}=\big(\bU_{\!\mathrm{X}}\bOmega+\rd\bU_{\!\mathrm{X}}\big)\bU_{\!\mathrm{X}}^{-1}=\eps\tilde{\bOmega}_1+\widehat{\bOmega}_{0}\,,
\end{align}
where $\tilde{\bOmega}_1$ is the same matrix as in eq.~\eqref{eq:A_to_Ut0}. Note that after the rotation by $\bU_{\!X}$, the differential equation matrix is free of poles in $\eps$.
Prior to performing this last rotation, the cohomology intersection matrix takes the form 
\beq
 \bs{\widehat{C}}(\bx,\eps):=\bU_{\!\mathrm{X}}(\bx,\eps)\bC(\bx,\eps) \bU_{\!\mathrm{X}}(\bx,-\eps)^T\,.
 \eeq
 
In order to achieve complete $\eps$-factorisation, we still need to perform the rotation $\bUt^{(0)}$. 
After the final rotation by $\bUt^{(0)}$, we reach the canonical form, and we have
\begin{align}\label{eq:lastrotC}
\bUt^{(0)}(\bx)\bs{\widehat{C}}(\bx,\eps)\bUt^{(0)}(\bx){}^T=f(\eps)\,\bDelta\,.
\end{align}
We may now apply Proposition~\ref{prop:main_1} to write
\beq
\bUt^{(0)}(\bx) = \bO(\bx)\bR(\bx)\,.
\eeq
 Since $\bO(\bx)$ is in the orthogonal group of $\bDelta$ (cf.~eq.~\eqref{eq:Orthgroup}), we see that  already the first rotation $\bR\in \mathcal{F}_{\mathrm{ss}}$  brings the intersection matrix into its canonical form,  
\begin{align}\label{eq:C_to_R}
\bR(\bx) \,\bs{\widehat{C}}(\bx,\eps)\bR(\bx)^T=f(\eps)\,\bDelta \,,
\end{align} 
while $\bO(\bx)$ has no effect on the intersection matrix, and only serves to $\eps$-factorise the differential equation. We therefore find the following  factorisation of the last rotation,
\beq
\bUt(\bx,\eps) = \bO(\bx)\bR(\bx)\bU_{\!\mathrm{m}}(\bx,\eps)\,,
\eeq
and each factor has a very clear interpretation: the rotation $\bU_{\!\mathrm{m}}(\bx,\eps)$ removes the poles in the differential equation matrix, the rotation $\bR(\bx)$ only involves functions from $\cFss$ and rotates the cohomology intersection matrix to its canonical form $\bDelta$, and the matrix $\bO(\bx)$ brings the differential equation into an $\eps$-factorised form and introduces $\eps$-functions into the differential equation matrix.

\paragraph{Parametrising the canonical intersection matrix.} We have already seen that the precise form of the canonical intersection matrix $\bDelta$ depends on the choice of the canonical basis. In this section we discuss how the form of $\bDelta$ changes with the canonical basis choice. 

As a first step, we need to discuss the freedom in defining the canonical basis within the framework of the method of refs.~\cite{Gorges:2023zgv,Duhr:2025lbz}. We have already seen that the method of refs.~\cite{Gorges:2023zgv,Duhr:2025lbz} endows the space of master integrals with a filtration $F_{\mathrm{t},p}$, which captures how the master integrals can mix under the rotation to the canonical basis. We can understand this freedom via the following reasoning: we know that the last rotation $\bUt$ must preserve the filtration $F_{\mathrm{t},p}$, and its shape is constrained so that it lies in the group $\Gpar$. In particular, we may ask, what freedom we have in defining the rotation $\bUt^{(0)}$, after we have fixed the matrix $\bU_{\!\mathrm{m}}$. It is easy to see that $\bUt^{(0)}$ is fixed by the differential equation
\beq\label{eq:Ut0_DEQ}
\rd\bUt^{(0)} = -\bUt^{(0)}\widehat{\bOmega}_0\,,
\eeq
where $\widehat{\bOmega}_0$ 
was defined in
eq.~\eqref{AltdeqbeforeUt}. It follows that, once $\bU_{\!\mathrm{m}}$ has been fixed, any solution to eq.~\eqref{eq:Ut0_DEQ} takes the form $\bM\bUt^{(0)}$, where $\bUt^{(0)}$ now denotes a particular solution to eq.~\eqref{eq:Ut0_DEQ} and $\bM$ is a constant matrix from $\Gpar$. In other words, all other canonical bases constructed in this way will only differ by a constant rotation that respects the filtration $F_{\mathrm{t},p}$.

Let us now discuss how the canonical intersection matrix depends on the choice of basis, i.e., how it changes with $\bM$. We assume that we have determined the canonical intersection matrix $\bDelta$ corresponding to the particular solution $\bUt^{(0)}$, cf.~eq.~\eqref{eq:lastrotC}. Under a constant rotation $\bM\in\Gpar$, the canonical intersection matrix changes to
\beq\label{eq:Delta'}
\bDelta' = \bM\bDelta \bM^T\,.
\eeq
Note that Observation~\ref{obs:Ut_in_UT} also holds for this other choice of canonical intersection matrix. Since $\bM\in\Gpar$, we can write it in the form
\beq
\bM = \bM_{\!R}\bM_O\,,
\eeq
with $\bM_{\!R}\in \SUT$ and $\bM_O\in\OT$. Since $\OT$ can be identified with the orthogonal group of $\Delta$, the $\Delta$-orthogonal part drops out from eq.~\eqref{eq:Delta'}, and we have
\beq\label{eq:Delta'_to_Delta}
\bDelta' = \bM_{\!R}\bDelta \bM_{\!R}^T\,.
\eeq
We thus see that the different choices for the canonical intersection matrix are parametrised by the set $\SUT$. Based on our explicit computations, we find that: 
\begin{observation}\label{obs:anti-diag}
    There is a basis in which the canonical intersection matrix is block anti-diagonal, with the blocks being dictated by the filtration $F_{\mathrm{t},p}$. 
\end{observation}

\paragraph{An analytic method to determine $\bDelta$.}
We now argue that, by combining the properties from this subsection, we can often determine the form of $\bDelta$ without having to compute any intersection numbers. In a nutshell, the idea is that  $\Gpar$ and $\bUt^{(0)}$ a priori do not rely on the form of $\bDelta$, but nevertheless they know about it through the decomposition from Proposition~\ref{prop:main_1}. 

Using the notation introduced earlier in this section, we can see from eqs.~\eqref{eq:A_phiA},~\eqref{eq:A_to_Ut0} and~\eqref{eq:lastrotC} that the intersection matrix before the last rotation (cf.,~eq.~\eqref{AltdeqbeforeUt}) satisfies the relation
\begin{align}\label{eq:Omega1CTilde1}
\tilde{\bOmega}_1\bs{\widehat{C}}=\bs{\widehat{C}}\tilde{\bOmega}_1^T\,.
\end{align}
Assume now that we have determined the matrix $\tilde\bOmega_1$ by rotating the system with the matrix $\bU_{\!\mathrm{X}}$, but we did not compute the intersection matrix, so that $\bs{\widehat{C}}$ is unknown. Then we may interpret eq.~\eqref{eq:Omega1CTilde1} as a linear system for $\bs{\widehat{C}}$. Note that the entries of $\tilde\bOmega_1$ lie in $\cFss$, and so the entries of $\bs{\widehat{C}}$ take values in the same space. 
The system in eq.~\eqref{eq:Omega1CTilde1} does not yet determine $\bs{\widehat{C}}$ uniquely. We may further constrain it by imposing that $\bs{\widehat{C}}$ is (anti-)symmetric and exploiting that 
\beq\label{eq:Ctilde_Gpar}
\varphi_{\bs{\widehat{C}}}\Big(\bUt^{(0)}\Big)\in\Gpar\,.
\eeq
Equation~\eqref{eq:Ctilde_Gpar} follows from eq.~\eqref{eq:lastrotC}, because
\begin{align}
\varphi_{\bs{\widehat{C}}}\Big(\bUt^{(0)}\Big)=\bUt^{(0)\,-1}\varphi_{\bDelta} \Big(\bUt^{(0)}\Big) \bUt^{(0)}\,,
\end{align}
and Observation~\ref{obs:Ut_in_UT} implies that the right-hand side lies in $\Gpar$.

Once the (typically very simple) expression for $\bs{\widehat{C}}$ has been found, we may perform the last rotation, which transforms $\bs{\widehat{C}}$ into the canonical intersection matrix, cf.~eq.~\eqref{eq:lastrotC}. We see that only the $\Delta$-symmetric part plays a role in this last step, cf.~eq.~\eqref{eq:C_to_R}, and we know that these relations can be solved algebraically over $\cFss$.
Since the $\eps$-functions are only defined up to a constant (because they are defined as solutions to a system of first-order differential equations), all entries containing linearly independent combinations of $\eps$-functions can be set to zero, yielding the desired relations. After imposing these conditions, we obtain the final, constant expression for $\bDelta$. 
The advantage of this approach is that it allows us to extract in one go the canonical intersection matrix $\bDelta$ and the algebraic relations that define the entries of $\bR$.
We illustrate this strategy on an explicit example in section~\ref{app.moref43}.

\section{Examples}
\label{sec:examples}

In this section we provide a set of examples to illustrate the general results of sections~\ref{sec:can_int_matrix} and~\ref{sec:constraints}. In particular, we consider one- and multi-parameter families of integrals related to Calabi-Yau varieties in subsections \ref{CY1V} and \ref{fourlooptwomassbaanna} and families of integrals related to higher-genus Riemann surfaces in subsection \ref{subsec:higherGenus}. In many applications appearing in physics, Feynman integrals are related to more complicated geometries. We consider such an example in subsection \ref{fourlooptwomassbaanna}, see also ref.~\cite{Duhr:2025kkq}. 

\subsection{Deformed Calabi-Yau operators}
\label{CY1V}
We start by discussing a class of systems of differential equations in one variable obtained as deformations of so-called \emph{Calabi-Yau (CY) operators}~\cite{Almkvist2,Bogner:2013kvr,BognerThesis,CYoperators}. More precisely, we consider systems defined in the following way: it is well known that an $N\times N$ linear system of ordinary differential equations is equivalent to finding the kernel of a linear differential operator of order $N$. Following ref.~\cite{Duhr:2025lbz}, we consider $N\times N$ systems that lead to a differential operator $\mathcal{L}_{\eps}$ such that $\mathcal{L}_{\eps=0}$ is a CY operator according to the definition of ref.~\cite{Almkvist2,Bogner:2013kvr,BognerThesis,CYoperators}. In particular, this implies that $\mathcal{L}_{\eps=0}$ is (essentially) self-adjoint and has a point of maximal unipotent monodromy (MUM) at the origin. CY operators of order $N$ are interesting, because they correspond to the Picard-Fuchs operators of large classes of one-parameter families of CY $(N-1)$-folds, including families of CY varieties associated to (the maximal cuts of) relevant multi-loop Feynman integrals~\cite{Bonisch:2020qmm,Bonisch:2021yfw,Pogel:2022vat,Pogel:2022ken,Pogel:2022yat,Frellesvig:2024rea,Duhr:2022pch,Duhr:2023eld,Duhr:2024hjf,Frellesvig:2023bbf,Klemm:2024wtd,Driesse:2024feo,Frellesvig:2024zph,Brammer:2025rqo}. They also cover certain classes of interesting hypergeometric functions~\cite{Duhr:2025lbz}.

We will not give an extensive review of (deformed) CY operators. A detailed discussion on how to obtain a canonical basis for these cases can be found in ref.~\cite{Duhr:2025lbz} (see also refs.~\cite{Pogel:2022vat,Pogel:2022ken,Pogel:2022yat} for a discussion of the equal-mass bananas using a different, but equivalent, method). Here we only highlight that, even though $\mathcal{L}_{\eps=0}$ is (by definition) self-dual, the same does not need to be true for $\mathcal{L}_{\eps}$~\cite{Duhr:2025lbz}. In the following we restrict the discussion to cases where $\mathcal{L}_{\eps}$ leads to a self-dual scenario. We will come back to non self-dual cases in section~\ref{subsec:CY_non_self-dual}. Based on the results for the canonical forms for the examples in ref.~\cite{Duhr:2025lbz}, we can distill the general form of the canonical intersection matrix for deformed CY operators. For self-dual scenarios we can also determine the number of generators of $\cFeps$. In the remainder of this section we present those results, and we illustrate them by discussing in detail the cases of two deformed CY operators of order three and four, respectively.

\subsubsection{Generalities}
\label{sec:deformed_CY}
In ref.~\cite{Duhr:2025lbz} an extensive list of deformed CY operators was studied and the corresponding systems were transformed into canonical form. Based on these results, we conjecture that for the cases of a deformed CY operator of order $N$, the  $\eps^0$ part of the final rotation to the canonical form takes the form
\begin{equation}\label{eq:Ut0_CY_op}
    \bUt^{(0)}=\begin{pmatrix}
    
         1 \\
         t_{2,1} & 1 \\
         \vdots & &\ddots \\
        t_{N,1} & \dots & t_{N,N-1}  &1
        
    \end{pmatrix}\,.
\end{equation}
From the shape of this matrix we can see that in this case the group $\Gpar$ is the whole group UT of lower-triangular unipotent matrices. We can also read off the filtration $F_{\textrm{t},p}$. All the graded quotients $F_{\mathrm{t},p-1}/F_{\mathrm{t},p}$ are one-dimensional, and we can identify this filtration with the Hodge filtration on the middle cohomology of the CY variety.
There is a basis in which the canonical intersection matrix takes the form
\begin{equation}
\label{deltaCY1}
    \bDelta=\bK_{N}\, ,
\end{equation}
where $\bK_N$ is the $N\times N$ matrix
\begin{equation}\label{eq:K_matrix}
\bK_N = \bK_N^T=\bK_N^{-1}= \left(\begin{smallmatrix}
0 & 0 & \ldots & 0& 1\\
0 & 0 &  & 1& 0\\
\vdots &  & {.^{.^{.^{.^.}}}} & & \vdots\\
\phantom{.}&&&&\\
0 & 1 & \ldots & 0& 0\\
1 & 0 &  & 0& 0\\
\end{smallmatrix}\right) \,.
\end{equation}
Note that this matrix is anti-diagonal, in agreement with Observation~\ref{obs:anti-diag}.
The $\frac{N(N-1)}{2}$ non-trivial entries $t_{ij}$ of $\bUt^{(0)}$ are a set of generators for the function field $\cFeps$.
The map $\varphi_{\Delta} = \varphi_{\bK_N}$ corresponds to reflecting the elements of an $N\times N$ matrix over the anti-diagonal, and the shape of $\bUt^{(0)}$ is clearly preserved by this operation, i.e., $\bUt^{(0)}\in\Gpar$, in agreement with Observation~\ref{obs:Ut_in_UT}.

So far all considerations hold for general deformed CY operators, independently of the assumption of self-duality. For the rest of this section, we focus on cases where the deformed CY operator is self-dual for generic values of $\eps$. The case of non-self-dual deformed CY operators will be discussed in section~\ref{subsec:CY_non_self-dual}. In a self-dual scenario
we can apply Proposition~\ref{prop:main_1} to write $\bUt^{(0)}=\bO\bR$, where $\bO\in\textrm{OT}^{\bK_N}$ and $\bR\in\textrm{ST}^{\bK_N}$. Note that $\textrm{ST}^{\bK_N}$ is precisely the set of unipotent, lower-triangular matrices that are persymmetric, and so $\bR$ is persymmetric. The dimension of $\textrm{ST}^{\bK_N}$ is given by
\beq
\dim \textrm{ST}^{\bK_N} = \begin{cases} 
\frac{1}{4}N^2, & \text{if } N \text{ even}\,, \\
\frac{1}{4}(N^2 - 1), & \text{if } N \text{ odd}\,.
\end{cases}
\eeq
It follows that
\beq
\label{eq:dim_OTKN}\dim\textrm{OT}^{\bK_N} = 
\begin{cases} 
\frac{1}{4}N(N-2), & \text{if } N \text{ even}\,, \\
\frac{1}{4}(N-1)^2, & \text{if } N \text{ odd}\,.
\end{cases}
\eeq
Equation~\eqref{eq:dim_OTKN} provides an upper bound on the number of generators for $\cFeps$, i.e., on the number of genuine $\eps$-functions that we need to introduce to obtain a system in canonical form. Let us make a few comments. First, we observe that $\dim\textrm{OT}^{\bK_2}=0$, and so $\cFeps=\cFss$ for $N=2$, in agreement with the analysis in ref.~\cite{Duhr:2024uid}. Second, we may ask when the upper bound is saturated. While we do not have a definite answer to this question, we expect that the bound is always saturated, and that the minimal number of generators for $\cFeps$ is given by eq.~\eqref{eq:dim_OTKN}. We will see examples of this below, and it also precisely matches the number of independent functions introduced in refs.~\cite{Pogel:2022yat,Pogel:2022ken,Pogel:2022vat} to $\eps$-factorise the equal-mass banana integrals. 

So far all considerations were generic and apply to arbitrary self-dual, deformed CY operators. In the remainder of this subsection, we discuss two examples where the solutions are hypergeometric functions, which for $\eps=0$ compute the periods of a family of K3 surfaces and CY threefolds, respectively.

\subsubsection{Hypergeometric \texorpdfstring{${}_3F_2$}{3F2} 
functions associated with a one-parameter family of K3 surfaces}\label{hypergeometricexample}
\newcommand{\lt}{{\scriptscriptstyle(3F2)}} 
 Let us illustrate the previous discussion on the normalised version of the hypergeometric ${}_3F_2$ function:\footnote{ This version differs from the standard definition by a factor $ (\Gamma(1+m_1\eps)\Gamma(1+m_2\eps))/(\Gamma\left(\frac{1}{2}+n_2\eps\right)\Gamma\left(\frac{1}{2}+n_3\eps\right)\Gamma\left((m_1-n_2)\eps-\frac{1}{2}\right)\Gamma\left((m_2-n_3)\eps-\frac{1}{2}\right))$.} 
\begin{align}
\label{eq:3F2_def}
&{}_3\mathcal{F}_2\left(\bs{\alpha}\, ; \, \bs{\beta}\, ;\lambda\right)=\int_{D} \Phi_{3F2}(x_1,x_2)\, \frac{\rd x_1\wedge \rd x_2}{(1-x_1)(x_1-x_2)x_2}\, . 
\end{align}
Here $D$ is some cycle and 
\begin{align}
\label{phi32}
    \Phi_{3F2}(x_1,x_2) = x_1^{\alpha_1-\beta_2} x_2^{\alpha_2} (1-x_1)^{\beta_1-\alpha_1} (1- \lambda x_2)^{-\alpha_0} (x_1-x_2)^{\beta_2-\alpha_2}\, . 
\end{align}
The multivalued function in eq.~(\ref{phi32}) can be interpreted as the twist defining the twisted cohomology group $H_{\text{dR}}^{2}(X_{3F2}, \nabla_{3F2})$ with 
\begin{align}
    X_{3F2}= \mathbb{C}^2 - \left\{ (x_1, x_2): x_1 =0 \cup x_2=0 \cup x_1=1\cup x_2= \frac{1}{\lambda} \cup x_1=x_2\right\} \,,
\end{align}
and $\nabla_{3F2} = \rd + \rd\! \log \Phi_{3F2}$, with $\rd = \rd x_1 \partial_{x_1}+\rd x_2 \partial_{x_2}$. 
We consider the case where 
\begin{equation}\bsp
\label{cyparameters}
 \bs{\alpha}&=\{\alpha_0,\alpha_1,\alpha_2\}=  \left\{\frac{1}{2}+ n_1 \eps,\frac{1}{2}+ n_2 \eps,\frac{1}{2}+ n_3 \eps\right\} \,,\\
 \bs{\beta}&= \{\beta_1,\beta_2\}=\left\{ \, 1+ m_1 \eps,  1+ m_2 \eps\,\right\}\,.
\esp\end{equation}
We choose our initial basis for the twisted cohomology group with $I^\lt_i=\int_D  \Phi_{3F2}\, \varphi_i^\lt$ as follows~\cite{Duhr:2025lbz}:
\begin{align}\label{initialbasisCY1}
I_1^\lt&={_3\mathcal{F}}_2\left(\bs{\alpha}\, ; \, \bs{\beta}\, ;\lambda\right)\,,\quad I_2^\lt=\,\partial_\lambda I_1^\lt\, ,\quad I_3^\lt=\,\partial^2_\lambda I_1^\lt\,.
\end{align}
For the dual basis we choose 
\begin{align}
\label{initialduabasisCY1}
\check{\varphi}^\lt_i = \left.\frac{(x_1-1)(x_1-x_2)x_2}{x_1(x_2\lambda-1)}  \varphi_i^\lt\right|_{\eps\rightarrow -\eps}\, .
\end{align}
Note that taking $\eps\rightarrow -\eps$ explicitly in eq.~\eqref{initialduabasisCY1} comes from the fact that our initial basis in eq.~\eqref{initialbasisCY1} is $\eps$-dependent.

The hypergeometric function in eq.~\eqref{eq:3F2_def} with $\bs{\alpha}$ and $\bs{\beta}$ as in eq.~\eqref{cyparameters} is annihilated by a differential operator $\mathcal{L}^\lt_{\eps}$ of order three. The operator $\mathcal{L}^\lt_{\eps=0}$ is a CY operator, and its solutions compute the periods of a one-parameter family of K3 surfaces (cf., e.g., refs.~\cite{Duhr:2022pch,Duhr:2023eld}). The operator $\mathcal{L}^\lt_{\eps}$ is self-dual for generic values of $\eps$ and $m_1,m_2,n_1, n_2,n_3$~\cite{Duhr:2025lbz}. Hence, this operator falls into the class defined at the beginning of this subsection. With the choice of bases in eqs.~\eqref{initialbasisCY1} and~\eqref{initialduabasisCY1}, the period matrix and its dual are related by
\beq
\bs{\check{P}}^\lt(\lambda,\eps) = \bs{{P}}^\lt(\lambda,-\eps)\,.
\eeq

In ref.~\cite{Duhr:2025lbz} it was shown how to transform the first-order linear system corresponding to $\mathcal{L}^\lt_{\eps}$ to a canonical basis. Details can be found in ref.~\cite{Duhr:2025lbz}. Here we only note that the final part of the rotation takes the form
\begin{align}\label{eq:bUt_3F2}
\bUt=\left(
\begin{array}{ccc}
 1 & 0 & 0 \\
t_1 & 1 & 0 \\
t_2+\frac{t_{2,-1}}{\eps}  & t_3 & 1 \\
\end{array}
\right) \,,
\end{align}
in agreement with the general form given in eq.~\eqref{eq:Ut0_CY_op}. 
We see that it contains four potentially new functions, three of which appear in the $\mathcal{\eps}^0$ part of the rotation. 

We now compute the canonical intersection matrix for this case, and show that we obtain the general form in eq.~\eqref{deltaCY1}. We have explicitly computed the cohomology intersection matrix $\bs{C}^\lt$ for the choice of initial bases in eqs.~\eqref{initialbasisCY1}, and~\eqref{initialduabasisCY1}. It is a rational function of $\lambda$ and $\eps$. We rotate it into the canonical basis,
\begin{equation}
    \bs{C}^\lt_c(\lambda,\eps) =\bs{U}(\lambda,\eps) \bs{C}^\lt(\lambda,\eps)\bs{U}(\lambda,-\eps)^T\,.
\end{equation}
At first glance, $\bs{C}^\lt_c(\lambda,\eps)$ does not appear to be constant in $\lambda$. However, we can evaluate the entries numerically, and we find
\begin{align}\label{canonicalC3F2}
    \bs{C}^\lt_c(\lambda,\eps)=-\eps^4 \bK_3 \,.
\end{align}
in agreement with eq.~\eqref{deltaCY1}.

If we require eq.~\eqref{canonicalC3F2} to hold, we obtain relations among the $\eps$-functions. In this case, we find the relations
\begin{align}\label{rels3F2}
    t_3+t_1&\,=\psi_0\,\frac{ m_1+m_2-\lambda (n_1+n_2+n_3)}{\sqrt{1-\lambda}}=:2 r_1^\lt \,, \\
   t_2+\frac{t_1^2}{2}&\,=\frac{\psi_0^2}{2}  (m_1 m_2-\lambda (n_1n_2+n_1n_3+n_2 n_3))=:(r_1^\lt)^2+2r_2^\lt \notag \,,
\end{align}
where $\psi_0$ is the solution of $\mathcal{L}^\lt_{\eps=0}$ that is holomorphic at the MUM-point $\lambda=0$.
Let us make two comments. First, recall that the $\eps$-functions, and thus the previous relations, are only defined up to some undetermined additive constants, which we can interpret as the integration constants of the functions we fixed, and we set them to zero.
Second, in this particular case the relations in eq.~\eqref{rels3F2} can also be found from the defining differential equations of the functions $t_i$. In cases with more variables this can however quickly become complicated, but our method still allows us to find such relations in an algorithmic way.

We may solve the relations in eq.~\eqref{rels3F2} and express $t_2$ and $t_3$ in terms of $t_1$ and functions from $\cFss$. In our case, this is particularly easy, because eq.~\eqref{rels3F2} is linear in $t_2$ and $t_3$. In other cases (with more integrals and/or depending on more variables), these relations will in general be non-linear and hard to solve. As discussed in section~\ref{sec:constraints}, we can use the decomposition from Proposition~\ref{prop:main_1} to parametrise $\bUt$ and define a set of generators for $\cFeps$, and we reduce the problem to solving only linear constraints. We now illustrate this on our example.

The group $\textrm{UT}^{\bK_3}$ has dimension 3, in agreement with the parametrisation in eq.~\eqref{eq:Ut0_CY_op} for $N=3$. From Proposition~\ref{prop:main_1}, we know that there is $\bR^\lt\in\textrm{ST}^{\bK_3}$ and $\bO^\lt\in\textrm{OT}^{\bK_3} $ such that
\beq
\bUt^{(0)} = \bO^\lt\bR^\lt\,.
\eeq
The matrix $\bR^\lt$ satisfies $\bR^\lt=\bK_3\big(\bR^\lt\big){}^T\bK_3$, and so it is persymmetric. An easy computation shows that
\begin{align}\label{gensym3F2}
\bR^\lt=\left(
\begin{array}{ccc}
 1 & 0 & 0 \\
 r_1^\lt & 1 & 0 \\
r_2^\lt  & r_1^\lt & 1 \\
\end{array}
\right)\,,
\end{align}
where $r_1^\lt$ and $r_2^\lt$ are related to $t_1$, $t_2$, $t_3$ as in eq.~\eqref{rels3F2}. Using eq.~\eqref{constrantR}, we can fix $r_1^\lt$ and $r_2^\lt$ to the values in $\cFss$ given in eq.~\eqref{rels3F2}. Similarly, we find
\begin{align}
\bs{O}^\lt=\left(
\begin{array}{ccc}
 1 & 0 & 0 \\
 G^\lt  & 1 & 0 \\
 -\frac{\left(G^\lt\right)^2}{2}  &- G^\lt  & 1 \\
\end{array}
\right)\,,
\end{align}
where 
\beq
G^\lt=\frac{1}{2}(t_1-t_3)\,.
\eeq
 We thus see that $\cFeps$ is generated by at most the single function $G^{\lt}$, in agreement with the upper bound in eq.~\eqref{eq:dim_OTKN}. We now argue that the bound in eq.~\eqref{eq:dim_OTKN} is saturated, and $G^{\lt}$ cannot be expressed in terms of functions from $\cFss$. 
If we change variables from $\lambda$ to the canonical variable $\tau$ defined by
\beq
\tau = \frac{\psi_1(\lambda)}{\psi_0(\lambda)}\,,
\eeq
where $\psi_1$ is the solution of $\mathcal{L}_{\eps=0}^\lt$ that diverges like a single power of a logarithm close to the MUM-point, then the field $\cFss$ can be identified with the field generated by meromorphic quasi-modular forms for the congruence subgroup $\Gamma(2)$. The function $G^\lt$ can be written as a double-primitive of a meromorphic modular form\footnote{In fact, a magnetic modular form~\cite{magnetic1,magnetic2,magnetic3,magnetic4,Bonisch:2024nru}.} of weight four for the congruence subgroup $\Gamma(2)$ with a triple pole at $\lambda=\frac{1}{2}$. It can be shown that it is not possible to evaluate this primitive in terms of (quasi-)modular forms for $\Gamma(2)$, i.e., in terms of elements from $\cFss$~\cite{matthes_AMS,Broedel:2021zij,brown_fonseca}. Hence, $G^{\lt}\notin\cFss$, and so the bound in eq.~\eqref{eq:dim_OTKN} is saturated. 

\subsubsection{Hypergeometric \texorpdfstring{${}_4F_3$}{4F3} functions 
associated with a one-parameter family of CY 3-folds}
\label{app.moref43}

As a second example, we discuss a CY operator of degree four whose solution can be written in terms of a hypergeometric ${}_4F_3$ function.  Since the general setup is very similar to section~\ref{hypergeometricexample}, we will be brief in defining the context. 

More specifically, we consider a hypergeometric ${}_4F_3$ function defined by the twist
 \begin{align}
 \label{eq:twist4F3}
\Phi_{4F3}= x_3^{\alpha_3} (1-x_1)^{\beta_1-\alpha_1} x_1^{\alpha_1-\beta_2} x_2^{\alpha_2-\beta_3} (1-x_3 \lambda)^{-\alpha_0} (x_1-x_2)^{\beta_2-\alpha_2} (x_2-x_3)^{\beta_3-\alpha_3}\,,
 \end{align}
 with parameters
 \begin{align}\label{eq:4f3parameters}
 \bs{\alpha}&=\{\alpha_0,\alpha_1,\alpha_2,\alpha_3\}=\left\{\frac{1}{2}+n_1 \eps,\frac{1}{2}+n_2 \eps,\frac{1}{2}+n_3 \eps,\frac{1}{2}+n_4 \eps\right\}\,,\\
\bs{\beta}&=\{\beta_1,\beta_2,\beta_3\}=\left\{1+m_1 \eps,1+m_2 \eps,1+m_3 \eps\right\}\,.\notag
 \end{align}
This hypergeometric function is annihilated by a fourth-order ordinary differential operator, which for $\eps=0$ is the CY operator corresponding to the entry AESZ 3 in the database of CY operators~\cite{Almkvist2}.
The solutions to this CY operator compute the periods of a one-parameter family of CY threefolds. 
The parameters $m_i$ and $n_i$ are assumed to be generic (in the sense discussed earlier), and the deformed CY operator is self-dual.

We choose as a basis for the twisted cohomology group defined by the twist in eq.~\eqref{eq:twist4F3} the differential forms $\varphi_i^\lbf$, such that
$I^\lbf_i=\int_D  \Phi_{4F3}\, \varphi_i^\lbf$ is
\begin{align}
\label{CY3basis}
I_1^\lbf&={_4\mathcal{F}}_3\left(\bs{\alpha}\, ; \, \bs{\beta}\, ;\lambda\right)\,,\,I_2^\lbf=\,\partial_\lambda I_1^\lbf\,,\,I_3^\lbf=\,\partial^2_\lambda I_1^\lbf\,,\, 
I_4^\lbf=\,\partial^3_\lambda I_1^\lbf\,.
\end{align}
We choose the dual basis as 
 \begin{align}\label{CY3basisdual}
\check{\varphi}_i^\lbf=\frac{(1-x_1) (x_1-x_2) (x_2-x_3) x_3}{x_1 x_2 (1-x_3 \lambda)}\varphi^\lbf_i|_{\eps\rightarrow -\eps}\,,
 \end{align}
 which ensures that the period matrix and its dual are simply related by changing the sign of $\eps$.

In ref.~\cite{Duhr:2025lbz} a canonical basis for the first-order linear system attached to this deformed CY operator was obtained.  
The last transformation $\bUt$ that rotates the basis in eq.~\eqref{CY3basis} into a canonical form is given by
 \begin{align}
 \label{eq:bUt_CY3}
 \bUt=\left(
\begin{array}{cccc}
 1 & 0 & 0 & 0 \\
t_1 & 1 & 0 & 0 \\
 t_2+\frac{t_{2,-1}}{\eps } & t_3& 1 & 0 \\
 t_4+\frac{t_{4,-1}}{\eps }+\frac{t_{4,-2}}{\eps ^2} &t_5+\frac{t_{5,-1}}{\eps } & t_6& 1 \\
\end{array}
\right)\,,
 \end{align}
 where the functions $t_i$ provide a set of generators for $\cFeps$, and they can be written as (iterated) integrals over functions from $\cFss$. We can also obtain converging series representations for all these functions close to the MUM-point $\lambda=0$. The functions $t_{i,-1}$ and $t_{4,-2}$ can be expressed in terms of the derivatives of the generators $t_i$.

 We have computed the intersection matrix $\bC^\lbf(\lambda,\eps)$ in the original bases in eqs.~\eqref{CY3basis} and~\eqref{CY3basisdual}. The resulting matrix $\bC^\lbf_c(\lambda,\eps)$ depends on the $\eps$-functions $t_i$. Using the series representations for the latter, we can evaluate all entries of $\bC^\lbf_c$ numerically, and we find $\bC^\lbf_c(\lambda,\eps) = \eps^5 \bDelta^\lbf$, with

\begin{align}\label{delta4f3}
\bDelta^\lbf=\bK_4 \,, 
\end{align}
in agreement with eq.~\eqref{deltaCY1}. Note that it is easy to see that $\bUt^\lbf\in\textrm{UT}^{\bK_4}$, so that Observation~\ref{obs:Ut_in_UT} is fulfilled.

In eq.~\eqref{eq:bUt_CY3}, we have introduced a set of 6 generators $t_i$, $i=1,\ldots,6$, for $\cFeps$. From eq.~\eqref{eq:dim_OTKN} with $N=4$, we expect that we need at most two generators, and so there must be four independent polynomial relations between the $t_i$. These relations are provided by eq.~\eqref{eq:A_phiA}. In order to solve these constraints, we use Proposition~\ref{prop:main_1} to write
\beq\label{eq:Ut0_CY3}
\bUt^{\lbf,(0)} = \bO^\lbf \bR^\lbf\,,
\eeq
with $\bO^\lbf\in\textrm{OT}^{\bK_4}$ and $\bR^\lbf\in\textrm{ST}^{\bK_4}$. In particular, $\bR^\lbf$ is persymmetric, and we parametrise it as
\begin{align}
\bs{R}^{(4F3)}=\left(
\begin{array}{cccc}
 1 & 0 & 0 & 0 \\
 r_1^\lbf & 1 & 0& 0 \\
 r_2^\lbf& r_3^\lbf & 1 & 0 \\
 r_4^\lbf & r_2^\lbf &  r_1^\lbf& 1 \\
\end{array}
\right)\,.
\end{align}
Similarly, we can parametrise $\bO^\lbf$ as
\begin{align}
\bs{O}^\lbf=\left(
\begin{array}{cccc}
 1 & 0 & 0 & 0 \\
 G_1^\lbf & 1 & 0 & 0 \\
 G_2^\lbf & 0 & 1 & 0 \\
 -G_1^\lbf G_2^\lbf & -G_2^\lbf & -G_1^\lbf & 1 \\
\end{array}
\right)\,.
\end{align}
The functions $G_1^\lbf$, $G_2^\lbf$ and $r_i^\lbf$, $i=1,\ldots,4$, provide an alternative set of generators for $\cFeps$. The relations between the two sets of generators is easily obtained from eq.~\eqref{eq:Ut0_CY3}. We find
\begin{align}\label{rels4F3}
r_1^\lbf&=\frac{1}{2} (t_1 + t_6)\notag\,,\\
r_2^\lbf&=\frac{1}{4}(t_1 t_3+2 t_2-t_3 t_6+2 t_5)\notag\,,\\
r_3^\lbf&=t_3\,,\\
r_4^\lbf&=\frac{1}{8}\Big(6 t_1 t_2-t_1^2 t_3-2 t_1 t_5-2 t_2 t_6+t_3 t_6^2+8 t_4-2 t_5 t_6\Big)\notag\,,\\
 G_1^\lbf&=\frac{1}{2} \left(t_1-t_6\right)\,,\notag\\
 G_2^\lbf&=\frac{1}{4} \left(2 t_2-t_1 t_3-2 t_5+t_3 t_6\right)\,.\nonumber
\end{align}
We know that the functions $r_i^\lbf$ can be expressed in terms of elements from $\cFss$ by solving the constraint in eq.~\eqref{constrantR}, which yields,
\begin{align}
r_1^\lbf&=Y_2^{1/3} \psi_0^{2/3}\,\frac{ \sigma^m_1-\lambda\,\sigma^n_1 }{2(1-\lambda)^{2/3}}\,, \notag\\
r_2^\lbf &=\frac{\psi_0^{4/3}}{Y_2^{1/3}}\,\frac{4\sigma^{m}_2-(\sigma_1^m)^2-2\,\lambda\,(2\sigma^m_2-\sigma^m_1\sigma^n_1+2\sigma^n_2)+\lambda^2(4\sigma^n_2-(\sigma^n_1)^2)}{8(1-\lambda)^{4/3}}\,, \notag \\
r_3^\lbf&=\frac{r_1^\lbf}{ Y_2}\,, \\
r_4^\lbf &=\frac{\psi_0^2}{16(1-\lambda)^2}\Big[(\sigma_1^m)^3-4\sigma^m_1\sigma_2^m+8\sigma_3^m-\lambda^3((\sigma_1^n)^3-4\sigma_1^n\sigma_2^n+8\sigma_3^n)\notag \\
&\qquad +\lambda\,(4\sigma_1^m\sigma_2^m-16\sigma_3^m-3(\sigma_1^m)^2\sigma_1^n+4\sigma_2^m\sigma_1^n+4\sigma_1^m\sigma_2^n-8\sigma_3^n) \notag \\
&\qquad +\lambda^2(8\sigma_3^m-4\sigma_1^n\sigma_2^n-4\sigma_2^m\sigma_1^n+3\sigma_1^m(\sigma_1^n)^2-4\sigma_1^m\sigma_2^n+16\sigma_3^n)\Big] \,, \notag
\end{align}
where 
\beq
\sigma_k^m = \sigma_k(m_1,m_2,m_3) \textrm{~~~and~~~}\sigma_k^n = \sigma_k(n_1,n_2,n_3,n_4)\,,
\eeq
and $\sigma_k$ denotes the elementary symmetric polynomial of degree $k$. Furthermore, $\psi_0$ is the solution of the CY operator holomorphic at the MUM-point $\lambda=0$ and $Y_2$ is the associated Yukawa coupling~\cite{Bogner:2013kvr}. The precise from of the Yukawa coupling is immaterial for our purposes, and it suffices to say that it is holomorphic at the MUM-point and can be expressed in terms of minors of the period matrix of the CY threefold. We thus see that, as expected, $\cFeps$ can be generated by the two functions $G_1^\lbf$ and $G_2^\lbf$ alone. These functions can be written as (iterated) integrals over elements from $\cFss$. While we do not have a proof, we expect these functions not to be expressible in terms of functions from $\cFss$. 

We stress that, in principle, it is not mandatory to use Proposition~\ref{prop:main_1}, and one may attempt to solve the constraints from eq.~\eqref{eq:A_phiA} directly in terms of the six generators $t_i$. The resulting equations, however, are non-linear, and are harder to solve. If we change instead to the new set of generators provided by Proposition~\ref{prop:main_1}, we can obtain the $r_i^\lbf$ by only solving simpler, linear constraints. 

\paragraph{Obtaining $\bDelta$ analytically.}
We conclude this example by illustrating the strategy introduced at the end of section~\ref{sec:consequences} of how to determine the canonical intersection matrix without having to evaluate the matrix of intersection numbers in the original basis or having to resort to numerical evaluations. We will see that we recover the form predicted in eq.~\eqref{deltaCY1}, though the method described here does not rely on eq.~\eqref{deltaCY1} and it can be applied also to other cases. 

We assume that $\bDelta$ is symmetric.\footnote{If we do not find a solution where $\bDelta$ is symmetric, then we just need to restart the reasoning with an antisymmetric matrix.}
With the rotation in eq.~\eqref{eq:bUt_CY3} it follows that any symmetric matrix $\widehat{\bC}$ satisfying the closure condition in eq.~\eqref{eq:Ctilde_Gpar} that lies in the solution space of eq.~\eqref{eq:Omega1CTilde1} is of the form
\begin{align}
\widehat{\bC}=c\left(
\begin{array}{cccc}
 0 & 0 & 0 & 1 \\
 0 & 0 & 1& \widehat{C}_{2,4} \\
 0 & 1& \widehat{C}_{3,3} &\widehat{C}_{3,4} \\
 1 &\widehat{C}_{2,4} & \widehat{C}_{3,4} & \widehat{C}_{4,4} \\
\end{array}
\right)\,,
\end{align}
with some normalization $c\in \mathcal{F}_{\mathrm{ss}}$
and 
\begin{align}
\widehat{C}_{2,4}=&\frac{(Y_2 \psi_0^2)^{\frac{1}{3}} \left(\lambda \sigma_1^n-\sigma_1^m\right)}{\left(1-\lambda\right)^{2/3}}\,,\notag\\
\widehat{C}_{3,3}=&\frac{\left(\psi_0^2\right)^{\frac{1}{3}} \left(\lambda \sigma_1^n-\sigma_1^m\right)}{\left((1-\lambda) Y_2\right)^{2/3}}\,,\\
\widehat{C}_{3,4}=&\frac{\left(\psi_0^2\right)^{2/3} }{(1-\lambda)^{\frac{4}{3}} Y_2^{\frac{1}{3}}}\Big[-2 \lambda \sigma_1^m \sigma_1^n+(\lambda-1) \left(\sigma_2^m-\lambda \sigma_2^n\right)+\left(\sigma_1^m\right)^2+\lambda^2 \left(\sigma_1^n\right)^2\Big]\,,\nonumber\\
\widehat{C}_{4,4}=&\frac{\psi_0^2 }{(1-\lambda)^2}\Big[2\sigma_1^m\sigma_2^m-(\sigma_1^m)^3-\sigma_3^m+\lambda(2\sigma_3^m-2\sigma_1^m\sigma_2^m+3(\sigma_1^m)^2\sigma_1^n\nonumber\\
&-2\sigma_1^m\sigma_2^n-2\sigma_2^m\sigma_1^n+\sigma_3^n)  -\lambda^2 (\sigma_3^m+ 3\sigma_1^m(\sigma_1^n)^2-2\sigma_1^m\sigma_2^n \nonumber\\
&-2\sigma_2^m\sigma_1^n -2\sigma_1^n\sigma_2^n+2\sigma_3^n)+\lambda^3((\sigma_1^n)^3-2\sigma_1^n\sigma_2^n+\sigma_3^n)\Big]\,. \nonumber
\end{align}
Rotating to $\bDelta$ gives
\begin{align}
\bDelta=
c\left(
\begin{array}{cccc}
 0 & 0 & 0 & 1 \\
 0 & 0 & 1 & \Delta_{2,4} \\
 0 &1& \Delta_{3,3}& \Delta_{3,4} \\
 1 &\Delta_{4,2}&\Delta_{4,3} & \Delta_{4,4}\\
\end{array}
\right) \,,
\end{align}
with 
\beq\bsp\label{eq:Delta4F3tdep}
\Delta_{2,4}&=\widehat{C}_{2,4}+t_1 + t_6\,,\\
\Delta_{3,3}&=\widehat{C}_{3,3}+ 2t_3\,,\\
\Delta_{3,4}&=\widehat{C}_{3,4}+t_5 +t_3 \widehat{C}_{2,4}+t_6 \left(\widehat{C}_{3,3}+t_3 \right)+ t_2\,,\\
\Delta_{4,4}&=\widehat{C}_{4,4}+t_6^2 \widehat{C}_{3,3}+2 t_6 \widehat{C}_{3,4}+2 t_4 +2 t_5 \left(\widehat{C}_{2,4}+t_6 \right)\,.
\esp\eeq
The only entries that are independent of the $\eps$-functions $t_i$ are either zero or lie on the anti-diagonal and are equal to $c$, implying that $c$ is a constant. We require all linearly independent entries that contain $\eps$-functions $t_i$ to be constant. This immediately gives the relations in eq.~\eqref{rels4F3}. 
The canonical intersection matrix takes the form 
 \beq\label{mostgendelta}
 \bDelta_d=\left(
 \begin{array}{cccc}
  0 & 0 & 0 & 1 \\
  0 & 0 & 1 & d_1 \\
  0 & 1 & d_2 & d_3 \\
  1 & d_1 & d_3 & d_4 \\
 \end{array}
 \right)\,,
 \eeq
where the $d_i$ are arbitrary constants introduced by requiring the relations in eq.~\eqref{eq:Delta4F3tdep} to evaluate to constants. The overall constant $c$ can be absorbed into $f(\epsilon)$ in eq.~\eqref{eqfeps}. For $d_i=0$ we recover the form of the intersection matrix in eq.~\eqref{deltaCY1}. We see that we can extract in one go both the relations from eq.~\eqref{eq:Delta4F3tdep} defining $\bR$ and the canonical intersection matrix from eq.~\eqref{deltaCY1}.

Let us conclude by giving an interpretation of the arbitrary coefficients $d_i$ introduced in eq.~\eqref{mostgendelta}. In section~\ref{sec:consequences} we have argued that the most general form of the canonical intersection matrix is given by eq.~\eqref{eq:Delta'_to_Delta}. We now show that the free parameters $d_i$ are precisely related to this freedom in defining the canonical intersection matrix. To see this, pick a constant lower-triangular unipotent $\bM\in\Gpar$,
 \begin{align}
 \bM=\left(
\begin{array}{cccc}
 1 & 0 & 0 & 0 \\
c_1 & 1 & 0 & 0 \\
 c_2 & c_3& 1 & 0 \\
 c_4 &c_5 & c_6& 1 \\
\end{array}
\right)\,.
 \end{align}
 We can write it as $\bM=\bM_R\bM_O$, where $\bM_R$ is $\bK_4$-symmetric and $\bM_O$ is $\bK_4$-orthogonal. We have
  \beq\bsp
 \bM_O &=\left(
\begin{array}{cccc}
 1 & 0 & 0 & 0 \\
 o_1 & 1 & 0 & 0 \\
 o_2 & 0 & 1 & 0 \\
 -o_1\, o_2 & -o_2 & -o_1 & 1 \\
\end{array}
\right)\,,\\
\bM_R&=\left(
\begin{array}{cccc}
 1 & 0 & 0 & 0 \\
 \frac{d_1}{2} & 1 & 0 & 0 \\
 \frac{1}{8} (4 d_3-d_2 d_2) & \frac{d_2}{2} & 1 & 0 \\
 \frac{1}{16} \left(d_2 d_1^2-4 d_3 d_1+8 d_4\right) & \frac{1}{8} (4 d_3-d_1 d_2 )& \frac{d_1}{2} & 1 \\
\end{array}
\right)\,,
\esp\eeq
with
 \begin{align}\label{didelta}
 o_1&=\frac{1}{2}(c_1-c_6)\notag\,,\quad 
 o_2=\frac{1}{4} \left(2 c_2-c_1 c_3-2 c_5+c_3 c_6\right)\notag\,,\\
 d_1&=c_1+c_6\,,\quad d_2=2 c_3\,, \quad d_3=c_2+c_5+c_3 c_6\,, \quad d_4=c_4+c_5 c_6\notag\,\,.
 \end{align}
 It is now easy to see that
 \beq
 \bM_R\bK_4\bM_R^T = \bDelta_d\,,
 \eeq
 i.e., the arbitrary parameters $d_i$ precisely parametrise the different forms for the canonical intersection matrix discussed in section~\ref{sec:consequences}.

\subsection{Lauricella functions associated to higher-genus curves}
\label{subsec:higherGenus}
 
While many efforts have focused on Feynman integrals associated with CY varieties, higher-genus Riemann surfaces have recently attracted some attention~\cite{Doran:2023yzu,Huang:2013kh,Georgoudis:2015hca,Marzucca:2023gto,Duhr:2024uid}. We will now study setups which we expect to appear in the maximal cuts of Feynman integrals associated with a hyperelliptic curve of genus $g$, i.e., a Riemann surface defined by a polynomial equation
\begin{equation}
    y^2=P(x) \,,
\end{equation}
where $P(x)$ is of degree $2g+1$ or $2g+2$. We will focus on maximal cuts admitting a Baikov representation with twist of the form
\begin{equation}\label{eq:Phi_g}
   \Phi_g= (x-b_1)^{-\frac{1}{2}+a_1\eps}\dots (x-b_n)^{-\frac{1}{2}+a_n\eps} \,,
\end{equation}
for $n=2g+1$ or $2g+2$. The $b_i$ are distinct complex numbers, and the $a_i$ are parameters. There may be further linear factors that become unity for $\eps=0$, which may lead to additional master integrals. 

The twist $\Phi_g$ in eq.~\eqref{eq:Phi_g} defines a twisted cohomology group, and the corresponding twisted periods can be expressed in terms of Lauricella $F_D^{(n)}$ functions. Hypergeometric functions of this type  and the associated twisted cohomology theories are well studied, and we refer to the literature for details (see, e.g., refs.~\cite{matsumotoFD,Brown:2019jng}). In ref.~\cite{Duhr:2024uid} some of the authors have shown how to apply the algorithm of ref.~\cite{Duhr:2025lbz} to bring the corresponding system of differential equations into a canonical form. While the explicit computations in ref.~\cite{Duhr:2024uid} were restricted to $g=1,2$, it is possible to extend the results to arbitrary genus. The details of the rotation to the canonical form are irrelevant for the purposes of this paper, and we only want to discuss the structure of the canonical intersection matrix and give an upper bound on the number of generators for $\cFeps$ here.
We will distinguish whether the hyperelliptic curve is defined by a polynomial of odd or even degree, and we will refer to the corresponding curves as being \emph{odd} or \emph{even}. For an odd curve there are $2g$ master integrals, which precisely correspond to the first and second kind differentials spanning the (first) cohomology group of the hyperelliptic curve. For an even curve there is an additional third kind differential which is associated with an extra puncture on the curve, typically chosen to be at infinity.

\paragraph{Odd curves.}
For an odd curve, we expect the final rotation to be independent of $\eps$ and to take the form
\begin{equation}\label{eq:Ut_odd}
    \bUt^\lodd=\begin{pmatrix}
        \mathds{1} & \bs{0} \\ \bT & \mathds{1}
    \end{pmatrix} \,.
\end{equation}
The filtration $F_{\mathrm{t},p}$ is the filtration provided by the Hodge filtration on the first cohomology group of the underlying hyperelliptic curve at $\eps=0$. There is a basis in which the canonical intersection matrix takes the block anti-diagonal form
\begin{equation}
    \bDelta\!^\lodd=\begin{pmatrix}
        \bs{0} & \mathds{1} \\ \mathds{1} & \bs{0} 
    \end{pmatrix} \,.
\end{equation}
The above matrices are block matrices, i.e., each entry is a $g\times g$ matrix for $g>0$. This is precisely what was observed for $g=1,2$ in ref.~\cite{Duhr:2024uid}, and we expect it to hold also for higher genera. Note that Observation~\ref{obs:Ut_in_UT} is satisfied.

Let us now discuss the number of generators for $\cFeps$. From eq.~\eqref{eq:Ut_odd}, we see that $\cFeps$ is generated by the $g^2$ entries of the matrix $\bT$. This number of generators, however, is too large, and we can identify a smaller set. To this effect, we apply Proposition~\ref{prop:main_1} to write
\beq
\bUt^\lodd = \bO^\lodd\bR^\lodd\,,
\eeq
with
\beq
\bO^\lodd = \begin{pmatrix}
        \mathds{1} & \bs{0} \\ \bT_{\!A} & \mathds{1}
    \end{pmatrix}\in\textrm{OT}^{\Delta^\lodd}\textrm{~~~and~~~}\bR^\lodd=\begin{pmatrix}
        \mathds{1} & \bs{0} \\ \bT_{\!S} & \mathds{1}
\end{pmatrix}\in\textrm{ST}^{\Delta^\lodd}\,.
\eeq
It is easy to check that the matrices $\bT_{\!S}$ and $\bT_{\!A}$ must be symmetric and antisymmetric, respectively. From eq.~\eqref{constrantR} it follows that the entries of $\bT_{\!S}$ lie in $\cFss$, and so the entries of $\bT_{\!A}$ alone generate $\cFeps$. Hence, we conclude that an upper bound on the number of generators for $\cFeps$ is given by
\beq\label{eq:dim_OT_higher_g}
\dim\textrm{OT}^{\Delta^\lodd} = \frac{g(g-1)}{2}\,.
\eeq
Note that this implies that $\cFeps=\cFss$ for $g=1$, in agreement with the corresponding result for $N=2$ in eq.~\eqref{eq:dim_OTKN} (because curves of genus $g=1$ are at the same time CY one-folds). Moreover, we see that we need a single generator for $g=2$, in agreement with the explicit computations in ref.~\cite{Duhr:2024uid}. In ref.~\cite{Duhr:2024uid} it was also argued that this single generator of $\cFeps$ for $g=2$ does not lie in $\cFss$, so that the bound is saturated at least for $g\le 2$, and we expect it to be always saturated.

\paragraph{Even curves.}
In this case we expect the final rotation to be a block matrix of the form
\begin{equation}
    \bUt^\leven=\begin{pmatrix}
        1 \\
        &  \ddots  \\
        & &  1   \\
        t_{0,1} & \dots & t_{0,g} & 1 \\
        t_{1,1}& \dots & t_{1,g}& t_{1,g+1} & 1 \\
        \vdots & & \vdots & \vdots & &\ddots \\
        t_{g,1} & \dots & t_{g,g} & t_{g,{g+1}} &  &&1
        
    \end{pmatrix}\,.
    \label{eq:higher_genus_even_Ut}
\end{equation}
The filtration $F_{\!\mathrm{t},p}$ is again provided by the (mixed) Hodge structure of the underlying (punctured) hyperelliptic curve at $\eps=0$, and there is a basis in which
 the canonical intersection matrix is
\begin{equation}\label{eq:Delta_higher_genus_even}
    \bDelta\!^\leven=\begin{pmatrix}
        & & \mathds{1} \\
        & 1 & \\
        \mathds{1} & &
    \end{pmatrix} \,.
\end{equation}
One might expect that one also needs further entries below the diagonal, i.e., some entries $t_{i,j}$ with $i>1,j>g+1$. However, one can see that these entries are necessarily zero for the expected form of $\bDelta\!^\leven$, because otherwise Observation~\ref{obs:Ut_in_UT} would be violated.

In general the final rotation matrix $\bUt^\leven$ has $g(g+2)$ degrees of freedom, and $\cFeps$ is generated by the $g(g+2)$ functions $t_{i,j}$. Using Proposition~\ref{prop:main_1}, we see that these degrees of freedom can be separated into
\beq\label{eq:higher_genus_even_dim}
\dim\textrm{OT}^{\Delta\!^\leven} = \frac{g(g+1)}{2} \textrm{~~~and~~~}
\dim\textrm{ST}^{\Delta\!^\leven} = \frac{g(g+3)}{2}\,.
\eeq
Hence, $\cFeps$ can be generated by at most $\frac{g(g+1)}{2}$ functions.
This is consistent with the results of ref.~\cite{Duhr:2024uid} for $g=1,2$.

\subsection{A four-loop banana integral with two different masses}
\label{fourlooptwomassbaanna}
\begin{figure}[!h]
\centering
\begin{tikzpicture}
\coordinate (llinks) at (-2.5,0);
\coordinate (rrechts) at (2.5,0);
\coordinate (links) at (-1.5,0);
\coordinate (rechts) at (1.5,0);
\begin{scope}[very thick,decoration={
    markings,
    mark=at position 0.5 with {\arrow{>}}}
    ] 
\draw [-, thick,postaction={decorate}] (links) to [bend right=25]  (rechts);
\draw [-, thick,postaction={decorate}] (links) to [bend right=-25]  (rechts);

\draw [-, thick,postaction={decorate}] (links) to [bend left=85]  (rechts);
\draw [-, thick,postaction={decorate}] (llinks) to [bend right=0]  (links);
\draw [-, thick,postaction={decorate}] (rechts) to [bend right=0]  (rrechts);
\draw [-, thick,postaction={decorate}] (links) to [bend right=0]  (rechts);
\end{scope}
\begin{scope}[very thick,decoration={
    markings,
    mark=at position 0.5 with {\arrow{>}}}
    ]
\draw [-, thick,postaction={decorate}] (links) to  [bend right=85] (rechts);
\end{scope}
\node (d1) at (0,1.1) [font=\scriptsize, text width=.15 cm]{$m_1$};
\node (d2) at (0,0.6) [font=\scriptsize, text width=.15 cm]{$m_1$};
\node (d3) at (0,-0.17) [font=\scriptsize, text width=.15 cm]{$m_1$};
\node (d4) at (0,-.65) [font=\scriptsize, text width=0.15 cm]{$m_2$};
\node (d4) at (0,0.2) [font=\scriptsize, text width=0.15 cm]{$m_1$};
\node (p1) at (-2.0,.25) [font=\scriptsize, text width=1 cm]{$p$};
\node (p2) at (2.4,.25) [font=\scriptsize, text width=1 cm]{$p$};
\end{tikzpicture}
\caption{The four-loop banana graph with two unequal masses.
}
\label{fig:banana4}
\end{figure}
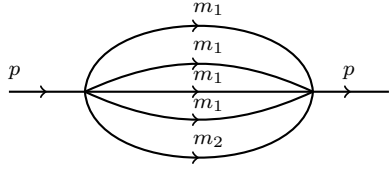
Next we apply our considerations to the family of Feynman integrals depicted in figure~\ref{fig:banana4}.
In our conventions, this family is defined as:
\begin{equation}
    I_{\nu_1,\nu_2,\dots,\nu_{14}}(p^2,m_1^2,m_2^2;\epsilon)=e^{4\gamma_\text{E} \eps}\int \dfrac{\text{d}^D \ell_1}{i\pi^{\frac{D}{2}}}\dfrac{\text{d}^D \ell_2}{i\pi^{\frac{D}{2}}} \dfrac{\text{d}^D \ell_3}{i\pi^{\frac{D}{2}}}\dfrac{\text{d}^D\ell_4}{i\pi^{\frac{D}{2}}} \dfrac{1}{D_1^{\nu_1} D_2^{\nu_2}\dots D_{14}^{\nu_{14}}}\,,
\end{equation}
with the denominators
\begin{align}
    D_1&=\ell_1^2-m_1^2\,,\quad D_2=\ell_2^2-m_1^2\,,\quad D_3=\ell_3^2-m_1^2\,,\quad
    D_4=\ell_4^2-m_1^2\,,\notag\\
    D_5&=(\ell_1+\ell_2+\ell_3+\ell_4-p)^2-m_2^2\,,\quad D_6=(\ell_1-p)^2\,,\quad D_7=(\ell_2-p)^2\,,\notag\\
    D_8&=(\ell_3-p)^2\,,\quad D_9=(\ell_4-p)^2\,,\quad D_{10}=(\ell_1-\ell_2)^2\,,\quad D_{11}=(\ell_1-\ell_3)^2\,,\notag\\
    D_{12}&=(\ell_1-\ell_4)^2\,,\quad D_{13}=(\ell_2-\ell_3)^2\,,\quad D_{14}=(\ell_2-\ell_4)^2\,.
\end{align}
We define the two dimensionless variables
\begin{equation}
    x_1=\dfrac{m_1^2}{p^2}\textrm{~~~and~~~} x_2=\dfrac{m_2^2}{p^2}\,,
\end{equation}
where $x_1=x_2=0$ is a MUM-point.
This integral family can be expressed in terms of nine master integrals, which we may choose as
\begin{equation}
 \bs{I^\text{(4-1)}}=\begin{pmatrix}I_{1,1,1,1,0,0,0,0,0,0,0,0,0,0}\\
 I_{1,1,1,0,1,0,0,0,0,0,0,0,0,0}\\
 I_{1,1,1,1,1,0,0,0,0,0,0,0,0,0}\\
 \partial_{x_1} I_{1,1,1,1,1,0,0,0,0,0,0,0,0,0}\\
 \partial_{x_2} I_{1,1,1,1,1,0,0,0,0,0,0,0,0,0}\\
 I_{1,1,1,1,1,-1,0,0,0,0,0,0,0,0}\\
 \partial_{x_1}^2 I_{1,1,1,1,1,0,0,0,0,0,0,0,0,0}\\
 \partial_{x_2}^2 I_{1,1,1,1,1,0,0,0,0,0,0,0,0,0}\\
 \partial_{x_1}^2\partial_{x_2} I_{1,1,1,1,1,0,0,0,0,0,0,0,0,0}\end{pmatrix} \,.
 \end{equation}
 Here, the first two integrals are tadpole integrals and the integral $I_{1,1,1,1,1,-1,0,0,0,0,0,0,0,0}$ involves an ISP.
 The other six master integrals define, at $\eps=0$, the periods of a two-parameter family of CY threefolds. More details about these CY threefolds and their periods can be found in ref.~\cite{Maggio:2025jel} and in appendix \ref{Yuk_appendix}. In that appendix we also discuss how to use the self-duality of the periods and eq.~\eqref{DEQC} to compute the Yukawa couplings as rational functions in $x_1$ and $x_2$.

A canonical basis for this family of Feynman integrals was determined in section 5.2 of ref.~\cite{Maggio:2025jel}. The rotation to this canonical basis is achieved through a sequence of rotations, 
\beq
\bs{J}^\text{(4-1)}=\bUt^{(\text{4-1})} \bUeps^\text{(4-1)}\bUss^\text{(4-1)} \bs{I}^{(\text{4-1})}\,,
\eeq
with $\bs{J}^\text{(4-1)}$ the vector of master integrals in the canonical basis. The expressions for the first two rotations
$\bs{U}_\eps^\text{(4-1)}$ and $\bs{U}_{\text{ss}}^\text{(4-1)}$ were given in ref.~\cite{Maggio:2025jel}. The entries of these matrices lie in $\mathcal{F}_\text{ss}$.\footnote{In particular we will be using the basis of eq.~(5.39) of ref.~\cite{Maggio:2025jel} after setting the auxiliary function $F_{i,j}^{\text{4-1}}$ to zero. 
Additionally, we make a slight change of ordering for the master integrals: $\{M_0^{\text{4-1}},M_1^{\text{4-1}},M_2^{\text{4-1}},M_3^{\text{4-1}},M_4^{\text{4-1}},M_7^{\text{4-1}},M_6^{\text{4-1}},M_5^{\text{4-1}},M_8^{\text{4-1}}\}$.} 

 The last rotation $\bs{U}_t^\text{(4-1)}$ is the focus of our paper,\footnote{Note that the functions $t_{i,j}^\text{(4-1)}$ can be related in a non-trivial way to the auxiliary functions $F_{i,j}^\text{(4-1)}$ (and their derivatives) of  ref.~\cite{Maggio:2025jel}. The reason for this non-trivial relation is that in ref.~\cite{Maggio:2025jel} there is no separation between the rotations $\bUss$ and $\bUt$.} and it reads 
 \begin{align}
 \bUt^{(\text{4-1})}{=}\begin{pmatrix}1&0&0&0&0&0&0&0&0\\
 0&1&0&0&0&0&0&0&0\\
 0&0&1&0&0&0&0&0&0\\
 0&0&t_{4,3}^\text{(4-1)}&1&0&0&0&0&0\\
 0&0&t_{5,3}^\text{(4-1)}&0&1&0&0&0&0\\
 0&0&t_{6,3}^\text{(4-1)}&0&0&1&0&0&0\\
 0&0&t_{7,3}^\text{(4-1)}{+}\frac{t_{7,3,{-}1}^\text{(4-1)}}{\eps}&t_{7,4}^\text{(4-1)}&t_{7,5}^\text{(4-1)}&0&1&0&0\\
 0&0&t_{8,3}^\text{(4-1)}{+}\frac{t_{8,3,{-}1}^\text{(4-1)}}{\eps}&t_{8,4}^\text{(4-1)}&t_{8,5}^\text{(4-1)}&0&0&1&0\\
 t_{9,1}^\text{(4-1)}&t_{9,2}^\text{(4-1)}&t_{9,3}^\text{(4-1)}{+}\frac{t_{9,3,{-}1}^\text{(4-1)}}{\eps}{+}\frac{t_{9,3,{-}2}^\text{(4-1)}}{\eps^2}&t_{9,4}^\text{(4-1)}{+}\frac{t_{9,4,{-}1}^\text{(4-1)}}{\eps}&t_{9,5}^\text{(4-1)}{+}\frac{t_{9,5,{-}1}^\text{(4-1)}}{\eps}&t_{9,6}^\text{(4-1)}&t_{9,7}^\text{(4-1)}&t_{9,8
}^\text{(4-1)}&1\end{pmatrix} \,.
 \end{align}
The functions $t_{i,j}^\text{(4-1)}$ are fixed by solving first-order differential equations and provide a set of generators for $\mathcal{F}_\eps$.
 We give an expression for $\bUt^{(\text{4-1})}$ expanded around the MUM-point $x_1=x_2=0$ in the ancillary file {\bf{\texttt banana.nb}}.

We now restrict the discussion to the maximal cuts. This can easily be done by deleting the first two rows and columns of all matrices introduced so far, and we only keep the last seven entries in the vector of master integrals. In particular, we denote the differential equation matrix for the maximal cuts in the canonical basis by $\bs{A}_\text{MC}^\text{(4-1)}$, and the last rotation matrix needed to bring the system of maximal cuts into canonical form by $\bs{U}_{\text{t,}\,\text{MC}}^{\text{(4-1)}}$.

We start by determining the canonical intersection matrix $\bDelta^\text{\!(4-1)}$. We proceed as explained in section~\ref{subsec:obtaining_Delta}. In our case eq.~\eqref{num} takes the form,
 \begin{equation}\label{4-1-diff}
     \bs{A}_\text{MC}^\text{(4-1)}\bs{\Delta}^\text{\!(4-1)}-\bs{\Delta}^\text{\!(4-1)}\bs{A}_\text{MC}^\text{(4-1)}{}^T=0\,.
 \end{equation}
 We interpret this equation as a linear system for the constant matrix $\bs{\Delta}^\text{\!(4-1)}$. Using the series representation for the $\eps$-functions $t_{i,j}^\text{(4-1)}$, we find,\footnote{Note that in ref.~\cite{Maggio:2025jel} we find a non-diagonal intersection matrix. This is expected, because the basis of ref.~\cite{Maggio:2025jel} is related to the one here through a constant rotation which (anti)-diagonalizes $\bs{\Delta}$.}
\beq
\bDelta^\text{\!(4-1)} = \bK_7\,.
\eeq

After having determined the canonical intersection matrix $\bs{\Delta}^\text{\!(4-1)}$, we can use it to reduce the number of generators of $\cFeps$. 
 Note that we have $\bs{U}_{\text{t,}\,\text{MC}}^{\text{(4-1)},\,(0)}\in \mathrm{UT}^{\bK_7}$. 
 Using Proposition \ref{prop:main_1}, we can write $\bs{U}_{\text{t,}\,\text{MC}}^{\text{(4-1)},\,(0)}=\bs{O}^\text{(4-1)}\bs{R}^\text{(4-1)}$, 
with
\begin{equation}\label{splittingY}
     \bs{R}^\text{(4-1)}=\begin{pmatrix}1&0&0&0&0&0&0\\
     r_1^\text{(4-1)}&1&0&0&0&0&0\\
     r_2^\text{(4-1)}&0&1&0&0&0&0\\
     r_3^\text{(4-1)}&0&0&1&0&0&0\\
     r_4^\text{(4-1)}&r_5^\text{(4-1)}&r_6^\text{(4-1)}&0&1&0&0\\
     r_7^\text{(4-1)}&r_8^\text{(4-1)}&r_5^\text{(4-1)}&0&0&1&0\\
     r_9^\text{(4-1)}&r_7^\text{(4-1)}&r_4^\text{(4-1)}&r_3^\text{(4-1)}&r_2^\text{(4-1)}&r_1^\text{(4-1)}&1
     \end{pmatrix}\,,\quad
 \end{equation}
 and
 \begin{equation}
 \label{eq:orthPartBanana}
 \bs{O}^\text{(4-1)}=\begin{pmatrix}1&0&0&0&0&0&0\\
     G_1^\text{(4-1)}&1&0&0&0&0&0\\
     G_2^\text{(4-1)}&0&1&0&0&0&0\\
     G_3^\text{(4-1)}&0&0&1&0&0&0\\
     G_4^\text{(4-1)}&G_5^\text{(4-1)}&0&0&1&0&0\\
     G_6^\text{(4-1)}&0&-G_5^\text{(4-1)}&0&0&1&0\\
    G_7^\text{(4-1)}&G_8^\text{(4-1)}& G_9^\text{(4-1)} &{-}G_3^\text{(4-1)}&{-}G_2^\text{(4-1)}&{-}G_1^\text{(4-1)}&1
     \end{pmatrix}\,,
     \end{equation}
     where we have abbreviated
     \begin{align}
   \notag      G_7^\text{(4-1)}&= {-}\dfrac{1}{2} (G_3^{\text{(4-1)}})^2 {-} G_2^\text{(4-1)} G_4^\text{(4-1)} {-} G_1^\text{(4-1)} G_6^\text{(4-1)}\,,\\
         G_8^\text{(4-1)}&={-}G_2^\text{(4-1)} G_5^\text{(4-1)} {-} G_6^\text{(4-1)}\,,\\
\notag         G_9^\text{(4-1)}&= G_1^\text{(4-1)}G_5^\text{(4-1)}-G_4^\text{(4-1)}\,. 
     \end{align}
     We now use  eq.~\eqref{constrantR} to fix the entries $r_i^\text{(4-1)}$ of $\bs{R}^\text{(4-1)}$. In our case eq.~\eqref{constrantR} takes the form
\begin{equation}\label{eq:banana_constraint}
    \big(\bs{R}^\text{(4-1)}\big)^2 \,\bs{\tilde{\Omega}}_1^\text{(4-1)}-\varphi_{\bK_7}\Big(\bs{\tilde{\Omega}}_1^\text{(4-1)}\Big)\,\big(\bs{R}^\text{(4-1)}\big)^2=0\,,
\end{equation}
where $\bs{\tilde{\Omega}}_1^\text{(4-1)}=(\bs{U}_\text{t, MC}^{\text{(4-1)},\,(0)})^{-1}\bs{A}_\text{MC}^\text{(4-1)}\bs{U}_\text{t, MC}^{\text{(4-1)},\,(0)}$ is independent of the functions $t_{i,j}^\text{(4-1)}$.
We interpret eq.~\eqref{eq:banana_constraint} as a set of polynomial constraints for the functions $r_i^\text{(4-1)}$, or a linear systems for the entries of $\big(\bs{R}^\text{(4-1)}\big)^2$. 
The latter matrix takes a similar form to $\bs{R}^\text{(4-1)}$,
\begin{equation}
    (\bs{R}^\text{(4-1)})^2=\begin{pmatrix}1&0&0&0&0&0&0\\
    f_1^\text{(4-1)}&1&0&0&0&0&0\\
     f_2^\text{(4-1)}&0&1&0&0&0&0\\
     f_3^\text{(4-1)}&0&0&1&0&0&0\\
     f_4^\text{(4-1)}&f_5^\text{(4-1)}&f_6^\text{(4-1)}&0&1&0&0\\
     f_7^\text{(4-1)}&f_8^\text{(4-1)}&f_5^\text{(4-1)}&0&0&1&0\\
     f_9^\text{(4-1)}&f_7^\text{(4-1)}&f_4^\text{(4-1)}&f_3^\text{(4-1)}&f_2^\text{(4-1)}&f_1^\text{(4-1)}&1
     \end{pmatrix}\,,
\end{equation}
with the added benefit of admitting a simple relationship to the functions $t_{i,j}^\text{(4-1)}$
\beq\bsp
    f_1^\text{(4-1)}&=t_{4,3}^\text{(4-1)}+t_{9,8}^\text{(4-1)}\,,\\
    f_2^\text{(4-1)}&=t_{5,3}^\text{(4-1)}+t_{9,7}^\text{(4-1)}\,,\\
    f_3^\text{(4-1)}&=t_{6,3}^\text{(4-1)}+t_{9,6}^\text{(4-1)}\,,\\
    f_4^\text{(4-1)}&=t_{7,3}^\text{(4-1)} + t_{5,3}^\text{(4-1)} t_{7,5}^\text{(4-1)} + t_{4,3}^\text{(4-1)} t_{8,5}^\text{(4-1)} + t_{9,5}^\text{(4-1)}\,,\\
    f_5^\text{(4-1)}&=t_{7,4}^\text{(4-1)}+t_{8,5}^\text{(4-1)}\,,\\
    f_6^\text{(4-1)}&=2 t_{7,5}^\text{(4-1)}\,,\\
    f_7^\text{(4-1)}&=t_{5,3}^\text{(4-1)}t_{7,4}^\text{(4-1)} +t_{8,3}^\text{(4-1)}  + t_{4,3}^\text{(4-1)}t_{8,4}^\text{(4-1)} +t_{9,4}^\text{(4-1)}\,,\\
    f_8^\text{(4-1)}&=2t_{8,4}^\text{(4-1)}\,,\\
    f_9^\text{(4-1)}&=2t_{5,3}^\text{(4-1)}t_{7,3}^\text{(4-1)} + 2 t_{4,3}^\text{(4-1)}t_{8,3}^\text{(4-1)} + (t_{6,3}^\text{(4-1)})^2 + 2 t_{9,3}^\text{(4-1)}\,.
\esp\eeq
In the ancillary file {\bf{\texttt banana.nb}}, we give the explicit relations between the functions $r_i^\text{(4-1)}$ and $f_i^\text{(4-1)}$, and both can be expressed in terms of the elements of $\cFss$.
As an example, the simplest explicit expression for an $f_i^\text{(4-1)}$ in terms of periods and their derivatives is:
\begin{align}
f_3^\text{(4-1)}=-\dfrac{2 i \sqrt{5} \psi_0}{x_1 (\mathfrak{J}_{1,2} \mathfrak{J}_{2,1} -\mathfrak{J}_{1,1} 
   \mathfrak{J}_{2,2} )}\biggl[& -\mathfrak{J}_{1,1} \mathfrak{J}_{1,2}  (\mathfrak{J}_{2,2} -1) x_2\notag\\
   &-\mathfrak{J}_{2,2}  x_1(2 \mathfrak{J}_{1,1}  x_1-\mathfrak{J}_{1,1} 
   x_2+\mathfrak{J}_{1,1} +\mathfrak{J}_{2,1}  x_1)\\
   &+\mathfrak{J}_{1,2}  \mathfrak{J}_{2,1} x_1 (2 x_1+x_2+1)\biggr]\notag\,,
   \end{align}
   where $\psi_0$ is the period that is holomorphic at the MUM-point $x_1=x_2=0$ and the Jacobian $\mathfrak{J}$ is defined by
\begin{equation}
  \label{eq:jacobian}  \begin{pmatrix}\text{d}x_1\\\text{d}x_2\end{pmatrix}=\begin{pmatrix}\mathfrak{J}_{1,1}&\mathfrak{J}_{1,2}\\
    \mathfrak{J}_{2,1}&\mathfrak{J}_{2,2}\end{pmatrix}\begin{pmatrix}\text{d}\tau_1\\\text{d}\tau_2\end{pmatrix}\,,
\end{equation}
where the canonical variables 
\beq\tau_i=\frac{\psi_1^{(i)}}{\psi_0} \,,
\eeq
are defined as a ratio between the periods $\psi_1^{(1)}$, $\psi_1^{(2)}$, which diverge with a single power of a logarithm at the MUM-point, and the holomorphic period $\psi_0$.
   We provide explicit expressions for all $f_i^\text{(4-1)}$ and $r_i^\text{(4-1)}$ in the ancillary file {\bf{\texttt banana.nb}}.
   To conclude, we find that $\cFeps$ is generated by the six functions $G_i^\text{(4-1)}$ from eq.~\eqref{eq:orthPartBanana}. 

\section{Beyond maximal cuts and self-duality}
\label{sec:non-maximal_cuts}

So far we have focused on self-dual scenarios, which include the maximal cuts of Feynman integrals. Our main result is Proposition~\ref{prop:main_1}, which allows us to identify those $\eps$-functions that can be expressed in terms of elements from $\cFss$. In applications to Feynman integrals, however, one needs to solve the full system, including the contributions from non-maximal cuts, where self-duality is lost. In this section we give some perspective on how our method generalises beyond the maximal cut (for an alternative approach how to constrain entries of the differential equation matrix beyond the maximal cut, see ref.~\cite{Pogel:2024sdi}). We discuss two examples where we can obtain relations between $\eps$-functions for a non-self-dual scenario by relating them to a self-dual case. In the first example self-duality is recovered by taking some limit. In the second example, we work in the opposite direction, and we construct an associated self-dual scenario by adding terms to the twist and taking limits at the end. This latter approach is in principle  applicable to Feynman integrals where we add analytic regulators to all propagators.

\subsection{A non-self-dual deformed CY operator}
\label{subsec:CY_non_self-dual}

Let us start by discussing the case of a deformed CY operator $\mathcal{L}_{\eps}$ that is not self-dual for $\eps\neq 0$. Indeed, it was pointed out in ref.~\cite{Duhr:2025lbz} that, even though the CY operator $\mathcal{L}_{\eps=0}$ is essentially self-adjoint by definition, the resulting deformed operator is typically not, unless additional constraints are imposed on some parameters. Focusing on deformed CY operators of hypergeometric type, a notable exception are those where the indices $\bs{\alpha}$ and $\bs{\beta}$ of the hypergeometric function are integers or half-integers for $\eps=0$~\cite{Duhr:2024xsy,Duhr:2025lbz}. This was the case in section~\ref{CY1V}, cf.~eqs.~\eqref{cyparameters} and~\eqref{eq:4f3parameters}. 

In this section we discuss an example of a deformed CY operator of order three that is not self-dual.
More specifically, we consider a hypergeometric ${}_3F_2$ function defined by the twist \eqref{phi32} with
\beq\bsp
\label{cyparametersBSD}
 \bs{\alpha}&=\{\alpha_0,\alpha_1,\alpha_2\}=  \left\{\frac{1}{4} +n_1\eps, \frac{2}{4}+n_2\eps, \frac{3}{4} + n_3\eps\right\}\,,\\
 \bs{\beta} &= \{\beta_1,\beta_2\}=\left\{ \, 1+ m_1\eps, 1+m_2\eps\,\right\}\,.
\esp\eeq
As an initial basis for the twisted cohomology group we choose
the same basis as in eq.~\eqref{initialbasisCY1}.
The corresponding deformed CY operator $\mathcal{L}^\lb_{\eps}$ is not self-dual, but self-duality only holds for $n_1=n_3$.
Nevertheless, we observe that, if we pick the dual basis as
\begin{align}
\check{\varphi}_j^\lb=\frac{\Gamma \left(j-\frac{1}{4}\right) \Gamma \left(\frac{1}{4}-m_2 \eps+n_1 \eps\right)}{\Gamma \left(j-\frac{3}{4}\right) \Gamma \left(\frac{3}{4}-m_2 \eps+n_1 \eps\right)} \frac{(x_1-1)x_2 (x_1-x_2)}{x_1 (\lambda x_2 -1)} \left(\varphi^\lb_{j}\right)_{\eps\rightarrow -\eps, n_1\leftrightarrow n_3} \,,
\end{align}
then, as a consequence of the symmetry of $ _3F_2\left(\alpha_0,\alpha_1,\alpha_2;\beta_1,\beta_2;\lambda\right)$  under exchanging $\alpha_0$ and $\alpha_2$, the period matrix and its dual are related by 
 \begin{align}\label{eq:3F2_A_symmetry}
\bs{\check{P}}^\lb={\bP^\lb}|_{\eps \rightarrow -\eps,n_1\leftrightarrow n_3} \,.
 \end{align}
 In particular, for $n_1=n_3$ we may choose the dual basis such that the period matrix and its dual are simply related by changing the sign of $\eps$.
 
 In ref.~\cite{Duhr:2025lbz} it was discussed how to construct a rotation matrix to a canonical basis. The matrix $\bUt$ has exactly the same form as in eq.~\eqref{eq:bUt_3F2}, and the $\eps$-functions $t_1$, $t_2$ and $t_3$ (and their derivatives) generate the function field $\cFeps$.
 We can also rotate the dual system into canonical form, and we denote the corresponding matrix by $\bs{\check{U}_\textrm{{t}}}$. It has the same general form as in eq.~\eqref{eq:bUt_3F2}, but with a priori different $\eps$-functions, which we denote $\check{t}_i$. We can also determine the canonical intersection matrix. To this effect, we have determined the intersection matrix in the original bases. After rotation to the canonical bases, we find that the following relation numerically holds,
 \begin{align}\label{eq:Delta_3F2_A}
\bC^\lb_c=\bU\bC^\lb\check{\bU}^T=-\eps^4 \,\frac{\Gamma \left(\frac{3}{4}-n_1\eps \right) \Gamma \left(\frac{1}{4}-m_2 \eps +n_1 \eps\right)}{\Gamma \left(\frac{1}{4}-n_3 \eps \right) \Gamma \left(\frac{3}{4}-m_2 \eps +n_3 \eps \right)}\,\bK_3
\,,
\end{align}
 in agreement with eq.~\eqref{deltaCY1}. Equation~\eqref{eq:Delta_3F2_A} imposes the following relations between the entries of $\bUt$ and $\bs{\check{U}}_\textrm{t}$,
 \beq\bsp
\check{t}_3+t_1 &\,= \psi_0\,\frac{ m_1+m_2- \lambda (n_1+n_2+ n_3)}{\sqrt{1- \lambda}}\,,\\
\label{3f2NoSDRel}
\check{t}_1+t_3&\,= \psi_0\, \frac{m_1+m_2- \lambda(n_1+n_2+ n_3)}{\sqrt{1- \lambda}}\,,\\
\check{t}_2 + t_2&\,=t_1t_3-\psi_0\,t_1\,\frac{ m_1+m_2- \lambda(n_1+n_2+n_3)}{\sqrt{1- \lambda}}+m_1m_2\psi_0^2  \\
&\qquad - \lambda(n_1n_2+n_1n_3+n_2n_3)\psi_0^2 \,, \\\check{t}_{2,-1}-t_{2,-1}&\,=\frac{ \lambda}{4}(n_1-n_3)\psi_0^2\,,
\esp\eeq
where $\psi_0$ denotes the solution of the CY operator $\mathcal{L}_{\eps=0}^\lb$ that is holomorphic at the MUM-point $\lambda=0$.
These relations allow us to write the entries of $\bs{\check{U}}_\textrm{t}$ in terms of those from $\bUt$ and express the fact that the $\eps$-functions of a system and its dual are the same (see the discussion in section~\ref{sec:canon_int}). As expected, the lack of self-duality implies that from the canonical intersection matrix alone we cannot constrain the entries of $\bUt$. In particular, at this point we have not achieved any reduction in the number of generators for $\cFeps$.

 In this particular case, however, we can do better, despite the lack of self-duality.
 Indeed, as a consequence of our choice of basis and of eq.~\eqref{eq:3F2_A_symmetry}, we can easily see that we have the relation
\beq
\bs{\check{U}}_\textrm{t}=\bUt|_{\eps\rightarrow -\eps,n_1\leftrightarrow n_3}\,,
\eeq
which implies
\beq
\check{t}_i = t_i|_{n_1\leftrightarrow n_3}\,.
\eeq
We may then use the first relation in eq.~\eqref{3f2NoSDRel} to express $t_1$ in terms of $\check{t}_3$, which can itself be expressed in terms of $t_3|_{n_1\leftrightarrow n_3}$. We obtain in this way an explicit relation between the $\eps$-functions $t_1$ and $t_3$, and so we see that we can generate $\cFeps$ using $t_1$ and $t_2$ alone.

From the preceding discussion it follows that $\cFeps$ can be generated by at most two functions. In section~\ref{sec:deformed_CY} we have argued that, in a self-dual scenario, the number of generators for $\cFeps$ is bounded by $\dim \textrm{OT}^{\bK_N}$, which is equal to 1 for $N=3$. It is thus interesting to ask if we may further reduce the number of generators also in our non-self-dual scenario. We now argue that this is not possible, and that the two functions $t_1$ and $t_2$ are independent and do not lie in $\cFss$. In order to see this, we start by giving a concrete description of the function field $\cFss$. This is possible because $\mathcal{L}_{\eps=0}^\lb$ is a CY operator of degree three, and every such operator is the symmetric square of a CY operator of degree two, which admits a modular parametrisation~\cite{Doran:2005gu,Bogner:2013kvr,BognerThesis}. More specifically, let us change variables to the canonical variable
\beq\label{eq:tau_def_3F2_2}
\tau = \frac{\psi_1(\lambda)}{\psi_0(\lambda)}\,,
\eeq
where $\psi_1(\lambda)$ is the solution of $\mathcal{L}_{\eps=0}^\lb$ that diverges like a single power of a logarithm at the MUM-point $\lambda=0$. 
Its inverse is the mirror map, which in this case can be expressed in terms of modular functions for the congruence subgroup $\Gamma_0(2)$,
\beq
\label{z3F2}
    \lambda = \frac{256\,A(\tau)}{(1 + 64 A(\tau))^2}\,,
\end{equation}
where $A(\tau)$ is a Hauptmodul for $\Gamma_0(2)$, which can be expressed in terms of the Dedekind eta function
\beq
\label{def:A}
A(\tau) = \frac{\eta(2\tau)^{24}}{\eta(\tau)^{24}}\,,\textrm{~~~with~~~} \eta(\tau) = e^{i\pi\tau/12}\prod_{n=1}^\infty(1-e^{2\pi in\tau})\,.
\eeq
The function $\psi_0(\tau) := \psi_0(\lambda(\tau))$ then defines an Eisenstein series of weight two for $\Gamma_0(2)$. Since $\cFss$ is generated by $\psi_0$ and its derivatives, we conclude that we can identify $\cFss$ with the function field generated by (meromorphic) quasi-modular forms for $\Gamma_0(2)$. More details can be found in appendix~\ref{appmodular}. Here it suffices to say that the functions $t_1$ and $t_2$ can be expressed in terms of iterated integrals of the magnetic modular forms for $\Gamma_0(2)$ defined in ref.~\cite{Bonisch:2024nru} (see appendix~\ref{appmodular} for the explicit expressions). These magnetic modular forms lie in $\cFss$. It  follows from the result of ref.~\cite{Broedel:2021zij} that these magnetic modular forms do not admit primitives that are themselves quasi-modular, or said differently, the primitives of these magnetic modular forms do not lie in $\cFss$. From there we can conclude that $t_1$ and $t_2$ do not lie in $\cFss$. Finally, from the explicit expressions in appendix~\ref{appmodular}, we can see that $t_1$ and $t_2$ are independent (because they involve primitives of magnetic modular forms of different weights and with poles at different locations). Hence, we conclude that the two functions $t_1$ and $t_2$ form a minimal set of generators for $\cFeps$. 

So far the discussion applies to arbitrary (non-zero) values of the $n_i$. We know from eq.~\eqref{eq:dim_OTKN} with $N=3$ that in the self-dual case $n_1=n_3$ we expect at most one generator for $\cFeps$. Hence, we expect that for $n_1=n_3$ the number of generators must drop. Let us briefly discuss this reduction in the number of generators. First, in the self-dual case we may apply Proposition~\ref{prop:main_1} and write $\bUt^{(0)}$ as a product of a $\bK_3$-orthogonal matrix $\bO$ and a $\bK_3$-symmetric matrix $\bR$. The form of $\bR$ is identical to eq.~\eqref{gensym3F2} (with $r_i^\lbt$ replaced by $r_i^\lb$). Two of the relations in eq.~\eqref{3f2NoSDRel} are now trivially satisfied, such that the relations reduce to
\beq\bsp\label{relsd}
t_1+t_3&\,=\psi_0\,\frac{m_1+m_2- \lambda (n_2+2 n_3)}{\sqrt{1- \lambda}\,,}=2 r_1^\lb\,,\\
t_2+\frac{1}{2}t_1^2 &=\psi_0^2\,\frac{ m_1 m_2- \lambda n_3 (2 n_2+n_3)}{2}=(r_1^\lb)^2+2r_2^\lb\,.
\esp\eeq
From the last equation we can see that $t_2+\frac{1}{2}t_1^2\in \mathcal{F}_{\mathrm{ss}}$. Thus, we have found a combination of the two generators $t_1$ and $t_2$ that lies in $\cFss$, and we can eliminate $t_2$ in terms of $t_1$, thereby reducing the number of generators to be (at most) one. This generator may be chosen as the combination of the $t_i$ that appears in the $\bK_3$-orthogonal matrix, which is $G:= \frac{1}{2}(t_1- t_3)$. The fact that $G\notin\cFss$ can be proven in the same way as for $t_1$ and $t_2$ above. 

\subsection{The Legendre family of elliptic curves with an additional puncture}
\label{subsec:legendrePuncture}

In this section we discuss another example of a non self-dual scenario where we can identify $\eps$-functions and the relations between them. 
We  study a hypergeometric function defined by the twist
\begin{equation}
\label{eq:twistLegendrePuncture}
    \Phi^{\scriptscriptstyle(\mathrm{Leg})}=x^{-\frac{1}{2}+a_1\eps}(1-x)^{-\frac{1}{2}+a_2\eps}(\lambda-x)^{-\frac{1}{2}+a_3\eps} (t-x)^{0} \,,
\end{equation}
where $x$ is the integration variable and $\lambda$ and $t$ are external parameters, and we included a factor of $(t-x)^{0}$ to illustrate that this pole will appear in our basis forms, but it is not regulated by the twist $\Phi^{\scriptscriptstyle(\mathrm{Leg})}$. The appearance of such an unregulated singularity is the hall-mark of the non-self-dual setup, and requires the use of relative twisted cohomology~\cite{matsumoto_relative_2019-1,matsumotoFD,matsumotoRelative,Caron-Huot:2021xqj,Caron-Huot:2021iev,Brunello:2023rpq}. As a consequence, it is in general not possible to identify a basis such that the dual period matrix can be identified with the period matrix with $\eps$ replaced by $-\eps$. Moreover, some of the entries of the period matrix will involve integrands that have a pole with non-vanishing residue at $x=t$. Note that this is precisely the setup typically encountered for Feynman integrals (or non-maximal cuts, where only residues at a subset of the propagator poles have been taken). This example thus serves as a first step towards understanding the relations between $\eps$-functions beyond the maximal cut.

For $\eps=0$ the twist in eq.~\eqref{eq:twistLegendrePuncture} reduces to the one defining periods of the Legendre family of elliptic curves, defined by
\begin{equation}
    y^2=x(x-1)(x-\lambda) \,.
\end{equation}
Hence, if the basis only contained the regulated singularities, then we would be in the setup of an odd elliptic curve discussed in section~\ref{subsec:higherGenus} for $g=1$. From eq.~\eqref{eq:dim_OT_higher_g} with $g=1$, we see that we do not need to introduce any $\eps$-functions. As we will see, the additional unregulated pole will lead to an $\eps$-function required to define the canonical basis.
 We note that the situation discussed here is very similar to the case of an even elliptic curve discussed in section~\ref{subsec:higherGenus}. Indeed, according to the definition in section~\ref{subsec:higherGenus}, an even elliptic curve can be understood as a family with an additional puncture at infinity. However, in section~\ref{subsec:higherGenus} that singularity is regulated by the twist (with exponent $-(a_1+a_2+a_3+a_4)\eps$). Here infinity is a branch point of the elliptic curve, and the additional puncture is at $x=t$. Importantly this puncture is not regulated by the twist. 

 Let us now discuss how we can obtain a canonical basis for our case. Since so far a canonical basis for this class of functions has not been derived yet, we spend some time discussing details. 
We will first introduce a regulator $a_4\eps$ for the last factor in the twist, i.e., we consider the twist
\begin{equation}
\label{eq:twistLegendrePuncture_rho}
    \Phi_{a_4}^{\scriptscriptstyle(\mathrm{Leg})}=x^{-\frac{1}{2}+a_1\eps}(1-x)^{-\frac{1}{2}+a_2\eps}(\lambda-x)^{-\frac{1}{2}+a_3\eps} (t-x)^{a_4\eps} \,.
\end{equation}
This twist defines a family of Appell $F_1$ functions. Since all singularities are regulated by the twist for $a_4\neq0$, we are now in a self-dual scenario. We proceed as before, and we construct bases for the twisted cohomology group and its dual such that the corresponding period matrices just differ by changing the sign of $\eps$. Note that, upon performing a M\"obius transformation that sends $t$ to infinity, we recover the case of an even elliptic curve discussed in section~\ref{subsec:higherGenus}. From eq.~\eqref{eq:higher_genus_even_dim} we then expect that we need at most one $\eps$-function to generate $\cFeps$. We will see that the bound is also saturated in this example, and the generator does not lie in $\cFss$.

\paragraph{The canonical basis for $a_4\neq 0$.}
\newcommand{\legendre}{{\scriptscriptstyle(\mathrm{Leg})}}

We start by discussing the class of hypergeometric functions defined by the twist in eq.~\eqref{eq:twistLegendrePuncture_rho} for $a_4\neq0$. 
We choose the following basis for the twisted cohomology group
\begin{equation}\bsp
\label{eq:legendrePunctureBasis}
\varphi_1^\legendre &\,= \rd x\,,\qquad 
\varphi_2^\legendre = \frac{(1-2a_3\eps)\rd x}{2(x-\lambda)}\,,\qquad
\varphi_3^\legendre = \frac{y_t\,\rd x}{x-t}\,,
\esp
\end{equation}
with $y_t=\sqrt{t(t-1)(t-\lambda)}$.
These basis elements correspond to first, second and third kind differentials on the elliptic curve for $\eps=0$. We choose a dual basis so that the dual period matrix is obtained from the period matrix by simply changing the sign of $\eps$,
\beq\bsp\label{eq:F1_dual_basis}
\check{\varphi}_1^{\legendre}&\, = \frac{\rd x}{x(x-1)(x-\lambda)}\,,\\ 
\check{\varphi}_2^{\legendre} &\,= \frac{(1+2a_3\eps)\rd x}{2x(x-1)(x-\lambda)}\,,\\ 
\check{\varphi}_3^{\legendre} &\,= \frac{y_t\rd x}{x(x-1)(x-\lambda)(x-t)}\,.
\esp\eeq

We now construct the rotation to a canonical basis. We perform a change of basis via a rotation matrix as in eq.~\eqref{Usplitting}, with
\beq\bsp
\bUss(\lambda,t) &\,=  \begin{pmatrix}
        \frac{1}{\psi_0} & 0 & 0 \\
        \frac{\lambda\psi_0-\eta_0}{8\pi i} & \frac{\lambda(\lambda-1)\psi_0}{4\pi i} & 0 \\
        0 & 0 & 1
    \end{pmatrix}\,,\\ 
    \bUeps(\eps) &\,=     \begin{pmatrix}
        \eps & 0 & 0 \\ 0 & 0 & \eps \\ 0 & 1 & 0
    \end{pmatrix}\,,\\
        \bUt(\lambda,t,\eps)&\,=\begin{pmatrix}
        1 & 0 & 0 \\ t_1(\lambda,t) & 1 & 0 \\ t_2(\lambda,t) & t_3(\lambda,t) & 1
    \end{pmatrix}\,.
    \esp\eeq
Here, the period $\psi_0$ and quasi-period $\eta_0$ are defined by
\begin{equation}
\psi_0(\lambda)=\oint_{\gamma_a}\frac{\rd x}{y},\qquad \eta_0(\lambda)=\oint_{\gamma_a}\frac{x\rd x}{y} \,,
\end{equation}
with $\gamma_a$ the cycle encircling the branch cut between $0$ and $1$. The form of $\bUt$ is consistent with eq.~\eqref{eq:higher_genus_even_Ut} for $g=1$. We see that the function field $\cFeps$ is generated by the three $\eps$-functions $t_1$, $t_2$ and $t_3$. 
Since we are in a self-dual scenario, we expect, however, from eq.~\eqref{eq:higher_genus_even_dim} that $\cFeps$ is generated by at most one element. We now determine the relation between the functions $t_i$ and identify this single generator.

We have computed the intersection matrix in the original bases in eqs.~\eqref{eq:legendrePunctureBasis} and~\eqref{eq:F1_dual_basis}. We find that, after rotation to the canonical basis, the canonical intersection matrix is given by
\begin{equation}
    \bDelta=\begin{pmatrix}
        0 & 0 & 1 \\ 0 & -\frac{4\pi i}{a_4} & 0 \\ 1 & 0 & 0
    \end{pmatrix} \,.
\end{equation}
Note that this form agrees with eq.~\eqref{eq:Delta_higher_genus_even}, up to rescaling of one of the basis elements. However, since we treat $a_4$ as a regulator that will be sent to zero at the end, we prefer to keep the dependence in the canonical intersection matrix. 

We can now use Proposition~\ref{prop:main_1} to factorise $\bUt$ into a $\Delta$-orthogonal and a $\Delta$-symemtric matrix,
\beq
\bUt = \bO\bR\,,
\eeq
with
\beq
\bO =\begin{pmatrix}
        1 & 0 & 0 \\ G_1 & 1 & 0 \\ \frac{a_4 G_1^2}{8\pi i} & \frac{G_1 a_4}{4\pi i}  & 1
    \end{pmatrix}\,,\qquad \bR = \begin{pmatrix}
        1 & 0 & 0 \\ r_1 & 1 & 0 \\ r_2 & -\frac{a_4 r_1}{4\pi i} & 1
    \end{pmatrix} \,,
\end{equation}
where the functions $r_1,r_2$ can be expressed in terms of the period $\psi_0$ by
\beq\bsp
    r_1&\,=\frac{1}{2y_t}t(t-1)\psi_0 \,, \\
    r_2&\,=-\frac{\psi_0^2}{32\pi i}\left[ 4a_1(\lambda-1)+4\lambda a_2+4a_3(2\lambda-1)+a_4\,\frac{3t(t-1)-4\lambda(\lambda-1)}{t-\lambda}\right]\,.
\esp\eeq
The function $G_1(\lambda,t)$, instead, is not fixed by our considerations. It can be cast in the form of a combination of (incomplete) elliptic integrals. An alternative representation will be presented below. We thus see that, in agreement with eq.~\eqref{eq:higher_genus_even_dim}, $\cFeps$ is generated by (at most) one function. 

After having obtained the canonical basis for $a_4\neq0$, we can now extract the canonical basis in the limit where the regulator $a_4$ is sent to zero. On the non-dual side this is straightforward, because nothing qualitatively changes in the limit. On the dual side it is convenient to replace the third basis element by $\tfrac{\rd x}{x-t}$, since, by taking leading terms in a small $a_4$-expansion, its intersection numbers reduce to those one would obtain from a delta-form \cite{Brunello:2023rpq}. It is then straightforward to construct a dual canonical rotation starting from the known canonical rotation for the self-dual basis. The constructed rotation contains two dual $\eps$-functions, which we can show to be contained in $\cFeps$, as expected, and so they can both be expressed in terms of $G_1$ and functions from $\cFss$.

\paragraph{Identifying the $\eps$-function.}
Let us now discuss how we can identify an expression for the $\eps$-function $G_1$ in terms of a special function known in the mathematics literature.
In order to do that, we use the fact that in this particular case we could also evaluate all the master integrals using direct integration techniques in terms of elliptic multiple polylogarithms (eMPLs)~\cite{brownLevinEMPls,Broedel:2014vla,Broedel:2017kkb,enriquezZerbiniEMPLs}. We stress that the differential equations and direct integration approaches are complementary: while direct integration works very well at low orders in $\eps$ and directly provides results in terms of a well-established class of transcendental functions, it becomes increasingly cumbersome at higher orders. Differential equations, instead, allow one to easily obtain explicit results even at very high orders in $\eps$, but the appropriate class of transcendental functions is not visible at this point, because the differential equation depends on the $\eps$-function $G_1$. By combining the two approaches, we can express the differential equation matrix in terms of the integration kernels that define eMPLs, thereby easily allowing for a solution in terms of eMPLs at any desired order in the $\eps$-expansion.

The direct integration approach starts by selecting a set of master integrals that is finite at $\eps=0$. The main idea is then to perform all integrations in terms of eMPLs. We will not review this procedure here, as it is irrelevant to understand the result, and we refer instead to the literature~\cite{Broedel:2019hyg,Campert:2020yur,Stawinski:2023qtw}. Here it suffices to say that we pick the following set of master integrals
\begin{align}
    M_1&=\int_0^1\frac{\rd x}{y}\,\Phit^{\scriptscriptstyle(\mathrm{Leg})}_{a_4}(x) \,, \nonumber\\
    M_2&=\frac{\sqrt{\lambda}}{2}\int_0^1\rd x\,\Phit^{\scriptscriptstyle(\mathrm{Leg})}_{a_4}(x)\Psi(x)\,, \\
    M_3&=\int_0^1\frac{y_t \rd x}{y(x-t)}\Phit^{\scriptscriptstyle(\mathrm{Leg})}_{a_4}(x) \nonumber\,,
\end{align}
with
\begin{equation}
    \Psi(x)=\frac{1}{y}\left[-\frac{2}{\sqrt{\lambda}}x+\frac{2(-2-2\lambda+\lambda^{3/2}+\lambda^{5/2})\omega_1+6\eta_1}{3\lambda^{2}\omega_1} \right] \,.
\end{equation}
Here we decomposed the twist as $\Phi^{\scriptscriptstyle(\mathrm{Leg})}_{a_4}(x)=\tfrac{1}{y}\Phit^{\scriptscriptstyle(\mathrm{Leg})}_{a_4}(x)=\tfrac{1}{y}+\mathcal{O}(\eps)$ to make the square root defining the elliptic curve manifest. These integrals can straightforwardly be rotated into the starting basis from before. In particular $M_1,M_3$ are precisely the first and third integrals in our earlier basis, while $M_2$ is some simple linear combination. We can hence easily find the rotation that relates these integrals to our canonical basis. The reason for the somewhat involved expression for the second integral $M_2$ is that it is particularly convenient to map to the torus description of the elliptic curve using the relations derived in ref.~\cite{Broedel:2017kkb}.

Integrating the above expressions is straightforward using the methods from refs.~\cite{Broedel:2019hyg,Campert:2020yur,Stawinski:2023qtw}. We simply expand the integrand in $\eps$, producing additional logarithms. The integrand can then be mapped to the torus description and integrated up order by order. The only subtlety here is that for the integration of $M_2$ we need to integrate by parts to map to the torus description. This leads to singularities after expanding in $\eps$ which requires us to add and subtract suitable counterterms. Note that, since our goal is to gain information on the $\eps$-function that enters the differential equation matrix, it is sufficient to evaluate all master integrals up to the order in $\eps$ that involves eMPLs of length at most one.

With explicit expressions for $M_1,M_2,M_3$ and the rotation to a canonical basis at hand, we can write down expressions for the canonical basis integrals $J_1,J_2,J_3$ in terms of eMPLs in an $\eps$-expansion. Since these are canonical integrals, they should consist of pure combinations of eMPLs at every order in the expansion \cite{Broedel:2018qkq}. However, also the $\eps$-function will enter the expansion, and so the requirement of purity will allow us to fix it.
Explicitly we find that the pure integral $J_2$ is given by
\begin{equation}
    J_2=\eps\left(\frac{1}{2}G_1(\lambda,t)-g^{(1)}(z_t,\tau)+\frac{1}{4y_t}t(t-1)\psi_0 \right) +\mathcal{O}(\eps^2)\,,
    \label{eq:J2_Legendre}
\end{equation}
where $g^{(1)}$ is defined as
\begin{equation}
    g^{(1)}(z,\tau)=\frac{\pi \theta_1'(\pi z,\tau)}{\theta_1(\pi z,\tau)}-\frac{\pi\theta_1''(0,\tau)}{2\theta_1'(0,\tau)} \,,
\end{equation}
in terms of the odd Jacobi theta function $\theta_1$ and its derivatives with respect to the first argument.
Here $z_t$ is the point on the torus $\mathbb{C}/(\mathbb{Z}+\mathbb{Z}\tau)$ defined by 
\begin{equation}
    z_t=\frac{1}{\psi_0}\int_{\infty}^t\frac{\rd x}{y} \,,
\end{equation}
where $\tau$ is the modular parameter
\begin{equation}
    \tau=\frac{\psi_1}{\psi_0},\qquad \psi_1=\oint_{\gamma_b}\frac{\rd x}{y} \,,
\end{equation}
with $\gamma_b$ the contour encircling the branch points $\lambda,1$.
Since $J_2$ is an element of a canonical basis, we expect it to only involve pure functions at every order in $\eps$. The expression in eq.~\eqref{eq:J2_Legendre}, however, is not pure. 
We hence expect that the $\eps$-function $G_1(\lambda,t)$ is such that it cancels the other terms, possibly up to pure functions. By explicitly checking the defining differential equations for $G_1(\lambda,t)$ we find that no additional pure functions are needed, and we have
\begin{equation}
    G_1(\lambda,t)=2g^{(1)}(z_t,\tau)-\frac{1}{2y_t}t(t-1)\psi_0 \,.
\end{equation}
Hence, we see that the $\eps$-function $G_1$ is just the $g^{(1)}$ function, up to a term proportional to the period $\psi_0$. 

\section{Conclusions and outlook}
\label{sec:conclusions}

In this paper we continued our study of the properties of canonical bases obtained via the method of ref.~\cite{Gorges:2023zgv} through the lens of twisted cohomology. We have built on the result of ref.~\cite{Duhr:2024xsy}, which states that the intersection matrix computed in canonical bases is a constant, and the observation of ref.~\cite{Duhr:2024uid}, which implies that this constancy can be used to derive polynomial relations relating $\eps$-functions on the maximal cut, or more generally for a self-dual scenario (see also ref.~\cite{Pogel:2024sdi} for a related idea). Applying these ideas in practice, however, is hampered by the fact that computing the intersection matrix in the original basis can be prohibitively complex, and solving the ensuing non-linear polynomial relations can be complicated. In order to improve on this situation, we extended the results of our previous papers in two directions: First, we studied novel properties of the canonical intersection matrix and how it can be obtained without having to explicitly evaluate the intersection matrix  in the original basis. Second, we proved a decomposition theorem that allows us to identify a priori the $\eps$-functions that can be expressed in terms of algebraic functions and (derivatives of) periods, and the relevant polynomial constraints can always be reduced to linear relations. The power of this approach has already been demonstrated in ref.~\cite{Duhr:2025kkq} in the context of the three-loop banana integral with four distinct masses. We presented several new examples, including various classes of hypergeometric functions already considered in ref.~\cite{Duhr:2025lbz} and a four-loop banana integral with two distinct masses discussed in ref.~\cite{Maggio:2025jel}. We also discussed on two examples how our ideas can be extended to non self-dual scenarios.

The results of our paper show  that twisted cohomology may teach us important lessons about the structure of canonical differential equations, and that there may still be interesting structures to be uncovered in the intersection of these two fields which are often studied in isolation from each other. In the future, we are planning to study the interplay between symmetries and the different pairings in twisted (co)homology, the recently defined filtrations on canonical bases inspired by Hodge theory from ref.~\cite{e-collaboration:2025frv}, as well as the connection to the equivariant iterated Eisenstein integrals from ref.~\cite{Brown:2020mhp,Dorigoni:2022npe,Dorigoni:2024oft,Dorigoni:2024iyt,Frost:2025lre}.

\acknowledgments

The authors are grateful to Gaia Fontana, Federico Gasparotto, Christoph Nega, Sid Smith and Lorenzo Tancredi for dicussions.
The work of CS is supported by the CRC 1639 ``NuMeriQS'', and the work of CD, FP, SM and SFS is funded by the European Union
(ERC Consolidator Grant LoCoMotive 101043686). Views
and opinions expressed are however those of the author(s)
only and do not necessarily reflect those of the European
Union or the European Research Council. Neither the
European Union nor the granting authority can be held
responsible for them. 
YS acknowledges support from the Centre for Interdisciplinary Mathematics at Uppsala University and partial support by the European Research Council under ERC Synergy Grant MaScAmp
101167287.

\begin{appendix}
\section{Some magnetic modular forms}\label{appmodular}
In this appendix we show the explicit expressions for the $\eps$-functions introduced for the example in section~\ref{subsec:CY_non_self-dual} in terms of iterated integrals over meromorphic modular forms. We show the results as functions of the canonical variable $\tau$ introduced in eq.~\eqref{eq:tau_def_3F2_2}.

For $t_1$ we obtain 
\beq\bsp
    t_1(\tau)=&-\frac{(64\, A(\tau)+1) \phi_0(\tau)}{192 \,A(\tau)-3}(m_1+m_2)-\frac{256\, A(\tau)\, \phi_0(\tau)}{3-12288 \,A(\tau)^2}(n_1+n_2+n_3)\\ &-\int^{\tau}_{i\infty} \!\!\!\rd\tau'\int^{\tau'}_{i\infty}\!\!\! \rd\tau''\, f_1(\tau'')\,,
\esp\eeq
where $A$ is defined in eq.~\eqref{def:A} and $\phi_0$ is an Eisenstein series of weight two for the congruence subgroup $\Gamma_0(2)$,
\beq 
\phi_0(\tau) = \theta_2^4(\tau)+\theta_3^4(\tau) \,,
\eeq
where $\theta_2$ and $\theta_3$ are Jacobi $\theta$ functions.
The integration kernel is
\beq\bsp\label{f1itint}
f_1(\tau)&\,=32\,\partial_{\tau} C_4(\tau)(n_1-n_3)+\frac{16}{3}C_{6,a}(\tau)(m_1+m_2-n_1-n_2-n_3)\\&\,+\frac{16}{3}C_{6,b}(\tau)(n_1-2\,n_2+n_3)\,.
\esp\eeq
The functions $C_4$, $C_{6a}$ and $C_{6b}$ are meromorphic modular forms for $\Gamma_0(2)$,
\beq\bsp
C_4&\,=\frac{A}{3(1+64\,A)}(4\,\phi_0^2-E_4)\,,\\
C_{6a}&\,=\frac{A(5+1408\,A+20480\,A^2)}{9(1-64\,A)^3}(8\,\phi_0^3+E_6)\,,\\
C_{6b}&\,=\frac{A(1-64\,A)}{9(1+64\,A)^2}(8\,\phi_0^3+E_6)\,,
\esp\eeq 
where $E_4$ and $E_6$ are Eisenstein series for the full modular group, respectively, of weight 4 and 6,
\beq\bsp
E_4(\tau) =&  1 + 240 \sum_{n=1}^{\infty} \frac{n^{3} q^{n}}{1 - q^{n}}\,,\\
E_6(\tau) =&  1 - 504 \sum_{n=1}^{\infty} \frac{n^{5} q^{n}}{1 - q^{n}}\,,
\esp\eeq
with $q:=e^{2\pi i \tau}$.
They have the property of being magnetic (see refs.~\cite{magnetic1, magnetic2, magnetic3, magnetic4, Bonisch:2024nru}).
The third function $t_3$ can be obtained from eq.~\eqref{f1itint} by using the relation in eq.~\eqref{3f2NoSDRel},
\begin{align}
\nonumber t_3(\tau)&\,=- t_1(\tau)-\frac{\phi_0(\tau) }{4096 A(\tau)^2-1}\,\left[(64 A(\tau)+1)^2 (m_1+m_2)-256 A(\tau) (n_1+n_2+n_3)\right]\\
 &\,-64(n_1-n_3)\int_{i\infty}^{\tau}\!\!\!\rd \tau'C_4(\tau')\,.
\end{align}
The last function cannot be obtained from the relations in eq.~\eqref{f1itint}. It has to be computed separately. We find that it can be expressed as  
\begin{align}
 t_2(\tau)=&-\frac{1}{2}t_1(\tau)^2+\frac{1}{2} \phi_0(\tau)^2m_1\,m_2-\frac{128\, A(\tau) }{(64\, A(\tau)+1)^2}\phi_0(\tau)^2(n_1\,n_2+n_1\,n_3+n_2\,n_3)\nonumber\\
 &-64\,(n_1-n_3)\,t_1(\tau)\int^{\tau}_{i\infty}\!\!\! \rd\tau'\, C_4(\tau')+64\,(n_1-n_3)\,\int_{i\infty}^{\tau}\!\!\! \rd\tau'\, C_4(\tau')\,t_1(\tau')\\
 &+64 \,n_2\,(n_1-n_3)\int_{i\infty}^{\tau} \!\!\!\rd\tau' \,\frac{ A(\tau')\, (64\, A(\tau')-1) }{(64\, A(\tau')+1)^3}\phi_0(\tau')^3 \,.\notag
\end{align}
In the self-dual limit there is only one genuinely new function. The $\Delta$-symmetric part is fixed by eq.~\eqref{relsd} while in the $\Delta$-orthogonal part the only new function is 
\begin{align}
\nonumber G(\tau)&=\frac{1}{2}\big(t_1(\tau)- t_3(\tau)\big)\\&=\frac{\phi_0(\tau)}{24576 A(\tau)^2-6} \left[(64 A(\tau)+1)^2 (m_1+ m_2)-256 A(\tau) (n_2+2 n_3)\right]\\
&-\frac{16}{3}\int_{i\infty}^{\tau}\!\!\!\rd \tau'\int_{i\infty}^{\tau'}\!\!\!\rd \tau''\left[(m_1+m_2-n_2-2 n_3)\, C_{6,a}(\tau'')-2\, (n_2-n_3)\,C_{6,b}(\tau'')\right]\,. \notag
\end{align}

\section{Yukawa three-couplings for the four-loop banana integral}\label{Yuk_appendix}
In this appendix we give a brief explanation on how to compute the Yukawa couplings as rational functions in $x_1$ and $x_2$.
They are needed to express the functions $f_i$ of section \ref{fourlooptwomassbaanna} as elements of $\mathcal{F}_\text{ss}$.
\paragraph{Periods of a CY-threefold.}

As explained in ref.~\cite{Maggio:2025jel}, at $\eps=0$, six master integrals can be identified as periods of a family of CY threefolds with the following Hodge numbers
\begin{equation}
    h^{(3,0)}=h^{(0,3)}=1\,,\quad h^{(2,1)}=h^{(1,2)}=2\,.
\end{equation}
In particular, the associated system of differential equations allows for six solutions which are the {\it periods} of this family.
We choose a Frobenius basis for these periods which reads
\begin{equation}
 \bs{\psi}=(\psi_0,\,\psi_1^{(1)},\psi_1^{(2)},\psi_2^{(2)},\psi_2^{(1)},\psi_3)^T \,,
 \end{equation}
 where the subscript denotes the logarithmic power in the Frobenius expansion.
 These periods satisfy some quadratic relations of the form
 \begin{equation}\label{quad_rels}
     \bs{\psi}\bs{\Sigma}\partial_{\bs{x}}\bs{\psi}=\left\{\begin{matrix}
      0\,\qquad\text{if}\quad r<3\,,\\ \hspace{-2.5mm}  C^{\bs x}_{\bs{k}}\quad\text{otherwise.}
     \end{matrix}\right.
 \end{equation}
where $\partial_{\bs x}:=\partial_{x_{k_1}}...\partial_{x_{k_r}}$ and $\bs{k}=(k_1,\dots,k_r)$.
The homology intersection matrix $\bs{\Sigma}$ reads
 \begin{equation}
     \bs{\Sigma}=\begin{pmatrix}
        \bs{0}&\bs{K}_3\\
        -\bs{K}_3&\bs{0}
     \end{pmatrix}\,.
 \end{equation}
 The functions $C^{\bs x}_{\bs{k}}$ are the so-called {\it Yukawa $r$-couplings}, and they are rational functions in the variables $\bs{x}$. 
 The Yukawa three-coupling will be of particular interest to us. Let us now consider the Wronskian matrix $\bs{P}^\text{CY3}$ of the system
 \begin{equation}
\bs{P}^\text{CY3}=\begin{pmatrix}
\bs{\psi}^T\\
 \partial_{x_1}\bs{\psi}^T\\
 \partial_{x_2}\bs{\psi}^T\\
 \partial_{x_1}^2 \bs{\psi}^T\\
 \partial_{x_2}^2 \bs{\psi}^T\\
 \partial_{x_1}^2\partial_{x_2}\bs{\psi}^T\end{pmatrix}\,.
 \end{equation}
As the periods of a CY threefold can be chosen self-dual, we can write down:
 \begin{equation}
     \bs{P}^\text{CY3}\bs{\Sigma}(\bs{P}^\text{CY3})^T=\bs{Z}^\text{CY3}(x_1,x_2)\,,
 \end{equation}
where one can use the quadratic relations from eq.~\eqref{quad_rels} to show that the cohomology intersection matrix takes the form
 \begin{equation}\label{ZZ}
 \bs{Z}^\text{CY3}=\begin{pmatrix}
 0&0&0&0&0&C^{\bs x}_{(1,1,2)}\\
 0&0&0&-C^{\bs x}_{(1,1,1)}&-C^{\bs x}_{(1,2,2)}&*\\
 0&0&0&-C^{\bs x}_{(1,1,2)}&-C^{\bs x}_{(2,2,2)}&*\\
 0&C^{\bs x}_{(1,1,1)}&C^{\bs x}_{(1,1,2)}&0&*&*\\
 0&C^{\bs x}_{(1,2,2)}&C^{\bs x}_{(2,2,2)}&*&0&*\\
 -C^{\bs x}_{(1,1,2)}&*&*&*&*&0
 \end{pmatrix} \,,
 \end{equation}
 where $*$ denotes entries depending on Yukawa $r$-couplings with $r>3$. 
 Also note that eq.~(\ref{ZZ}) is antisymmetric. 
Knowing the differential equation matrix $\bs{\Omega}^\text{CY3}$ for the maximal cut at $\eps=0$, (without the ISP $I_{1,1,1,1,1,-1,0,0,0,0,0,0,0,0}$) we can use the relation
\begin{equation}\label{connection_intersection}
    \bs{\Omega}^\text{CY3}\bs{Z}^\text{CY3}+\bs{Z}^\text{CY3}(\bs{\Omega}^{\text{CY3}})^T=\text{d}\bs{Z}^\text{CY3}\,,
\end{equation}
to determine (up to a constant factor fixed by the Frobenius expansion) the cohomology intersection matrix $\bs{Z}^{\text{CY3}}$ as a rational function in $x_1$ and $x_2$.
The Yukawa three-couplings read
 \begin{align}
     C^{\bs{x}}_{(1,1,1)}&=\dfrac{12(-64x_1^2+(x_2-1)^2)}{x_1^3(256x_1^2+(x_2-1)^2-32 x_1 (1 + x_2))(16 x_1^2+(x_2-1)^2-8x_1(1+x_2))} \,,\\
C^{\bs{x}}_{(1,1,2)}&=\dfrac{6(1+64x_1^2+20x_1(x_2-1)+(2-3x_2)x_2)}{x_1^2x_2(256x_1^2+(x_2-1)^2-32x_1(1+x_2))(16x_1^2+(x_2-1)^2-8x_1(1+x_2))} \,,\\
    C^{\bs{x}}_{(1,2,2)}&=\dfrac{24(1-10x_1+x_2)}{x_1x_2(256x_1^2+(x_2-1)^2+32x_1(1+x_2))(16x_1^2+(x_2-1)^2-8x_1(1+x_2))} \,,\\
     C^{\bs{x}}_{(2,2,2)}&=\dfrac{-6(1+64x_1^2+5x_2(2+x_2)-20x_1(1+3x_2))}{x_2^2(x_2-1)(256x_1^2+(x_2-1)^2-32x_1(1+x_2))(16x_1^2+(x_2-1)^2-8x_1(1+x_2))} \,.
 \end{align}
 \paragraph{Canonical basis for the periods.}
  We can now choose an alternative basis of periods by rotating out the semi-simple part through a gauge transformation $(\bs{U}_\text{ss}^\text{CY3})^{-1}$ which leads to $\bs{\hat{P}}^\text{CY3}=(\bs{U}_\text{ss}^\text{CY3})^{-1}\bs{P}^\text{CY3}$, where:
\begin{align}
\bs{ \hat{P}}^\text{CY3}=\begin{pmatrix}
\hat{\bs{\psi}}\\
 \partial_{\tau_1}\bs{\hat{\psi}}\\
 \partial_{\tau_2}\bs{\hat{\psi}}\\
 \bigl(\mathcal{Y}_3\partial_{\tau_1}-\mathcal{Y}_4\partial_{\tau_2}\bigr) \partial_{\tau_2}\bs{\hat{\psi}}\\
\bigl(\mathcal{Y}_1\partial_{\tau_1}-\mathcal{Y}_2\partial_{\tau_2}\bigr) \partial_{\tau_1}\bs{\hat{\psi}}\\
 \partial_{\tau_1}\bigl(\mathcal{Y}_1\partial_{\tau_1}-\mathcal{Y}_2\partial_{\tau_2}\bigr) \partial_{\tau_1}\bs{\hat{\psi}}\end{pmatrix} \,,
 \end{align}
 with $\bs{\hat{\psi}}$ the period vector normalized by $\psi_0$:
  \begin{equation}
\bs{\hat{\psi}}=\dfrac{1}{\psi_0}\bs{\psi}\,,
 \end{equation}
 and the moduli are defined by
 \begin{equation}
     \tau_1=\dfrac{\psi_1^{(1)}}{\psi_0}\textrm{~~~and~~~} \tau_2=\dfrac{\psi_1^{(2)}}{\psi_0}\,.
 \end{equation}
 The functions $\mathcal{Y}_i$ are rational functions of the Yukawa three-couplings and read
\begin{equation}\label{Ys}
\begin{split}
    \mathcal{Y}_1&=\dfrac{C_{(1,2,2)}^{\bs \tau}}{C_{(1,1,1)}^{\bs \tau}C_{(1,2,2)}^{\bs \tau}{-}(C_{(1,1,2)}^{\bs \tau})^2}\,,\quad \mathcal{Y}_2=\dfrac{C_{(1,1,2)}^{\bs \tau}}{C_{(1,1,1)}^{\bs \tau}C_{(1,2,2)}^{\bs \tau}{-}(C_{(1,1,2)}^{\bs \tau})^2}\,,\notag\\
    \mathcal{Y}_3&=\dfrac{C_{(1,2,2)}^{\bs \tau}}{(C_{(1,2,2)}^{\bs \tau})^2{-}C_{(1,1,2)}^{\bs \tau}C_{(2,2,2)}^{\bs \tau}}\,,\quad \mathcal{Y}_4=\dfrac{C_{(1,1,2)}^{\bs \tau}}{(C_{(1,2,2)}^{\bs \tau})^2{-}C_{(1,1,2)}^{\bs \tau}C_{(2,2,2)}^{\bs \tau}}\,,
\end{split}
\end{equation}
where the functions $C_{(i,j,k)}^{\bs{\tau}}$ are the Yukawa three-couplings in $\tau$,
\begin{equation}
    C_{(i,j,k)}^{\bs{\tau}}=\bs{\hat{\psi}}\bs{\Sigma}\partial_{\tau_i}\partial_{\tau_j}\partial_{\tau_k}\bs{\hat{\psi}}\,,
\end{equation}
and they are related to the Yukawa couplings of eq.~(\ref{quad_rels}) through
\begin{equation}
    C_{(i,j,k)}^{\bs{\tau}}=\dfrac{\mathfrak{J}_{a,i}\,\mathfrak{J}_{b,j}\,\mathfrak{J}_{c,k}}{\psi_0^2}C_{(a,b,c)}^{\bs{x}}\,,
\end{equation}
where repeated indices are summed over, and 
the Jacobian $\mathfrak{J}$ is defined in eq.~\eqref{eq:jacobian}.
In this new basis, the differential equation becomes nilpotent and reads
 \begin{equation}
\partial_{\tau_j}\bs{ \hat{P}}^\text{CY3}=\begin{pmatrix}0&\delta_{j,1}&\delta_{j,2}&0&0&0\\
0&0&0&C^{\bs\tau}_{(1,1,j)}&C^{\bs \tau}_{(1,2,j)}&0\\
0&0&0&C^{\bs \tau}_{(1,2,j)}&C^{\bs\tau}_{(2,2,j)}&0\\
0&0&0&0&0&\delta_{1,j}\\
0&0&0&0&0&\delta_{2,j}\\
0&0&0&0&0&0\end{pmatrix}\bs{ \hat{P}}^\text{CY3}\,.
\end{equation}
In addition, the cohomology intersection matrix $\bs{\hat{Z}}$ becomes constant:
\begin{equation}
    \bs{\hat{Z}}^\text{CY3}=
    \begin{pmatrix} 0 & \bs{0} & 1\\
    \bs{0} &-\bK_4 &\bs{0}\\
    1& \bs{0} &0\end{pmatrix} \,.
\end{equation}

\end{appendix}
\bibliographystyle{JHEP}
\bibliography{biblio.bib}

\end{document}